\documentclass[a4paper,10pt]{amsart}
\usepackage[foot]{amsaddr}
\usepackage{amssymb}
\usepackage[utf8]{inputenc}
\usepackage{tikz, fp, tikz-cd}
\usepackage[a4paper]{geometry}
\usepackage{mathtools}
\usepackage{multicol}
\usepackage{bbm}
\usepackage{array}
\usepackage{stmaryrd}
\usepackage{csquotes}
\usepackage{environ} %package pour les macros d'édition
\usepackage{listings}
\usepackage{subcaption}
\usepackage{tabularray}
\usepackage{shuffle}
\usepackage[english]{babel}
\usepackage{siunitx}
\usepackage{nth}
\usepackage{hyperref}
\usepackage{soul}
\usepackage{algorithm}
\usepackage{algorithmicx}
\usepackage{algpseudocode, algpseudocodex}
\usepackage{alltt}
\usepackage{comment}

\usetikzlibrary{patterns}
\usetikzlibrary{arrows}
\usetikzlibrary{decorations.pathreplacing}
\usetikzlibrary{automata} 

\title{Effective approach of the Tridendriform Schroeder tree algebra }
\author{P. Catoire}
\address{Institut Montpelliérain Alexander Grothendieck, Université de Montpellier - Place
Eugène Bataillon, 34095 Montpellier Cedex 5, France}
\email{catoire\_research@proton.me}
\author{J. Fromentin}
\address{Université du Littoral Côte d’Opale, UR 2597 LMPA, Laboratoire de Mathématiques Pures
et Appliquées Joseph Liouville, F-62100 Calais, France}
\email{jean.fromentin@univ-littoral.fr}

\DeclareMathOperator{\V}{V}
\DeclareMathOperator{\Vint}{V_{\textrm{int}}}
\DeclareMathOperator{\E}{E}
\DeclareMathOperator{\Eint}{E_{\textrm{int}}}

\DeclareMathOperator{\Prim}{Prim}
\DeclareMathOperator{\ima}{Im}

\DeclareMathOperator{\Id}{Id}
\DeclareMathOperator{\IC}{InitChain}

\DeclareMathOperator{\Coass}{Coass}
\DeclareMathOperator{\Codend}{Codend}
\DeclareMathOperator{\Code}{Code}

\DeclareMathOperator{\Grid}{Grid}

\DeclareMathOperator{\bat}{Sh}
\DeclareMathOperator{\batc}{QSh}

\DeclareMathOperator{\LC}{LCD}
\DeclareMathOperator{\RC}{RCD}
\DeclareMathOperator{\RTL}{RTL}
\DeclareMathOperator{\LTL}{LTL}
\DeclareMathOperator{\GRID}{Gr}

\DeclareMathOperator{\TSym}{TSym}
\DeclareMathOperator{\K}{\mathbb{K}}

\DeclareMathOperator{\PW}{PW}
\DeclareMathOperator{\LPPW}{LPPW}
\DeclareMathOperator{\ew}{\varepsilon} %the empty word

\newcommand\Sch{\mathsf{Sch}}
\newcommand\SchF{\mathsf{FSch}}
\newcommand\Schl{\mathsf{Sch}_\ell}
\newcommand\nn{n}
\newcommand\NN{\mathbb{N}}
\newcommand\pleft{\prec}
\newcommand\pmid{\cdot}
\newcommand\pright{\succ}

\newcommand{\IEM}[2]{\llbracket #1,#2 \rrbracket}

\NewDocumentCommand{\Sage}{}{Sagemath}
\newcommand{\Schw}{Schroeder} %le mot Schroeder pour éviter les fautes de frappes

\NewDocumentCommand{\LCD}{m}{\ensuremath{\LC\left(#1\right)}}
\NewDocumentCommand{\RCD}{m}{\ensuremath{\RC\left(#1\right)}}

\newcounter{thmcount}
\newtheorem{thm}[thmcount]{Theorem}
\newtheorem{prop}[thmcount]{Proposition}
\newtheorem{Cor}[thmcount]{Corollary}
\newtheorem{Lemma}[thmcount]{Lemma}

\theoremstyle{definition}
\newtheorem{exam}[thmcount]{Example}
\newtheorem{defi}[thmcount]{Definition}
\newtheorem{Not}[thmcount]{Notation}
\newtheorem{Rq}[thmcount]{Remark}

\newcommand{\Y}{\raisebox{-0.25\height}{\begin{tikzpicture}[line cap=round,line join=round,>=triangle 45,x=0.2cm,y=0.2cm]
			\draw [line width=.5pt] (0.,0.)-- (0.,1.);
			\draw [line width=.5pt] (0.,1.)-- (-1.,2.);
			\draw [line width=.5pt] (0.,1.)-- (1.,2.);
			\filldraw (0,1) circle(0.15);
			\draw[color=black, fill=white] (-1,2) circle (0.15);
			\draw[color=black, fill=white] (1,2) circle (0.15);
\end{tikzpicture}}}

\newcommand{\YY}{\raisebox{-0.3\height}{\begin{tikzpicture}[line cap=round,line join=round,>=triangle 45,x=0.2cm,y=0.2cm]
			\draw (0,0)--(0,1);
			\draw (0,1)--(-2,3);
			\draw (-1.25,2.25)--(-0.5,3);
			\draw (1.25,2.25)--(0.5,3);
			\draw (0,1)--(2,3);
			\draw[color=black, fill=white] (-2,3) circle (0.15);
			\draw[color=black, fill=white] (2,3) circle (0.15);
			\draw[color=black, fill=white] (-0.5,3) circle (0.15);
			\draw[color=black, fill=white] (0.5,3) circle (0.15);
			\filldraw (0,1) circle(0.15);
			\filldraw (1.25,2.25) circle(0.15);
			\filldraw (-1.25,2.25) circle(0.15);
\end{tikzpicture}}}

\newcommand{\peignedroit}[3]{\raisebox{-0.5\height}{\begin{tikzpicture}[line cap=round,line join=round,>=triangle 45,x=0.25cm,y=0.25cm]
			\draw (0,0)--(0,1);
			\draw (0,1)--(-1,2) node[left]{${#1}$};
			\draw (0,1)--(6,7);
			\draw (2,3)--(1,4) node[left]{${#2}$};
			\draw (3.,4.5) node[rotate=45]{$\cdots$};
			\draw (5,6)--(4,7) node[left]{${#3}$};
			\filldraw (0,1) circle (0.15);
			\filldraw (2,3) circle (0.15);
			\filldraw (5,6) circle (0.15);
			\draw[color=black, fill=white](6,7) circle (0.15);
\end{tikzpicture} }}

\newcommand{\peignegauche}[3]{\raisebox{-0.5\height}{\begin{tikzpicture}[line cap=round,line join=round,>=triangle 45,x=0.25cm,y=0.25cm]
			\draw (0,0)--(0,1);
			\draw (0,1)--(1,2) node[right]{${#1}$};
			\draw (0,1)--(-6,7);
			\draw (-2,3)--(-1,4) node[right]{${#2}$};
			\draw (-3.,4.5) node[rotate=135]{$\cdots$};
			\draw (-5,6)--(-4,7) node[right]{${#3}$};
			\filldraw (0,1) circle (0.15);
			\filldraw (-2,3) circle (0.15);
			\filldraw (-5,6) circle (0.15);
			\draw[color=black, fill=white](-6,7) circle (0.15);
\end{tikzpicture}}}

\newcommand{\unitree}{\begin{tikzpicture}[line cap=round,line join=round,>=triangle 45,x=0.2cm,y=0.2cm]
		\filldraw[color=black,fill=white] (0,0) circle (0.15);		
\end{tikzpicture}}

\newcommand{\balais}{\raisebox{-0.3\height}{\begin{tikzpicture}[line cap=round,line join=round,>=triangle 45,x=0.2cm,y=0.2cm]
			\draw (0,0)--(0,2);
			\draw (0,1)-- (-1,2);
			\draw (0,1)-- (1,2);
			\filldraw[color=black, fill=white] (-1,2) circle(0.15);
			\filldraw[color=black, fill=white] (0,2) circle(0.15);
			\filldraw[color=black, fill=white] (1,2) circle(0.15);
			\filldraw[color=black] (0,1) circle(0.15);
\end{tikzpicture}}}
\newcommand{\balaisg}{\raisebox{-0.3\height}{\begin{tikzpicture}[line cap=round,line join=round,>=triangle 45,x=0.2cm,y=0.2cm]
			\draw (0,0)-- (0,1);
			\draw (0,1)-- (-1,2);
			\draw (-0.5,1.5)-- (0,2);
			\draw (0,1)-- (1,2);
			\filldraw[color=black, fill=white] (-1,2) circle(0.15);
			\filldraw[color=black, fill=white] (0,2) circle(0.15);
			\filldraw[color=black, fill=white] (1,2) circle(0.15);
			\filldraw[color=black] (0,1) circle(0.15);
			\filldraw[color=black] (-0.5,1.5) circle(0.15);
\end{tikzpicture}}}
\newcommand{\balaisd}{\raisebox{-0.3\height}{\begin{tikzpicture}[line cap=round,line join=round,>=triangle 45,x=0.2cm,y=0.2cm]
			\draw (0,0)-- (0,1);
			\draw (0,1)-- (-1,2);
			\draw (0.5,1.5)-- (0,2);
			\draw (0,1)-- (1,2);
			\filldraw[color=black, fill=white] (-1,2) circle(0.15);
			\filldraw[color=black, fill=white] (0,2) circle(0.15);
			\filldraw[color=black, fill=white] (1,2) circle(0.15);
			\filldraw[color=black] (0,1) circle(0.15);
			\filldraw[color=black] (0.5,1.5) circle(0.15);
\end{tikzpicture}}}

\NewEnviron{Jean}{{\color{red} \BODY }}
\NewEnviron{Pierre}{{\color{blue} \BODY }}
\NewEnviron{Loic}{{\color{olive} \BODY }}

\begin{document}

\begin{abstract}
	We introduce a primitive computation problem in the free tridendriform algebra generated by one element which is a Hopf algebra based on \Schw{} trees. We know a non-combinatorial complex iterative way to generate all of them. To understand it clearer, we want to implement this method on a computer. For this purpose, we need to introduce initial chains to implement Schroeder trees and the corresponding algebraic operations to compute the primitive elements and check numerically that they are. In this paper, we detail how we made the problem mathematically understandable for a computer and how we did implement it. 
\end{abstract}

\maketitle

\section{Introduction}

\subsection{Aim and structure of the paper}

The objective of this work is to implement the operations of a Hopf algebra involving \Schw{} trees. The vector space of \Schw{} trees introduced in definition~\ref{defi:Sch_tree} can be endowed with three linear maps $\prec, \succ$ and $\cdot$ making it a free tridendriform algebra structure~\cite{LodayRonco_98,Catoire_2023, Manchon_2020} where the products can be described with quasi-shuffles (theorem~\ref{thm:produit}) . 
Tridendriform algebras arise in different topics in mathematics like free probability~\cite{Celestino_2022, Ebrahimi-Fard_Patras_2018}, the study of Rota-Baxter algebras~\cite{Manchon_RB_2020, Gao_2019} or in other topics like the study of multiple zeta values~\cite{Catoire_2025, Ho00}.
Some computation tools already exist for some algebras like Loday-Ronco algebra~\cite{Sage_dend} on \Sage{}~\cite{Sage_doc,Sage_book} but not for the free tridendriform algebra. 
Hence, implementing the free tridendriform algebra~\cite{Catoire_2023} (over one generator) allows us to implement any tridendriform algebras coding a quotient map by a tridendriform ideal from \Schw{} trees to the desired objects.

This algebra turns out to be a Hopf algebra whose coproduct is given in theorem~\ref{thm:coproduit}. Computing the space of primitive elements of this algebra, we describe at least the largest cocommutative algebra inside it thanks to the Cartier-Quillen-Milnor-Moore theorem~\cite{CartierPatras_2021,May_69,MilnorMoore_65}. An inductive algorithm is given in a previous work~\cite{Catoire_2023} (shown in figure~\ref{fig:algo_idea}) to compute those elements using only $\prec, \succ$ and $\cdot$ operations. 

We want to compute all primitives up to a certain degree in order to eventually spot some combinatorial properties. However, the number of primitive elements is given by big \Schw{} numbers~\cite[referenced A00638]{OEIS} which are increasing quickly. So, we had to implement this algebraic structure in a computer.
To solve this problem, we used a particular representation of \Schw{} trees as \emph{initial chains} of a poset $(\mathcal{P}(\IEM{1}{n},\subseteq)$ (definition~\ref{def:init_chain}) or specific sequences over the alphabet $\{0,1\}$ (definition~\ref{defi:lifetime}). Those representations of trees are called \emph{codes}. Thanks to this point of view, we are able to describe the action of quasi-shuffles on a pair of codes in definition~\ref{defi:quasishuffle_code} such that it describes the initial action of these elements onto a pair of \Schw{} trees (theorem~\ref{Cor:iso_codes_trees}). We present some algorithms in pseudo-code that were used to implement this algebraic structure.

The paper is divided in the following way:
\begin{itemize}
	\item Section~\ref{sec:the_problem} describes and gives the minimum piece of information needed to understand the general mathematical setting in which we are working. We introduce the main objects of this work, \Schw{} trees in definition~\ref{defi:Sch_tree} and quasi-shuffles in definition~\ref{def:quasi_shuffles}. Then, we describe the Hopf algebra structure on $(\K\Sch,\prec, \succ, \cdot)$ and we introduce the algorithm in figure~\ref{fig:algo_idea} which computes its space of primitive elements.
	\item Section~\ref{sec:pov} introduces the concept of levelled \Schw{} trees (definition~\ref{defi:levelled_trees}), initial chains and the map associating to any initial chain a $0$'s  and $1$'s sequence (definition~\ref{def:init_chain}). Then, we choose a way to map any \Schw{} tree to a levelled \Schw{} tree (proposition~\ref{prop:trees_to_levels}) which gives a way to get an initial chain of a part poset from any \Schw{} tree. A sequence obtained by this transformation will be called a \emph{tree code}. Finally, we present a tridendriform structure on tree codes (definition~\ref{defi:quasishuffle_code}) that turns out to be isomorphic to the one on \Schw{} trees (corollary~\ref{Cor:iso_codes_trees}). We also give details how one finds admissible cuts from a tree code, see section~\ref{sec:cuts}.
	\item Section~\ref{sec:algorithms} gives the pseudo-codes of algorithms implemented in {\ttfamily SageMaths}~\cite{github} enabling us to compute the action of a quasi-shuffle onto a pair of trees, to cut a tree given a pruning and to enumerate prunings.
	\item The last section~\ref{sec:coding} explains the technical details we used to perform our computations and the results it produces. In particular, the code in {\ttfamily C++} is available online~\cite{github}.
\end{itemize}

\subsection{Table of notations}

We give a list of most used notations in this work:
\begin{itemize} 
	\item for any $n\in\NN^*, \IEM{1}{n}$ is the set $\{1,\dots, n\}$;
	\item for any set $X$, we denote by $\K X$ the vector space over $\K$ whose basis is $X$. We denote $X^{\star}$ the set of finite words over the alphabet $X$. 
	We also denote by $T(\K X)$ the $\K$-vector space whose basis is $X^{\star}$. For any $n\in\NN,$ we denote $T^{\leq n}(\K X)$ the vector space of words with less than $n$ letters;
	\item for any graded set $X=\bigoplus_{n=0}^{+\infty} X_n,$ for any $n\in\NN$, we define $T(\K X)(n)$ the subspace whose basis is the set $X^{\star}$ where the sum of the degree's of its letters is $n$; 
	\item $\Sch(\nn)$ is the set of Schroeder tree with $\nn$ leaves;
	\item $\SchF=\Sch^{\star}$ is the set of \emph{forests} of Schroeder trees;
	\item $(\K\Sch, \pleft, \pmid, \pright)$ the graded tridendriform algebra on Schroeder tree;
%	\item $\vee$ is the grafting operator of many trees onto a common root;
	\item given a sequence $x=(x_1,\dots, x_n)$ where each  of its element appears at most once and $i\in\IEM{1}{n},$ we denote by $x\setminus (x_i)$ the sequence $x$ without the term $x_i$;
	\item given a rooted tree $t$, we denote by $\V(t)$ and $\E(t)$ its set of vertices and edges. We denote by $\Vint(t)$ and $\Eint(t)$ its set of \emph{internal} vertices and edges
\end{itemize}

\subsection{Acknowledgments}

The authors would like to thank LMPA of ULCO for welcoming the first author during this collaboration.

\section{About the original problem} \label{sec:the_problem}

\subsection{\Schw{} trees and comb representations}

Let us remind the following definition:
\begin{defi}\label{defi:Sch_tree}
	A \emph{\Schw{} tree} is a planar rooted tree such that any vertices which is not a leaf has at least $2$ children.
	This set can be graded with the number of leaves minus $1$. For any $n\in\NN\setminus\{0\}$, we define $\Sch(\nn)$ the set of \Schw{} trees with $n+1$ leaves.
	We define $\Sch(0)=\left\{\unitree\right\}$, where $\unitree$ is the tree whose root is also a leaf, and put: 
	\[
	\Sch\coloneqq \bigcup_{n=0}^{+\infty} \Sch(\nn) \text{ and } \overline{\Sch}=\bigcup_{n=1}^{+\infty} \Sch(\nn).
	\]
	Given any $t\in\Sch,$ we will denote $\Vint(t)$ the set of vertices which are not leaves and $\E(t)$ its set of edges. We call them \emph{internal vertices} of $t$. In the sequel, they are drawn as black disks.
\end{defi}
\begin{exam}\label{exam:Sch3} %Attention, c'est 3 ici, \Sch(n) c'est les arbres à n+1 feuilles.
	The elements of $\Sch(3)$ are
	
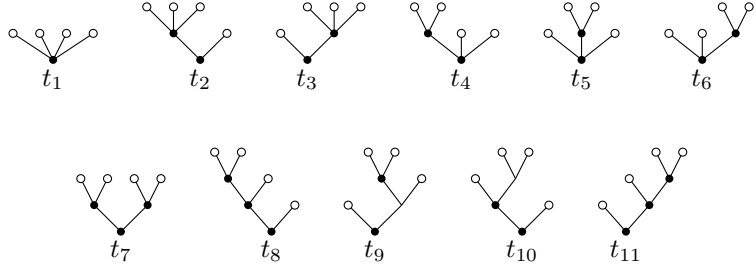
\vspace{0.5em}
	\begin{center}
		\begin{tikzpicture}[x=1em,y=1em]
			\draw(-1.5,1) -- (0,0);
			\draw(-0.5,1) -- (0,0);
			\draw(0.5,1) -- (0,0);
			\draw(1.5,1) -- (0,0);
			\draw[color=black, fill=white](-1.5,1) circle(0.15);
			\draw[color=black, fill=white](-0.5,1) circle(0.15);
			\draw[color=black, fill=white](0.5,1) circle(0.15);
			\draw[color=black, fill=white](1.5,1) circle(0.15);
			\fill(0,0) circle(0.15) node[below]{$t_1$};
		\end{tikzpicture}
		\hspace{1em}
		\begin{tikzpicture}[x=1em,y=1em]
			\draw(-1,2) -- (0,1);
			\draw(0,2) -- (0,1);
			\draw(1,2) -- (0,1);
			\draw(2,1) -- (1,0);
			\draw(0,1) -- (1,0);
			\draw[color=black, fill=white](-1,2) circle(0.15);
			\draw[color=black, fill=white](0,2) circle(0.15);
			\draw[color=black, fill=white](1,2) circle(0.15);
			\draw[color=black, fill=white](2,1) circle(0.15);
			\fill(1,0) circle(0.15) node[below]{$t_2$};
			\fill (0,1) circle(0.15);
		\end{tikzpicture}
		\hspace{1em}
		\begin{tikzpicture}[x=1em,y=1em]
			\draw(-1,1) -- (0,0);
			\draw(0,2) -- (1,1);
			\draw(1,2) -- (1,1);
			\draw(2,2) -- (1,1);
			\draw(1,1) -- (0,0);
			\draw[color=black, fill=white](-1,1) circle(0.15);
			\draw[color=black, fill=white](0,2) circle(0.15);
			\draw[color=black, fill=white](1,2) circle(0.15);
			\draw[color=black, fill=white](2,2) circle(0.15);
			\fill(0,0) circle(0.15) node[below]{$t_3$};
			\fill (1,1) circle(0.15);
		\end{tikzpicture}
		\hspace{1em}
		\begin{tikzpicture}[x=1em,y=1em]
			\draw(1,1) -- (-0.25,0);
			\draw(-0.25,1) -- (-0.25,0);
			\draw(-1,2) -- (-1.5,1);
			\draw(-2,2) -- (-1.5,1);
			\draw(-1.5,1) -- (-0.25,0);
			\draw[color=black, fill=white](1,1) circle(0.15);
			\draw[color=black, fill=white](-0.25,1) circle(0.15);
			\draw[color=black, fill=white](-1,2) circle(0.15);
			\draw[color=black, fill=white](-2,2) circle(0.15);
			\fill(-0.25,0) circle(0.15) node[below]{$t_4$};
			\fill (-1.5,1) circle(0.15);
		\end{tikzpicture}
		\hspace{1em}
		\begin{tikzpicture}[x=1em,y=1em]
			\draw(-0.25,2) -- (0.25,1);
			\draw(0.75,2) -- (0.25,1);
			\draw(-1,1) -- (0.25,0);
			\draw(0.25,1) -- (0.25,0);
			\draw(1.5,1) -- (0.25,0);
			\draw[color=black, fill=white](-1,1) circle(0.15);
			\draw[color=black, fill=white](-0.25,2) circle(0.15);
			\draw[color=black, fill=white](0.75,2) circle(0.15);
			\draw[color=black, fill=white](1.5,1) circle(0.15);
			\fill(0.25,0) circle(0.15) node[below]{$t_5$};
			\fill (0.25,1) circle(0.15);
		\end{tikzpicture}
		\hspace{1em}
		\begin{tikzpicture}[x=1em,y=1em]
			\draw(-1,1) -- (0.25,0);
			\draw(0.25,1) -- (0.25,0);
			\draw(1,2) -- (1.5,1);
			\draw(2,2) -- (1.5,1);
			\draw(1.5,1) -- (0.25,0);
			\draw[color=black, fill=white](-1,1) circle(0.15);
			\draw[color=black, fill=white](0.25,1) circle(0.15);
			\draw[color=black, fill=white](1,2) circle(0.15);
			\draw[color=black, fill=white](2,2) circle(0.15);
			\fill(0.25,0) circle(0.15) node[below]{$t_6$};
			\fill (1.5,1) circle(0.15);
		\end{tikzpicture}
		\vspace{1em}
	\end{center}
	\begin{center}
		\begin{tikzpicture}[x=1em,y=1em]
			\draw(0,2) -- (0.5,1);
			\draw(1,2) -- (0.5,1);
			\draw(2,2) -- (2.5,1);
			\draw(3,2) -- (2.5,1);
			\draw(0.5,1) -- (1.5,0);
			\draw(2.5,1) -- (1.5,0);
			\draw[color=black, fill=white](0,2) circle(0.15);
			\draw[color=black, fill=white](1,2) circle(0.15);
			\draw[color=black, fill=white](2,2) circle(0.15);
			\draw[color=black, fill=white](3,2) circle(0.15);
			\fill(1.5,0) circle(0.15) node[below]{$t_7$};
			\fill (0.5,1) circle(0.15);
			\fill (2.5,1) circle(0.15);
		\end{tikzpicture}
		\hspace{1em}
		\begin{tikzpicture}[x=1em,y=1em]
			\draw(-1,2) -- (-0.5,1);
			\draw(0,2) -- (-0.5,1);
			\draw(1,1) -- (0.25,0);
			\draw(-0.5,1) -- (0.25,0);
			\draw(0.25,0) -- (1.125, -1);
			\draw(2,0) -- (1.125, -1);
			\draw[color=black, fill=white](-1,2) circle(0.15);
			\draw[color=black, fill=white](0,2) circle(0.15);
			\draw[color=black, fill=white](1,1) circle(0.15);
			\draw[color=black, fill=white](2,0) circle(0.15);
			\fill(1.125,-1) circle(0.15) node[below]{$t_8$};
			\fill (0.25,0) circle(0.15);
			\fill (-0.5,1) circle(0.15);
		\end{tikzpicture}
		\hspace{1em}
		\begin{tikzpicture}[x=1em,y=1em]
			\draw(-1,2) -- (-0.5,1);
			\draw(0,2) -- (-0.5,1);
			\draw(1,1) -- (0.25,0);
			\draw(-0.5,1) -- (0.25,0);
			\draw(0.25,0) -- (-0.75,-1);
			\draw(-1.75,0) -- (-0.75,-1);
			\draw[color=black, fill=white](-1,2) circle(0.15);
			\draw[color=black, fill=white](0,2) circle(0.15);
			\draw[color=black, fill=white](1,1) circle(0.15);
			\draw[color=black, fill=white](-1.75,0) circle(0.15);
			\fill(-0.75,-1) circle(0.15) node[below]{$t_9$};
			\fill (0.25,0) circle (0.15);
			\fill (-0.5,1) circle(0.15);
		\end{tikzpicture}
		\hspace{1em}
		\begin{tikzpicture}[x=1em,y=1em]
			\draw(1,2) -- (0.5,1);
			\draw(0,2) -- (0.5,1);
			\draw(-1,1) -- (-0.25,0);
			\draw(0.5,1) -- (-0.25,0);
			\draw(-0.25,0) -- (0.75,-1);
			\draw(1.75,0) -- (0.75,-1);
			\draw[color=black, fill=white](1,2) circle(0.15);
			\draw[color=black, fill=white](0,2) circle(0.15);
			\draw[color=black, fill=white](-1,1) circle(0.15);
			\draw[color=black, fill=white](1.75,0) circle(0.15);
			\fill(0.75,-1) circle(0.15) node[below]{$t_{10}$};
			\fill (-0.25,0) circle(0.15);
			\fill (0.5,1) circle (0.15);
		\end{tikzpicture}
		\hspace{1em}
		\begin{tikzpicture}[x=1em,y=1em]
			\draw(1,2) -- (0.5,1);
			\draw(0,2) -- (0.5,1);
			\draw(-1,1) -- (-0.25,0);
			\draw(0.5,1) -- (-0.25,0);
			\draw(-0.25,0) -- (-1.125, -1);
			\draw(-2,0) -- (-1.125, -1);
			\draw[color=black, fill=white](1,2) circle(0.15);
			\draw[color=black, fill=white](0,2) circle(0.15);
			\draw[color=black, fill=white](-1,1) circle(0.15);
			\draw[color=black, fill=white](-2,0) circle(0.15);
			\fill(-1.125,-1) circle(0.15) node[below]{$t_{11}$};
			\fill (-0.25,0) circle(0.15);
			\fill (0.5,1) circle(0.15);
		\end{tikzpicture}.
	\end{center}
\end{exam}

Let $t$ be a tree. It can always be seen repsectively as a  right comb or a left comb:
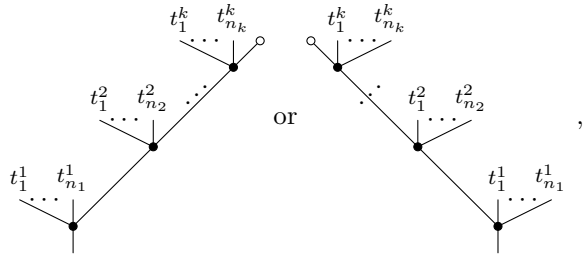
\begin{center}
	\raisebox{-0.5\height}{\begin{tikzpicture}[line cap=round,line join=round,x=1em,y=1 em]
			% Etage 1
			\draw (0,0)--(0,1);
			\filldraw (0,1) circle (0.15);
			%Etage 2
			\draw (0,1)--(-2,2)
			 node[above]{\footnotesize{$t^1_1$}};
			\draw (0,1)--(0,2) node[above]{\footnotesize{$t^1_{n_1}$}};
			\draw (-1,2) node{$\cdots$};
			% Etage 3
			\draw (0,1)--(7,8);
			\filldraw[color=white] (7,8) circle (0.15);
			\draw (3,4)--(1,5) node[above]{\footnotesize{$t^2_1$}};
			\filldraw (3,4) circle (0.15);
			\draw (3,4)--(3,5) node[above]{\footnotesize{$t^2_{n_2}$}};
			\draw (2,5) node{$\cdots$};
			%Etage...
			\draw (4.7,6) node[rotate=45]{$\cdots$};
			%Etage k
			\draw[color=black, fill=white](7,8) circle(0.15);
			\filldraw (6,7) circle (0.15);
			\draw (6,7)--(4,8) node[above]{\footnotesize{$t^k_1$}};
			\draw (6,7)--(6,8) node[above]{\footnotesize{$t^k_{n_k}$}};
			\draw (5,8) node {$\cdots$};
	\end{tikzpicture}} or
	\raisebox{-0.5\height}{\begin{tikzpicture}[line cap=round,line join=round,x=1em,y=1 em]
			% Etage 1
			\draw (0,0)--(0,1);
			\filldraw (0,1) circle (0.15);
			%Etage 2
			\draw (0,1)--(2,2)
			 node[above]{\footnotesize{$t^1_{n_1}$}};
			\draw (0,1)--(0,2) node[above]{\footnotesize{$t^1_{1}$}};
			\draw (1,2) node{$\cdots$};
			% Etage 3
			\draw (0,1)--(-7,8);
			\filldraw[color=white] (-7,8) circle (0.15);
			\draw (-3,4)--(-1,5) node[above]{\footnotesize{$t^2_{n_2}$}};
			\filldraw (-3,4) circle (0.15);
			\draw (-3,4)--(-3,5) node[above]{\footnotesize{$t^2_{1}$}};
			\draw (-2,5) node{$\cdots$};
			%Etage...
			\draw (-4.7,6) node[rotate=45]{$\cdots$};
			%Etage k
			\draw[color=black, fill=white](-7,8) circle(0.15);
			\filldraw (-6,7) circle (0.15);
			\draw (-6,7)--(-4,8) node[above]{\footnotesize{$t^k_{n_k}$}};
			\draw (-6,7)--(-6,8) node[above]{\footnotesize{$t^k_{1}$}};
			\draw (-5,8) node {$\cdots$};
	\end{tikzpicture}},
\end{center}
where $k$ is the number of nodes on the right-most branch of $t$ and for $i\in\IEM{1}{k}, n_i+1$ is the number of sons of the $i$-th node of this branch.

\begin{Not}
	Let $F=t_1\dots t_n$ be a forest composed of $n$ trees. For simplicity, we write:
	\[
	\raisebox{-0.3\height}
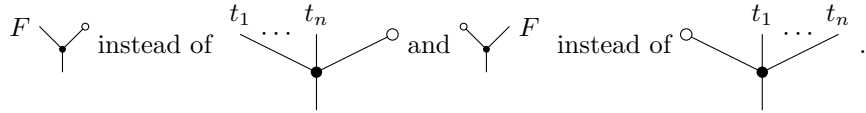
{\begin{tikzpicture}[line cap=round,line join=round,x=0.3cm,y=0.3cm]
			% Etage 1
			\draw (0,0)--(0,1);
			%Etage 2
			\draw (0,1)--(-1,2) node[left]{$F$};
			\draw (0,1)--(1,2);
			\draw[color=black, fill=white] (1,2) circle (0.15);
			\fill (0,1) circle(0.15);
	\end{tikzpicture}}  \text{ instead of }
		\raisebox{-0.5\height}{\begin{tikzpicture}[line cap=round,line join=round,x=0.5cm,y=0.5cm]
				% Etage 1
				\draw (0,0)--(0,1);
				%Etage 2
				\draw (0,1)--(-2,2) node[left,above]{$t_1$};
				\draw (0,1)--(0,2) node[right,above]{$t_{n}$};
				\draw (-1,2) node{$\cdots$};
				\draw (0,1)--(2,2);
				\draw[color=black, fill=white] (2,2) circle (0.15);
				\fill (0,1) circle(0.15);
		\end{tikzpicture}}
		 \text{ and }
		\raisebox{-0.3\height}{\begin{tikzpicture}[line cap=round,line join=round,x=0.3cm,y=0.3cm]
			% Etage 1
			\draw (0,0)--(0,1);
			%Etage 2
			\draw (0,1)--(1,2) node[right]{$F$};
			\draw (0,1)--(-1,2);
			\draw[color=black, fill=white] (-1,2) circle (0.15);
			\fill (0,1) circle(0.15);
	\end{tikzpicture}}  \text{ instead of }
		\raisebox{-0.5\height}{\begin{tikzpicture}[line cap=round,line join=round,x=0.5cm,y=0.5cm]
				% Etage 1
				\draw (0,0)--(0,1);
				%Etage 2
				\draw (0,1)--(2,2) node[left,above]{$t_{n}$};
				\draw (0,1)--(0,2) node[right,above]{$t_{1}$};
				\draw (1,2) node{$\cdots$};
				\draw (0,1)--(-2,2);
				\draw[color=black, fill=white] (-2,2) circle (0.15);
				\fill (0,1) circle(0.15);
		\end{tikzpicture}}.
	\]
\end{Not}

\begin{defi}With these notations, we notice that any tree $t$ can be seen respectively as:
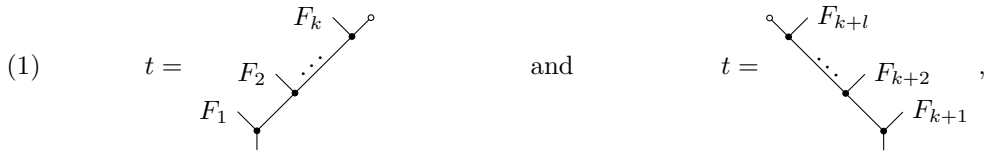
\begin{align}\label{eq:tree_combs}
	t&=\peignedroit{F_1}{F_2}{F_k}
	& \text{ and } &&  t=\peignegauche{F_{k+1}}{F_{k+2}}{F_{k+l}},
\end{align}
where $F_1,\dots,F_k$ and $F_{k+1},\dots, F_{k+l}$ are non-empty forests. The left representation is called the \emph{left comb} representation of $t$. The right one is called the \emph{right comb} reprensentation.
\end{defi}

\subsection{Our objective}

We describe the structure of bialgebra onto the vector space $\K\Sch$ and we explain how to compute its primitive elements thanks to previous results~\cite{Catoire_2023,euleridem}.

\begin{defi}\label{def:tridend}
	Let $\K$ be a field and $A$ be $\K$-vector space. We say $(A,\prec,\cdot,\succ)$ is a \emph{tridendriform algebra} if $\prec,\cdot$ and $\succ$ are three linear maps from $A\otimes A$ to $A$ such that for any $a,b,c\in A$:
	\begin{align}
		(a\prec b)\prec c&=a\prec(b*c), \label{eq:tri1}\\
		(a\succ b)\prec c&=a\succ(b\prec c), \label{eq:tri2} \\
		(a* b)\succ c&=a\succ(b\succ c),  \label{eq:tri3} \\
		(a\succ b)\cdot c&=a\succ(b\cdot c), \label{eq:tri4} \\
		(a\prec b)\cdot c&=a\cdot(b\succ c), \label{eq:tri5} \\
		(a\cdot b)\prec c&=a\cdot(b\prec c), \label{eq:tri6} \\
		(a\cdot b)\cdot c&=a\cdot (b\cdot c), \label{eq:tri7}
	\end{align}
	where $*$ is the \emph{associative product} of the tridendriform algebra defined for any $a,b\in A$ by $a*b\coloneqq a\prec b + a\cdot b +a\succ b$. 
	We call the products $\succ,\prec,\cdot$ respectively \emph{right}, \emph{left} and \emph{middle}.
\end{defi} 

The free tridendriform algebra is built on the vector space $\K\overline{\Sch}.$ We then add a unit $\unitree$ to the product $*$ by hand to get an algebra over $\K\overline{\Sch}\oplus \K\cdot \unitree=\K\Sch$ with product $*$. We will denote it by $(\K\Sch{},\prec,\cdot,\succ)$ this tridendriform algebra with an additionnal unit. We have a combinatorial interpretation of the product $*$. For this let us introduce:
\begin{defi}[Quasi-shuffle]\label{def:quasi_shuffles}
	Let $k,l\in\NN\setminus \{0\}$. A \emph{$(k,l)$-quasi-shuffle} is a surjective map $\sigma:\IEM{1}{k+l}\twoheadrightarrow\IEM{1}{n}$ for some positive integer $n$
	such that:
	\[
	\sigma(1)<\cdots<\sigma(k) \text{  and  } \sigma(k+1)<\cdots<\sigma(k+l).
	\]
	We will denote $\batc(k,l)$ the set of all $(k,l)$-quasi-shuffles.
\end{defi}
One can easily prove that~\cite[lemma~27]{Palacios_06}: 
\begin{Lemma}\label{lem:Indution_batc}
	Let $k,l\geq 2$. Then, the following sets are in bijections:
	\begin{align*}
		\{\sigma\in\batc(k,l) \,|\, \sigma^{-1}(\{1\})=\{1\}\} &\simeq \batc(k-1,l), \\
		\{\sigma\in\batc(k,l) \,|\, \sigma^{-1}(\{1\})=\{1,k+1\}\} &\simeq \batc(k-1,l-1), \\
		\{\sigma\in\batc(k,l) \,|\, \sigma^{-1}(\{1\})=\{k+1\}\} &\simeq \batc(k,l-1).
	\end{align*}
	For other cases, we have:
	\[ % font issue for the and here
	\batc(1,0)=\left\{\Id_{\{1\}}\right\}=\batc(0,1) \text{ \text{\rm and} } \batc(1,1)=\left\{\Id_{\{1,2\}}, (2,1)\right\}.
	\]
\end{Lemma}

Let $t,s$ two trees different from $\unitree$. We consider $t$ as a right comb and $s$ as a left comb:
\begin{align}\label{eq:comb_ts}
	t&=\peignedroit{F_1}{F_2}{F_k}
	& \text{ and } &&  s=\peignegauche{F_{k+1}}{F_{k+2}}{F_{k+l}},
\end{align}
where for all $i\in\IEM{1}{k+l}, F_i$ is a non-empty forest. Here, $k$ represents the number of nodes on the right-most branch of $t$ and $l$ is the number of nodes on the left-most branch of $s$.

\begin{defi}\label{def:quasiaction}
	Let $t$ and $s$ be two trees which respective right comb representation and left comb representation are given in equation~\eqref{eq:comb_ts}. Using the same notations, let $\sigma$ be a $(k,l)$-quasi-shuffle which has for image $\IEM{1}{n}$ with $n$ a non-negative integer. We denote $\sigma(t,s)$ the tree obtained this way:
	\begin{enumerate}
		\item  We first consider the ladder with $n$ nodes:
		\begin{center}
			\raisebox{-0.5\height}
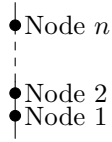
{\begin{tikzpicture}[line cap=round,line join=round,>=triangle 45,x=0.3cm,y=0.3cm]
					\begin{scope}{shift={(-3,0)}}
						\draw (0,0)--(0,2.5);
						\draw[dashed] (0,2.5)--(0,5);
						\draw (0,5)--(0,6);
						\filldraw [black] (0,1) circle (2pt) node[anchor=west]{Node $1$};
						\filldraw [black] (0,2) circle (2pt) node[anchor=west]{Node $2$};
						\filldraw [black] (0,5) circle (2pt) node[anchor=west]{Node $n$};
					\end{scope}
			\end{tikzpicture}}.
		\end{center}
		\item For all $i\in\IEM{1}{k},$ we graft $F_i$ as the \emph{left} son at the node $\sigma(i)$.
		\item For all $i\in\IEM{k+1}{k+l},$ we graft $F_i$ as \emph{right} son at the node $\sigma(i)$.
	\end{enumerate}
\end{defi}

\begin{exam}\label{exam:quasiaction}
	Consider the $(2,2)$-quasi-shuffle $\sigma=(1,3,2,3)$.
	
	 Let us take 
	$t=$	\raisebox{-0.3\height}
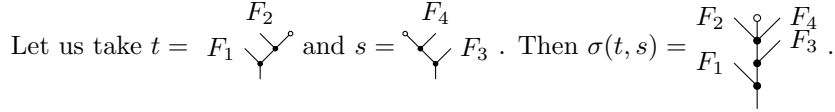
{\begin{tikzpicture}[line cap=round,line join=round,>=triangle 45,x=0.2cm,y=0.2cm]
			% Etage 1
			\draw (0,0)--(0,1);
			%Etage 2
			\draw (0,1)--(-1,2) node[left]{$F_1$};
			% Etage 3
			\draw (0,1)--(2,3);
			\draw (1,2)--(0,3) node[left,above]{$F_2$};
			%marquages des sommets
			\filldraw (0,1) circle (0.15);
			\filldraw (1,2) circle (0.15);
			\draw[color=black, fill=white] (2,3) circle (0.15);
	\end{tikzpicture}} and $s=\raisebox{-0.3\height}{\begin{tikzpicture}[line cap=round,line join=round,>=triangle 45,x=0.2cm,y=0.2cm]
			% Etage 1
			\draw (0,0)--(0,1);
			%Etage 2
			\draw (0,1)--(1,2) node[right]{$F_{3}$};
			% Etage 3
			\draw (0,1)--(-2,3);
			\draw (-1,2)--(0,3) node[right,above]{$F_{4}$};
			%marquages des sommets
			\filldraw (0,1) circle (0.15);
			\filldraw (-1,2) circle (0.15);
			\draw[color=black, fill=white] (-2,3) circle (0.15);
	\end{tikzpicture}}$.
	Then $\sigma(t,s)=$\raisebox{-0.5\height}{\begin{tikzpicture}[line cap=round,line join=round,>=triangle 45,x=0.3cm,y=0.3cm]
			\draw (0,0)--(0,4);
			%Noeud 1
			\draw (0,1)--(-1,2) node[left]{$F_1$};
			%Noeud 2
			\draw (0,2)--(1,3) node[right]{$F_3$};
			%Noeud 3
			\draw (0,3)--(-1,4) node[left]{$F_2$};
			\draw (0,3)--(1,4) node[right]{$F_4$};
			%marquage des noeuds
			\filldraw (0,1) circle (0.15);
			\filldraw (0,2) circle (0.15);
			\filldraw (0,3) circle (0.15);
			\draw[color=black, fill=white] (0,4) circle (0.15);
	\end{tikzpicture}}.
\end{exam} 

\begin{thm}[\cite{Catoire_2023}]\label{thm:produit}
	Let $t,s$ be two trees different from $\unitree$ as described above. Then: 
	\[
	t*s=\sum_{\sigma\in \batc(k,l)}\sigma(t,s).
	\]
	Moreover:
	\begin{gather*}
		t\prec s=\sum_{\substack{\sigma\in\batc(k,l) \\ \sigma^{-1}(\{1\})=\{1\}}} \sigma(t,s), 
		 t\succ s=\sum_{\substack{\sigma\in\batc(k,l) \\ \sigma^{-1}(\{1\})=\{k+1\}}} \sigma(t,s),\\
		t\cdot s=\sum_{\substack{\sigma\in\batc(k,l) \\ \sigma^{-1}(\{1\})=\{1,k+1\}}} \sigma(t,s).
	\end{gather*}
\end{thm}
This result is a combinatorial description of an inductive formula given in the previous work of J.L.~Loday and M.~Ronco~\cite{LodayRonco_98}.
Moreover,  $(\K\Sch,\prec,\cdot,\succ)$ can be endowed with a codendriform coproduct~\cite{Catoire_2023,LodayRonco_98} detailled in subsection~\ref{sec:cuts}:
\begin{thm}\label{thm:coproduit}
	We define a map $\Delta:\K\Sch\rightarrow \K\Sch\otimes \K\Sch$ by:
	\begin{align*}
		\Delta(t)=\sum_{c \text{ \normalfont pruning of }t} P^c(t)\otimes R^c(t),
	\end{align*}
	where $R^c(t)$ is the component of $t$ containing its root and $P^c(t)=P^c_1(t)*\cdots *P^c_k(t)$, where $P^c_i(t)$ are the cut trees naturally ordered from left to right. Moreover one can write $\Delta=\Delta_{\leftarrow}+\Delta_{\rightarrow}$ in such a way that $(\K\Sch,\prec,\cdot,\succ,\Delta_{\leftarrow},\Delta_{\rightarrow})$ is a $(3,2)$-dendriform bialgebra.
\end{thm}

From the Cartier-Quillen-Milnor-Moore theorem~\cite{CartierPatras_2021}, we can understand better the cocommutative part of any Hopf algebra from its primitive elements.
Hence, a natural question is to compute the primitive elements of $\K\Sch$. Now let us introduce some needed notations to reach the theorems for this purpose:
\begin{defi}
	Let $(C,{\Delta}_{\leftarrow},{\Delta}_{\rightarrow})$ be a codendriform coalgebra~\cite[definition~2 and 3]{Foissy_07}. We define:
	\begin{gather*}
		\Prim_{\Coass}(C):=\{ c\in C \,|\, \tilde{\Delta}(c)=0 \},
		\Prim_{\Codend}(C):=\{ c\in C \,|\, \tilde{\Delta}_{\leftarrow}(c)=\tilde{\Delta}_{\rightarrow}(c)=0 \}.
	\end{gather*}
	where $\tilde{\Delta}_{\leftarrow}(x)=\Delta_{\leftarrow}(x)-x\otimes 1$ and $\tilde{\Delta}_{\rightarrow}(x)=\Delta_{\rightarrow}(x)-1\otimes x$ for any $x\in C.$ 
\end{defi}

In particular, with respect to the grading of $\Sch$, we have:
\begin{align*}
	& \Prim_{\Coass}(\K\Sch(1))=\left\langle\Y\right\rangle, &\Prim_{\Codend}(\K\Sch(1))=\left\langle\Y\right\rangle, \\
	& \Prim_{\Coass}(\K\Sch(2))=\left\langle \balais, \balaisd - \balaisg  \right\rangle, & \Prim_{\Codend}(\K\Sch(2))=\left\langle \balais \right\rangle.
\end{align*}
\begin{thm}[\cite{Catoire_2023}]\label{thm:coassincodend}
	For all $n\in\NN$, we define:
	\begin{equation*}
		\theta_n:\left\lbrace \begin{array}{rcl}
			\Prim_{\Coass}(\K\Sch(n))& \rightarrow & \Prim_{\Codend}(\K\Sch(n+1)), \\
			a & \mapsto & a\cdot \Y.
		\end{array}\right.
	\end{equation*}
	Then, for all $n\in\NN, \theta_n$ is an isomorphism of vector spaces.
\end{thm}
Let us introduce for any $n\in\NN$, the vector space:
\begin{equation}\label{eq:Wn}
W_n=\Prim_{\Codend}(\K\Sch)_n \oplus \left\langle w \, \middle| \, w\in T^k(\Prim_{\Coass}(\K\Sch))(n) \text{ with } k>1 \right\rangle,
\end{equation}
where $T^k(V)$ is the vector space generated by words of length exactly $k$.
We state a variation of M.~Ronco's theorem~\cite[theorem~4.6]{euleridem} in the following theorem~\ref{thm:génération2} in terms of \emph{brace algebras}. We state it in a simpler way in exchange of a loss of injectivity:
\begin{thm}\label{thm:génération2}
	 One can define  a map $\omega$ using $\pright,\pmid$ and $\pleft$ giving rise for all $n\in\NN$ to:
	\begin{equation}\label{eq:Omega}
		\Omega_n:\left\lbrace\begin{array}{rcl}
			W_n & \twoheadrightarrow & \Prim_{\Coass}(\K\Sch(n)), \\
			x_1\otimes\dots\otimes x_k & \mapsto & \omega(x_1\otimes \dots \otimes x_k).
		\end{array}\right.
	\end{equation}
	This raises to a map $\Omega$ defined over $\oplus_{n=1}^{\infty} W_n$ such that $\Omega|_{W_n}=\Omega_n$. The map $\omega$ is defined for any $k\geq 1, x,y,x_1,\dots, x_k\in\K\Sch$ by:
	\begin{align*}
		&\omega(x)\coloneqq x, \\ 
		&\omega(x_1\otimes \ldots \otimes x_k \otimes y)= \sum_{i=0}^k (-1)^{k-i} \omega_{\prec}(x_1\otimes\ldots \otimes x_i) \succeq y \prec \omega_{\succeq}(x_{i+1}\otimes \ldots \otimes x_k),
	\end{align*}
	where $\omega_\prec$ and $\omega_{\succeq}$ are defined by induction
	\begin{align*}
		&\omega_{\prec}(\ew)=\unitree = \omega_{\succ}(\ew), \\
		&\omega_{\prec}(x_1\otimes \ldots \otimes x_k)= x_1\prec \omega_{\prec}(x_2\otimes \ldots \otimes x_k), \\
		&\omega_{\succeq}(x_1\otimes \ldots \otimes x_k)= \omega_{\succeq}(x_1\otimes \ldots \otimes x_{k-1})\succeq x_k.
	\end{align*}
\end{thm}

\begin{exam}
	To compute an element of $\Prim_{\Coass}(\K\Sch(3))$, we can choose $\Y\otimes \Y \otimes \Y$ or $\left(\balaisd - \balaisg\right) \otimes \Y$ which are both elements of $T^{\leq 3}(\K\Prim_{\Coass}(\K\Sch))(3).$ Hence:
\begin{align*}
	\Omega_3\left(\Y\otimes \Y \otimes \Y\right)&= \Y \prec \left( \Y \succeq \Y\right) - \Y \succeq \Y \prec \Y + \left(\Y \prec \Y\right)\succeq \Y,\\
	\Omega_3\left( \left(\balaisd - \balaisg\right) \otimes \Y\right) &= - \Y\prec \left(\balaisd - \balaisg\right) + \left(\balaisd - \balaisg\right)\succeq \Y.
\end{align*} 	
\end{exam}

So, we have a recipe involving $\theta$ and $\Omega$ to produce the primitives of this algebra by induction as shown in figure~\ref{fig:algo_idea}:
%\begin{figure}[h]
%	\centering
%	\begin{tikzcd}
%		&
%		\Prim_{\Codend}(\K\Sch(1)) \ar[out=-20, in=160, "\Omega_1", swap]{dl} \\ 
%		 \Prim_{\Coass}(\K\Sch(2)) \arrow[r,"\theta_1"] & \Prim_{\Codend}(\K\Sch(2)) \ar[out=-20, in=160, swap, "\Omega_2"]{dl} \\
%		\Prim_{\Codend}(\K\Sch(3)) \arrow[r,"\theta_2"] & \dots 
%	\end{tikzcd}
%	\caption{Diagram of the induction}
%	\label{fig:algo_idea0}
%\end{figure}
\begin{figure}[h]
	\centering
	\begin{tikzcd}
		\Prim_{\Codend}(\K\Sch(1)) \ar[d,"\Omega_1"]{dl} & \Prim_{\Codend}(\K\Sch(2)) \ar[d,"\Omega_2"]{dl} & 
		\Prim_{\Codend}(\K\Sch(3)) \ar[d,"\Omega_3"]{dl} &
		\dots  \\
		\Prim_{\Coass}(\K\Sch(1)) \ar[ur,"\theta_1"] &
		\Prim_{\Coass}(\K\Sch(2)) \ar[ur,"\theta_2"] &
		\Prim_{\Coass}(\K\Sch(3)) \ar[ur,"\theta_3"] & 
	\end{tikzcd}
	\caption{Diagram of the generation of primitives by induction}
	\label{fig:algo_idea}
\end{figure}
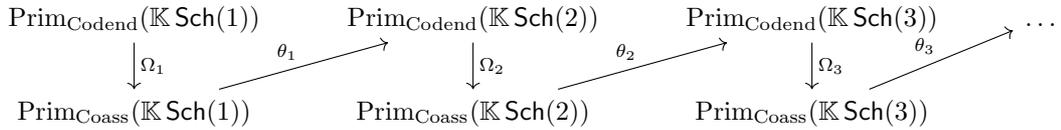

The process of figure~\ref{fig:algo_idea} enables us to compute $\Prim_{\Coass}(\K\Sch)$ by induction starting with ${\Prim_{\Codend}(\K\Sch(1))=\left\langle\Y\right\rangle}$. In this paper, our aim is to implement the algorithm of figure~\ref{fig:algo_idea} and to check our computations. Hence, we look for an interpretation of trees that a computer can understand and that is nicely compatible with the algebraic structure. 

\section{Our point of view of Schroeder trees and their operations}\label{sec:pov}

This section describes our point of view on \Schw{} using levelled \Schw{} trees (definition~\ref{defi:levelled_trees}), initial chains of the poset $(\mathcal{P},\subseteq)$ (definition~\ref{def:init_chain}) and finally grids unsing gray and white squares symbolizing \emph{tree codes} (definition~\ref{defi:lifetime}).

\subsection{Encoding \Schw{} trees}

The idea to represent \Schw{} trees into a computer is to add \emph{levels} giving rise to \emph{levelled \Schw{} trees}. From this intermediate representation, we can interpret a \Schw{} tree as a binary increasing sequence which represents the levels of the tree.

\subsubsection{Levelled Schroeder Tree}

\begin{defi}[\cite{Palacios_06}] \label{defi:levelled_trees}
	A \emph{levelled \Schw{} tree} is pair $(t,l)$ where $t$ is a \Schw{} tree and for one $n\in\NN,l:\Vint(t)\twoheadrightarrow \IEM{1}{n}$ is a map such that for any $(u,v)\in \Vint(t)$ where $u$ is on the path from $v$ of the root of $t$, we have $l(u)>l(v)$. For any $v\in \Vint(t), l(v)$ is called the \emph{level} of $v$. 
	
	Let $n\in\NN.$
	The set of levelled \Schw{} tree is denoted $\Schl$ and the subset of levelled \Schw{} trees with $n$ leaves is denoted $\Schl(\nn)$.
\end{defi}
\begin{Rq}
	A levelled \Schw{} tree is a particular case of \emph{heap ordered trees}~\cite[example~2.3]{Grossman_89}. Moreover, for the sake of readability, all vertices on a same level are horizontally aligned, see example~\ref{exam:Sch}.
\end{Rq}

Let $n\in\NN.$ To each $t\in\Schl(\nn)$ corresponds one $t'\in\Sch(\nn)$ with the surjection forgetting levels
\begin{equation}
	\pi_\nn:\Schl(\nn)\twoheadrightarrow\Sch(\nn).
\end{equation}
Maps $\pi_0$, $\pi_1$ and $\pi_2$ are bijective whereas $\pi_\nn$ is not injective for $n\geq 3$.

\begin{defi}
	Two levelled trees $t$ and $t'$ of $\Schl(\nn)$ are said \emph{equivalent}, denoted $t\equiv t'$, if it represents the same Schroeder tree, that is to say, $\pi_\nn(t)=\pi_\nn(t')$.
\end{defi}

\begin{exam}\label{exam:Sch}
	Referring to the example~\ref{exam:Sch3}, except for $i=7$, the set $\pi_3^{-1}(\{t_i\})$ has only one element.
For example $\pi_3^{-1}(\{t_8\})$ consists of the unique levelled tree:
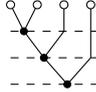
\begin{center}
	\begin{tikzpicture}[x=1em,y=1em]
		\draw(0,3) -- (0.5,2);
		\draw(1,3) -- (0.5,2);
		\draw(2,3) -- (2,2);
		\draw(3,3) -- (3,2) -- (3,1);
		\draw(0.5,2) -- (1.25,1);
		\draw(2,2) -- (1.25,1);
		\draw(1.25,1) -- (2.125, 0);
		\draw(3,1) -- (2.125, 0);
		\draw[color=black, fill=white](0,3) circle(0.15);
		\draw[color=black, fill=white](1,3) circle(0.15);
		\draw[color=black, fill=white](2,3) circle(0.15);
		\draw[color=black, fill=white](3,3) circle(0.15);
		\fill(2.125,0) circle(0.15);
		\fill (1.25,1) circle(0.15);
		\fill (0.5,2) circle(0.15);
		% niveaux en pointillés
		\draw[dashed] (0,0) -- (3.5,0);
		\draw[dashed] (0,1) -- (3.5,1);
		\draw[dashed] (0,2) -- (3.5,2);
	\end{tikzpicture}
\end{center}
Whereas the set $\pi_4^{-1}(\{t_7\})$ consists of the following two levelled trees:
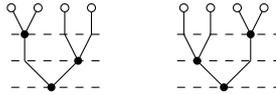
\begin{center}
	\begin{tikzpicture}[x=1em,y=1em]
		\draw(0,3) -- (0.5,2);
		\draw(1,3) -- (0.5,2);
		\draw(2,3) -- (2,2);
		\draw(3,3) -- (3,2);
		\draw(2,2) -- (2.5,1);
		\draw(3,2) -- (2.5,1);
		\draw(0.5,2) -- (0.5,1);
		\draw(0.5,1) -- (1.5,0);
		\draw(2.5,1) -- (1.5,0);
		\draw[color=black, fill=white](0,3) circle(0.15);
		\draw[color=black, fill=white](1,3) circle(0.15);
		\draw[color=black, fill=white](2,3) circle(0.15);
		\draw[color=black, fill=white](3,3) circle(0.15);
		\fill(1.5,0) circle(0.15);
		\fill (2.5,1) circle(0.15);
		\fill (0.5,2) circle(0.15);
		% niveaux en pointillés
		\draw[dashed] (0,0) -- (3.5,0);
		\draw[dashed] (0,1) -- (3.5,1);
		\draw[dashed] (0,2) -- (3.5,2);
	\end{tikzpicture}
	\hspace{2em}
	\begin{tikzpicture}[x=1em,y=1em]
		\draw(0,3) -- (0,2);
		\draw(1,3) -- (1,2);
		\draw(2,3) -- (2.5,2);
		\draw(3,3) -- (2.5,2);
		\draw(0,2) -- (0.5,1);
		\draw(1,2) -- (0.5,1);
		\draw(2.5,2) -- (2.5,1);
		\draw(0.5,1) -- (1.5,0);
		\draw(2.5,1) -- (1.5,0);
		\draw[color=black, fill=white](0,3) circle(0.15);
		\draw[color=black, fill=white](1,3) circle(0.15);
		\draw[color=black, fill=white](2,3) circle(0.15);
		\draw[color=black, fill=white](3,3) circle(0.15);
		\fill(1.5,0) circle(0.15);
		\fill (0.5,1) circle(0.15);
		\fill (2.5,2) circle(0.15);
		% niveaux en pointillés
		\draw[dashed] (-0.25,0) -- (3.5,0);
		\draw[dashed] (-0.25,1) -- (3.5,1);
		\draw[dashed] (-0.25,2) -- (3.5,2);
	\end{tikzpicture}
\end{center}

\end{exam}

We want to code \Schw{} trees as levelled Schroeder trees. For this purpose, in the next section, we choose a representative levelled \Schw{} tree for each \Schw{} tree.

\subsubsection{Our point of view on \Schw{} trees}
\begin{defi}[Initial chains] \label{def:init_chain}
	Let $\mathcal{P}=(P,\leq)$ be a lattice with a minimum and a maximum element.
	A \emph{chain} $C$ of $\mathcal{P}$ is a sequence $(c_0,c_1,\cdots,c_k)$ of elements of $P$ such that $c_i\leq c_{i+1}$ for any $i\in\IEM{1}{k-1}$.
	Let $C$ be a chain in this poset, we say it is \emph{initial} if $\min(\mathcal{P})$ and $\max(\mathcal{P})$ are in $C$.

	Let $n\in\NN\setminus \{0\}$.
	We define $\IC(\IEM{1}{\nn})$ the set of \emph{initial chain} of the poset $(\mathcal{P}(\IEM{1}{\nn}),\subseteq)$. 
\end{defi}
\begin{exam}
	For instance, in the poset $(\mathcal{P}(\IEM{1}{4}),\subseteq)$, we give some elements of $\IC(4)$:
	\begin{align*}
		  \emptyset \subseteq \IEM{1}{4} , &&
		  \emptyset \subseteq \IEM{1}{3} \subseteq \IEM{1}{4}.
	\end{align*}
\end{exam}

%\begin{figure}[h]
%	\centering
%	\caption{The $\mathcal{P}(\IEM{1}{4})$ poset with labels on edges}
%	\label{fig:labelled_IP4}
%	\begin{tikzpicture}[rotate =180 ]
%		%Noeud du diagramme
%		\node (P0) at (0,0) {$\{1,2,3,4\}$};
%		\node (P11) at (-2,1) {$\{1,2,3\}\{4\}$};
%		\node (P12) at (0,1) {$\{1,2\}\{3,4\}$};
%		\node (P13) at (2,1) {$\{1\}\{2,3,4\}$};
%		\node(P21) at (-2,2) {$\{1,2\}\{3\}\{4\}$};
%		\node(P22) at (0,2) {$\{1\}\{2,3\}\{4\}$};
%		\node(P23) at (2,2) {$\{1\}\{2\}\{3,4\}$};
%		\node(P31) at (0,3) {$\{1\}\{2\}\{3\}\{4\}$};
%		%Arêtes du poset
%		\draw (P0) -- (P11);
%		\draw (P0) -- (P12);
%		\draw (P0) -- (P13);
%		\draw (P11) -- (P21);
%		\draw (P11) -- (P22);
%		\draw (P12) -- (P21);
%		\draw (P12) -- (P23);
%		\draw (P13) -- (P22);
%		\draw (P13) -- (P23);
%		\draw (P21) -- (P31);
%		\draw (P22) -- (P31);
%		\draw (P23) -- (P31);
%	\end{tikzpicture}
%\end{figure}

But as said previously many levelled \Schw{} trees give a same \Schw{} tree. Hence, for each $n\in \NN$ we want to build a section of $\pi_n$ so that to a \Schw{} tree corresponds a \emph{unique} levelled \Schw{}.

\subsubsection{Our choice of section}

We introduce the following $s$ map to get a \emph{normal form} of a tree
\begin{prop}\label{prop:trees_to_levels}
	We denote $s:\Sch \rightarrow \Schl$ the unique map where for any $t\in\Sch$, $s(t)$ is the \emph{unique} levelled \Schw{} tree such that for any internal vertex of $t$, its level is strictly greater than all the internal vertices to its right.
	Then, $s$ is a well-defined section of the map $\pi\coloneqq \displaystyle\bigoplus_{n=0}^{+\infty} \pi_n$.
\end{prop}
\begin{proof}
	To prove $s$ is well-defined, we detail a construction of $s(t)$ by induction over the number of leaves of the tree of $t$. Let $t\in\Sch(n)$ with $n\in\NN^*$.
	\begin{description}
		\item[Initialization] for $n=1$, this is obvious as there exists only one representation of $\Y$ as a levelled \Schw{} tree.  
		\item[Heredity] now let us suppose that there exists $n$ such that for any $n'<n$, for any $t\in\Sch(n'), s(t)$ is well defined. We consider $t\in\Sch(n+1)$ and let us build $s(t)$. One can see:
		\[
		t=
\raisebox{-0.5\height}{\begin{tikzpicture}[line cap=round,line join=round,x=0.25cm,y=0.25cm]
		\draw (0,0) -- (0,1);
		\filldraw[color=black] (0,1) circle (1pt);
		% Bouquet
		\draw (0,1) -- (-2,3);
		\draw (0,1) -- (2,3);
		\draw (0,3) node{$\cdots$};
		\draw[left] (-2,3.2) node{$t^{(1)}$};
		\draw[ right] (2,3.2) node{$t^{(k)}$};
		\end{tikzpicture}}= t^{(1)}\vee t^{(2)} \vee \dots \vee t^{(k)}.
		\]
		where  $k\geq 2$ and for any $i\in\IEM{1}{k}, t^{(i)}$ is a \Schw{} tree with at most $n+1$ leaves.
		Hence, we build the levelled \Schw{} tree $s(t)$ as follows:
		\[
		s(t)=
	\raisebox{-0.5\height}
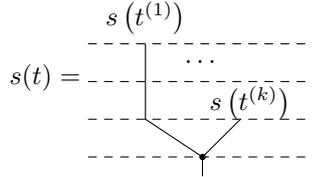
{\begin{tikzpicture}[line cap=round,line join=round,x=0.25cm,y=0.25cm]
			\draw (0,0) -- (0,1);
			\filldraw[color=black] (0,1) circle (1pt);
			% Bouquet
			\draw (0,1) -- (-3,3);
			\draw (-3,3) -- (-3,7);
			\draw (0,1) -- (2,3);
			\draw (0,6) node{$\cdots$};
			\draw[above] (-3,7) node{$s\left(t^{(1)}\right)$};
			\draw[ above] (2.5,2.5) node{$s\left(t^{(k)}\right)$};
			% niveaux en pointillés
			\draw[dashed] (-6,1) -- (6,1);
			\draw[dashed] (-6,7) -- (6,7);
			\draw[dashed] (-6,3) -- (6,3);
			\draw[dashed] (-6,5) -- (6,5);
	\end{tikzpicture}}
		\]
	\end{description}
	where all internal vertices of $s\left(t^{(k)}\right)$ has its level shifted by $1$ and for all $i\in\IEM{1}{k}$, all the vertices of $s\left(t^{(i)}\right)$ have their levels shifted by the level of its one plus the greater level obtained previously. 
	Finally by construction, we have $\pi\circ s=\Id_{\Sch}$.
\end{proof}
\begin{exam}
	Consider the tree
	
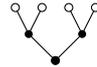
\begin{center}
	\begin{tikzpicture}[x=1em,y=1em]
		\draw(0,2) -- (0.5,1);
		\draw(1,2) -- (0.5,1);
		\draw(2,2) -- (2.5,1);
		\draw(3,2) -- (2.5,1);
		\draw(0.5,1) -- (1.5,0);
		\draw(2.5,1) -- (1.5,0);
		\draw[color=black, fill=white](0,2) circle(0.15);
		\draw[color=black, fill=white](1,2) circle(0.15);
		\draw[color=black, fill=white](2,2) circle(0.15);
		\draw[color=black, fill=white](3,2) circle(0.15);
		\filldraw[color=black] (1.5,0) circle(0.15);
		\fill (2.5,1) circle(0.15);
		\fill (0.5,1) circle(0.15);
	\end{tikzpicture}
\end{center}
	The only tree satisfying this definition among its two representatives of example~\ref{exam:Sch} is:
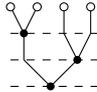
	\begin{center}
		\begin{tikzpicture}[x=1em,y=1em]
			\draw(0,3) -- (0.5,2);
			\draw(1,3) -- (0.5,2);
			\draw(2,3) -- (2,2);
			\draw(3,3) -- (3,2);
			\draw(2,2) -- (2.5,1);
			\draw(3,2) -- (2.5,1);
			\draw(0.5,2) -- (0.5,1);
			\draw(0.5,1) -- (1.5,0);
			\draw(2.5,1) -- (1.5,0);
			\draw[color=black, fill=white](0,3) circle(0.15);
			\draw[color=black, fill=white](1,3) circle(0.15);
			\draw[color=black, fill=white](2,3) circle(0.15);
			\draw[color=black, fill=white](3,3) circle(0.15);
			\fill(1.5,0) circle(0.15);
			\fill (0.5,2) circle(0.15);
			\fill (2.5,1) circle(0.15);
		% niveaux en pointillés
		\draw[dashed] (0,0) -- (3.5,0);
		\draw[dashed] (0,1) -- (3.5,1);
		\draw[dashed] (0,2) -- (3.5,2);
		\end{tikzpicture}
	\end{center}
\end{exam}

\subsubsection{The encoding}

Using the section $s$, we can identify any \Schw{} tree with a finite rectangular grid of white and gray squares thanks to the map of definition~\ref{defi:lifetime}.

\begin{defi}
	Let $t$ be a planar tree and let $v\in V(t)$. An \emph{angle} of $t$ on a vertex $v$ is a pair ${(x,y)\in V(t)}, x\neq y$ of adjacent vertices to $v$ such that the path from the root to $v$ does not contain $x$ and $y$, and the segment from $x$ to $y$ does not cross the tree. 
	
	Moreover, the set of angles of $t$, denoted $A(t)$, is totally ordered reading the tree from left to right.
	See example~\ref{exam:tree_to_code}. 
\end{defi}

\begin{defi}\label{defi:height}
	Let $t\in\Sch{}$.
	The \emph{height} of a tree $t$ denoted $h(t)$ is the integer $|\Vint(t)|$. 
\end{defi}
\begin{exam}\label{exam:tree_to_code}
	For instance, $h\left(\unitree\right)=0,h\left(\Y\right)=1$ and $h\left(\raisebox{-0.5\height}{\begin{tikzpicture}[x=1em,y=1em]
		\draw(0,2) -- (0.5,1);
		\draw(1,2) -- (0.5,1);
		\draw(2,2) -- (2.5,1);
		\draw(3,2) -- (2.5,1);
		\draw(0.5,1) -- (1.5,0);
		\draw(2.5,1) -- (1.5,0);
		\draw (2.5,1) -- (2.5,2);
		\draw[color=black, fill=white](0,2) circle(0.15);
		\draw[color=black, fill=white](1,2) circle(0.15);
		\draw[color=black, fill=white](2,2) circle(0.15);
		\draw[color=black, fill=white](3,2) circle(0.15);
		\draw[color=black, fill=white](2.5,2) circle(0.15);
		\filldraw[color=black] (1.5,0) circle(0.15);
		\fill (2.5,1) circle(0.15);
		\fill (0.5,1) circle(0.15);
	\end{tikzpicture}}
	\right)=3.$
\end{exam}

\begin{defi}\label{defi:lifetime}
	For any $n\in\NN,$ we introduce a map $\varphi_n: \Schl(n)\rightarrow \IC(\IEM{1}{n})$ defined for any $t\in\Schl(n)$ by the \emph{unique} initial chain $C=\left(C_0,\dots,C_{h(t)}\right)$ describing the lifetime of the angles of $t$, labelling them from left to right by $a_i$ with $i\in\IEM{1}{n}$ and putting $v_i$ the vertex of $a_i$,  in function of the levels of the tree:
	\[
	\forall i \in\IEM{0}{h(t)}, C_i=\left\{ j \in\IEM{1}{n} \,\middle|\, l(v_j)\leq i \text{ where } v_j \text{ is the vertex of }a_j\right\}.
	\]
	We also define $\Grid_n$, the map sending any element $(C_0,\ldots, C_h)$ of $\IC(\IEM{1}{n})$ onto a grid diagram with $h$ rows and $n$ columns of white / gray cells. The bottom row of the grid is for $C_0$ and the top row is for $C_k$. Then, for $(i,j)\in\IEM{1}{h(t)}\times\IEM{1}{n}$ cell at row $i$ and column $j$ is filled in gray if and only if $j$ appears in $C_i$.
    %We also define $\Set_n$, the map sending any element of $\IC(\IEM{1}{n})$ onto a list of length $n$ with $0$'s and $1$'s encoding the element of $\mathcal{P}(\IEM{1}{n})$ considered. More precisely, for any ${X\in\mathcal{P}(\IEM{1}{n})}$:
	%\[
	%\Set_n(X)\coloneq (\mathbbm{1}_X(i))_{i\in\IEM{1}{n}}.
	%\]
\end{defi}
\begin{exam}\label{Eg:tree_code}
	For instance, writing sequences from bottom to top:
	\begin{align}
		\raisebox{-0.3\height}
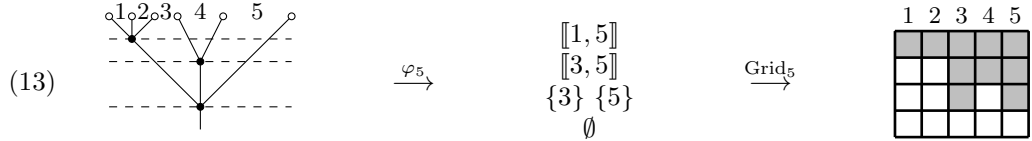
{\begin{tikzpicture}[x=0.3cm, y=0.3cm]
				\draw (0,0)--(0,1);
				\draw (0,1)--(-4,5);
				\draw (0,1)--(4,5);
				\draw (0,1)--(0,3);
				\draw (0,3)--(-1,5);
				\draw (0,3)--(1,5);
				\draw (-3,4)--(-2,5);
				\draw (-3,4)--(-3,5);
				% dessin des noeuds
				\filldraw (0,1) circle (0.15);
				\draw[ color=black, fill=white]  (-4,5) circle (0.15);
				\draw[ color=black, fill=white]  (-3,5) circle (0.15);
				\draw[ color=black, fill=white]  (-2,5) circle (0.15);
				\draw[ color=black, fill=white]  (-1,5) circle (0.15);
				\draw[ color=black, fill=white]  (4,5) circle (0.15);
				\draw[ color=black, fill=white]  (1,5) circle (0.15);
				\filldraw (-3,4) circle (0.15);
				\filldraw (0,3) circle (0.15);		
				% Décorations angulaires
				\draw (-3.5, 5.2) node{\small{$1$}};
				\draw (-2.5, 5.2) node{\small{$2$}};
				\draw (-1.5,5.2) node{\small{$3$}};
				\draw (0, 5.2) node{\small{$4$}};
				\draw (2.5, 5.2) node{\small{$5$}};
                % niveaux
                \draw[dashed] (-4,4) -- (4,4);
                \draw[dashed] (-4,3) -- (4,3);
                \draw[dashed] (-4,1) -- (4,1);
		\end{tikzpicture}} && \overset{\varphi_5}{\longrightarrow} && \begin{array}{c}
			\IEM{1}{5} \\
			\IEM{3}{5} \\
			\{3\}~\{5\} \\
			\emptyset
		\end{array}
		&& 
        \overset{\Grid_5}{\longrightarrow} &&
\begin{tikzpicture}[x=1em,y=1em,baseline = 1.8em]
\small
\filldraw [fill=lightgray] (0,3) rectangle ++ (1,1);
\filldraw [fill=lightgray] (1,3) rectangle ++ (1,1);
\filldraw [fill=lightgray] (2,3) rectangle ++ (1,1);
\filldraw [fill=lightgray] (3,3) rectangle ++ (1,1);
\filldraw [fill=lightgray] (4,3) rectangle ++ (1,1);
\filldraw [fill=lightgray] (2,2) rectangle ++ (1,1);
\filldraw [fill=lightgray] (3,2) rectangle ++ (1,1);
\filldraw [fill=lightgray] (4,2) rectangle ++ (1,1);
\filldraw [fill=lightgray] (2,1) rectangle ++ (1,1);
\filldraw [fill=lightgray] (4,1) rectangle ++ (1,1);
\draw [line width = 1] (0,0) -- ++ (5, 0);
\draw [line width = 1] (0,1) -- ++ (5, 0);
\draw [line width = 1] (0,2) -- ++ (5, 0);
\draw [line width = 1] (0,3) -- ++ (5, 0);
\draw [line width = 1] (0,4) -- ++ (5, 0);
\draw [line width = 1] (0,0) -- ++ (0, 4);
\draw [line width = 1] (1,0) -- ++ (0, 4);
\draw [line width = 1] (2,0) -- ++ (0, 4);
\draw [line width = 1] (3,0) -- ++ (0, 4);
\draw [line width = 1] (4,0) -- ++ (0, 4);
\draw [line width = 1] (5,0) -- ++ (0, 4);
\draw(0.5,4.6) node{$1$};
\draw(1.5,4.6) node{$2$};
\draw(2.5,4.6) node{$3$};
\draw(3.5,4.6) node{$4$};
\draw(4.5,4.6) node{$5$};
\end{tikzpicture}
		%&& \overset{B}{\longrightarrow} && \begin{array}{c}
%			31 \\ 28 \\ 20 \\ 0
%		\end{array}
	\end{align}
\end{exam}

\begin{defi}\label{defi:Code_map}
	We define the \emph{code map}, denoted $\Code$, as the map $\Grid\circ\varphi\circ s:\Sch\to \GRID$ where $\GRID$ is the set of all grids of any size beginin $\GRID$. Let $G\in\Grid$, we put:
	\begin{tikzpicture}
	    
	\end{tikzpicture}g whose first line is gray and its last one is white. 
\end{defi}
\begin{Rq} \noindent
	\begin{itemize}
		\item The height of the tree is the number of lines plus one needed to write its code. 
		\item As each element of $\Schl$ is in bijection with an initial chain. The section $s$ defines a normal form for initial chains that we do not study here.
	\end{itemize}
\end{Rq}

Thanks to this result we can now encode one \Schw{} tree as a list of integers/sequences over $\{0,1\}$ into a computer. Now, we can look at the action of quasi-shuffles on theses encodings from $\ima(\Code )$ for an easy computation in the free tridendriform algebra with one generator.

%\PCt{Je mets une partie pour décrire l'image des $\varphi_n$ que l'on ne conservera peut-être pas.}
%
%\subsubsection{Description of the image of \Schw{} trees}
%
%We can notice that the image of $\varphi\circ s$, where $\varphi=\displaystyle\bigoplus_{n=0}^{+\infty} \varphi_n$, is strictly included inside $\displaystyle \bigoplus_{n=0}^{+\infty} \IC(\IEM{1}{n})$. We give here a description of the image of $\varphi$. For this purpose we need to introduce some definitions:
%\begin{defi}
%	Let $n\in\NN$ and $x\in\mathcal{P}(\IEM{1}{n})$. An \emph{interval} in $x$ is a subset of $x$ of the shape $\IEM{i}{j}$ with $i<j.$
%	We will denote $\Con(x)$ the set of maximal intervals in $x$.
%\end{defi}
%With this connectivity concept, we have:
%\[
%x=\bigcup_{c\in\Con(x)} c.
%\]
%\begin{Rq}
%	Let $n\in\NN$ et $x\in\mathcal{P}(\IEM{1}{n})$.
%	Note that $\Con(x)$ has a total order induced by the order of the minimum of each block. More precisely, for $(c_1,c_2)\in\Con(x)$:
%	\[
%	c_1<c_2 \iff \min(c_1)<\min(c_2).
%	\]
%\end{Rq}
%\begin{defi}
%	Let $n\in\NN$ and consider $C\in\IC(\IEM{1}{n})$. Let $x\in C$ which is not an extrema of $C$. We denote by $\prev{x}$ the previous element in the chain $C$.  We denote by $\suiv{x}$ the next element in the chain $C$. 
%\end{defi}
% 
% \begin{prop}
% 	We have $C\in\ima(\varphi\circ s)$ if and only if:
% 	\[
% 	\forall x\in C\setminus\{\max(C)\}, \forall c\in\Con(x), \forall c'\in \Con(x), c'>c \implies c'\in \suiv{x}.
% 	\]
% \end{prop}
%\begin{proof}
%	Todo
%\end{proof}
%\PCt{Est-ce qu'on parle des suites croissantes pour le poids de Hamming ?} 

\subsection{The tridendriform structure on $\ima(\Code)$}

IN this section, we describe teh action of a quasi-shuffle over a pair of elements of $\ima(\Code)$. In our case, this operation is more complicated than just a concatenation.
\begin{Rq}
	For now, for any $n\in\NN$, we will no longer distinguish elements of $\IC(\IEM{1}{n})$ and its image by the map $\Grid{}$. This enables us to use the notations of both formalism.
\end{Rq}

\subsubsection{Preliminary definitions}

In order to describe properly to the computer what we want to do, we have to identify elements encoding the comb representation of the tree associated to it. 

\begin{defi}[forest code]\label{defi:forest_codes}
	Let $n\in\NN$. A \emph{forest code} $C$ is either an element of $\ima(\Code)$ \emph{or} a sequence such taht there exists $C'\in\ima(\Code)$ %be an initial chain of the set $\IEM{1}{n}$
	 such that $C=C'\setminus \min(C')$.
	The \emph{forest} associated to a forest code is the ordered concatenation of trees we get deleting the root of the tree encoded by $C'$. 
\end{defi}
\begin{Rq}
	The forest code associated to the empty chain, denoted $\varepsilon$, is $\unitree$. 
	Note that if $C$ is a forest code that is not an initial chain, then there exists a unique ordered forest $F$ such that $\Code(T)=C'$ where $T$ is the grafting of all the element of the forest $F$ to a common root. Hence, we extend the definition of the map $\Code$ to forests codes in such a way that $\Code(F)=C.$ 
\end{Rq}
\begin{exam}
	For instance:
	\[
	\Code\left(\balaisd~\Y~\unitree\right)=
	\raisebox{-0.5\height}{
~.
	\]
\end{exam}

\begin{Not}
We introduce the \emph{block notation} in $\GRID$. Let $G\in\Grid$:
\begin{itemize}
	\item \raisebox{-0.5\height}{%
} represents the same grid inserting $G$,
\item	\raisebox{-0.5\height}{%
} represents the grid inserting $G$ and copying its last line in the hatched area.
\end{itemize}
\end{Not}

\begin{defi}[left comb decomposition]\label{defi:left_comb_decomposition}
	Let $n\in\NN$ and $C\in\IC(\IEM{1}{n})$.
	Its \emph{left comb decomposition}, denoted \LCD{C}, is obtained inductively the following way:
	\begin{align*}
	\LCD{\varepsilon}=\varepsilon,  &&
	\LCD{C}=( F_1, \LCD{C_1} ),
	\end{align*}
	where in the induction we assumed $A\times (B \times C)\approx (A\times B) \times C$, $F_1$ and $C_1$ are such that
{C=%
} if it is a tree code or  {C = 
    %
} if it is a forest code,
    where $\tilde{F_1}$ and $\tilde{C_1}$ are respectively $F_1$ and $C_1$ keeping only once any instance of a set in the sequence. Moreover, the spotted central column is the \emph{leftmost} column in $C$ ending with a gray square or with a unique white square if the last line is white. 
\end{defi}
\begin{exam}\label{exam:left_comb_decomposition}
	We illustrate the induction of definition~\ref{defi:left_comb_decomposition} as compare the result to the comb decomposition of trees.
	\[
	s=\raisebox{-0.3\height}{%
}, \varepsilon\right).
	\end{align*}
	Associating to it the $3$-uplet of the forests code by inverting $\Code$ and applying it component-wise, we have:
	\[
	\left( \Y{}~\unitree, \balais{}, \unitree\right).
	\]
	We have recovered the comb decomposition of $t$ (after forgetting the rightmost $\unitree$ which always appears during this process) reading the forests from the root to the leftmost leaf. Moreover, this algorithm gives rise to the left comb representation of a grid:
	\[
	S=\raisebox{-0.3\height}{%
}
	\]
	
\end{exam}

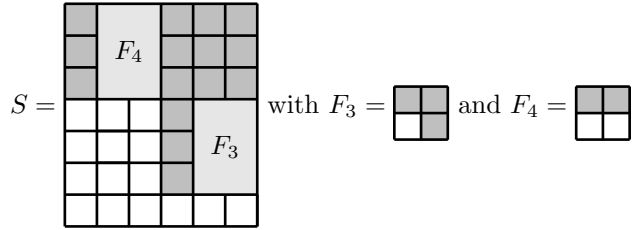
\begin{defi}[right comb decomposition] \label{defi:right_comb_decomposition}
		Let $n\in\NN$ and $C\in\IC(\IEM{1}{n})$.
		Its \emph{right comb decomposition}, denoted \RCD{C}, is obtained inductively the following way:
		\begin{align*}
		\RCD{\varepsilon}=\varepsilon, &&
		\RCD{C}=( F_1, \RCD{C_1} ),
		\end{align*}
		where in the induction we assumed $A\times (B \times C)\approx (A\times B) \times C$,$F_1$ and $C_1$ are such that
{C=%
} if it is a tree code or  {C = 
    %
} if it is a forest code,
    where $\tilde{F_1}$ and $\tilde{C_1}$ are respectively $F_1$ and $C_1$ keeping only once any instance of a set in the sequence. Moreover, the spotted central column is the \emph{rightmost} column in $C$ ending with a gray square or with a unique white square if the last line is white.
\end{defi}
\begin{exam}\label{exam:right_comb_decomposition}
	We illustrate the induction of definition~\ref{defi:right_comb_decomposition} as compare the result to the comb decomposition of trees.
	\[
	t=\raisebox{-0.3\height}{%
},\varepsilon\right)
	\end{align*}
	Associating to it the $3$-uplet of the forests code, we have:
	\[
	\left( \Y, \Y~\unitree, \unitree\right).
	\]
	We have recovered the comb decomposition of $t$ (after forgetting the rightmost $\unitree$ which always appears during this process) reading the forests from the root to the rightmost leaf. Moreover, this algorithm gives rise to the right comb representation of a grid:
	\[
	T=\raisebox{-0.3\height}{%
}
	\]
	\end{exam}
\begin{defi}\label{defi:tree_length}
	For any grid $G\in\ima(\Code)$, we define respectively the \emph{right tree length} and the \emph{left tree length} by:
	\[
	\RTL(G)\coloneqq|\RCD{G}|-1 \text{ and }\LTL(G)\coloneqq |\LCD{G}|-1.
	\]
\end{defi}
Now, we have an analogous of the comb decomposition for our grids. We can now describe the action of a quasi-shuffle on a pair of grids from $\ima(\Code)$.

\subsubsection{The tridendriform structure for initial chains}

By analogy with \Schw{} trees, we put:
\begin{defi}\label{defi:quasishuffle_code}
	Let $T\in\ima(\Code)$ with $n$ columns and $S\in\ima(\Code)$ with $m$ columns. Let us denote $k=\RTL(T)$ and $l=\LTL(S)$. We also put $h_T$ and $h_S$ the number of lines of $T$ and $S$. Consider $\sigma\in\batc(k,l)$ which has for image $\IEM{1}{r}$. We denote by $\sigma(T,S)$ the element of $\GRID$ with $n+m$ columns of $h_T+h_S+r-k-l+1$ lines defined by:
	\begin{enumerate}
		\item at the left top corner, fill the rectangle with $n$ columns and $h_T-k+1$ lines with the $h_T-k+1$ first lines of $T$, at the right of this rectangle put a gray rectangle with $n$ columns and $h_T-k$ lines;  
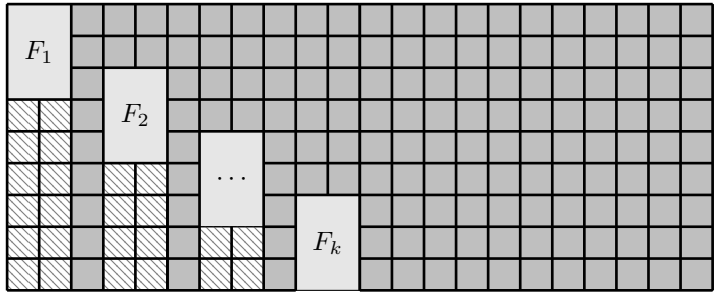
\begin{center}
    \raisebox{-0.3\height}{\begin{tikzpicture}[x=1.2em,y=1.2em,baseline=2.5em]
		\filldraw [fill=lightgray!40] (0,4) rectangle ++ (2,3);
		\filldraw [fill=lightgray!40] (3,2) rectangle ++ (2,3);
		\filldraw [fill=lightgray!40] (6,0) rectangle ++ (2,3);
		\filldraw [fill=lightgray!40] (9,-2) rectangle ++ (2,3);
		\filldraw [fill=lightgray] (2,5) rectangle ++ (3,2);
		\filldraw [fill=lightgray] (2,-2) rectangle ++ (1,7);
		\filldraw [fill=lightgray] (8,-2) rectangle ++ (1,3);
		\filldraw [fill=lightgray] (5,-2) rectangle ++ (1,9);
		\filldraw [fill=lightgray] (6,3) rectangle ++ (5,4);
		\filldraw [fill=lightgray] (8,1) rectangle ++ (3,2);
		\draw[pattern=north west lines, pattern color=gray] (0,-2) rectangle ++ (2,6);
		\draw[pattern=north west lines, pattern color=gray] (3,-2) rectangle ++ (2,4);
		\draw[pattern=north west lines, pattern color=gray] (6,-2) rectangle ++ (2,2);
		%horizontal
		\draw [line width = 1] (0,7) -- ++ (11, 0);
		\draw [line width = 1] (2,6) -- ++ (9, 0);
		\draw [line width = 1] (5,4) -- ++ (6, 0);
		\draw [line width = 1] (0,4) -- ++ (3, 0);
		\draw [line width = 1] (2,5) -- ++ (9, 0);
		\draw [line width = 1] (0,3) -- ++ (3, 0);
		\draw [line width = 1] (0,2) -- ++ (6, 0);
		\draw [line width = 1] (0,1) -- ++ (6, 0);
		\draw [line width = 1] (0,0) -- ++ (6, 0);
		\draw [line width = 1] (0,-1) -- ++ (9, 0);
		\draw [line width = 1] (0,-2) -- ++ (9, 0);
		\draw [line width = 1] (5,3) -- ++ (6, 0);
		\draw [line width = 1] (8,2) -- ++ (3, 0);
		\draw [line width = 1] (8,1) -- ++ (3, 0);
		\draw [line width = 1] (8,0) -- ++ (1, 0);
		%vertical
		\draw [line width = 1] (0,-2) -- ++ (0, 9);
		\draw [line width = 1] (2,-2) -- ++ (0, 9);
		\draw [line width = 1] (3,-2) -- ++ (0, 9);	
		\draw [line width = 1] (5, -2) -- ++ (0, 9);
		\draw [line width = 1] (6, -2) -- ++ (0, 9);
		\draw [line width = 1] (4, -2) -- ++ (0, 4);
		\draw [line width = 1] (11, -2) -- ++ (0, 9);
		%			\draw [line width = 1] (10, -2) -- ++ (0, 1);
		\draw [line width = 1] (10, 1) -- ++ (0, 6);
		\draw [line width = 1] (9, -2) -- ++ (0, 9);
		\draw [line width = 1] (8, -2) -- ++ (0, 9);
		\draw [line width = 1] (7, -2) -- ++ (0, 2);
		\draw [line width = 1] (7, 3) -- ++ (0, 4);
		\draw [line width = 1] (5, -2) -- ++ (0, 3);
		\draw [line width = 1] (4, 5) -- ++ (0, 2);
		\draw [line width = 1] (1,-2) -- ++ (0, 6);
		\draw (1,5.5) node{${F_1}$};
		\draw (4,3.5) node{${F_2}$};
		\draw (7,1.5) node{${\cdots}$};
		\draw (10,-0.5) node{${F_{k}}$};
		% tableau 
		\begin{scope}[shift={(11,-2)}]
					\filldraw [fill=lightgray] (0,0) rectangle ++ (11,9);
					\draw [line width = 1] (0,0) -- ++ (11, 0);
					\draw [line width = 1] (0,1) -- ++ (11, 0);
					\draw [line width = 1] (0,2) -- ++ (11, 0);
					\draw [line width = 1] (0,3) -- ++ (11, 0);
					\draw [line width = 1] (0,4) -- ++ (11, 0);
					\draw [line width = 1] (0,5) -- ++ (11, 0);
					\draw [line width = 1] (0,6) -- ++ (11, 0);
					\draw [line width = 1] (0,7) -- ++ (11, 0);
					\draw [line width = 1] (0,8) -- ++ (11, 0);
					\draw [line width = 1] (0,9) -- ++ (11, 0);
					\draw [line width = 1] (0,0) -- ++ (0, 9);
					\draw [line width = 1] (1,0) -- ++ (0, 9);
					\draw [line width = 1] (2,0) -- ++ (0, 9);
					\draw [line width = 1] (3,0) -- ++ (0, 9);
					\draw [line width = 1] (4,0) -- ++ (0, 9);
					\draw [line width = 1] (5,0) -- ++ (0, 9);
					\draw [line width = 1] (6,0) -- ++ (0, 9);
					\draw [line width = 1] (7,0) -- ++ (0, 9);
					\draw [line width = 1] (8,0) -- ++ (0, 9);
					\draw [line width = 1] (9,0) -- ++ (0, 9);
					\draw [line width = 1] (10,0) -- ++ (0, 9);
					\draw [line width = 1] (11,0) -- ++ (0, 9);
		\end{scope}
\end{tikzpicture}}
\end{center}
		\item for $h$ from $r$ to $1$ with a $-1$ step: 
		\begin{description} 
			\item[Case 1] if $\sigma^{-1}(\{h\})=\{i\}, i\in\IEM{1}{k}$, write white blocks below the code of $F_i$ and on its adjacent right column. Fill the remaining slots of this line with the content of the previous one. Hence, one just concatenate this line to the current grid where $I$ is the column positions for the code of $F_i$ and its right adjacacent column:
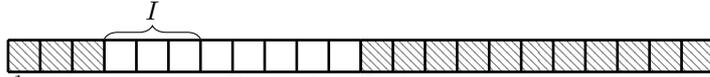
			\begin{center}
			
			\begin{tikzpicture}[x=1.2em,y=1.2em,baseline=2.5em]
	%dessins
	\draw[pattern=north west lines, pattern color=gray] (0,0) rectangle ++ (3,1);
	\draw[pattern=north west lines, pattern color=gray] (11,1) rectangle ++ (11,-1);
	% horizontal
	\draw[line width = 1] (0,0) -- ++ (11,0);
	\draw[line width = 1] (0,1) -- ++ (11,0);
	% vertical
	\draw[line width = 1] (0,0) -- ++ (0,1);
	\draw[line width = 1] (1,0) -- ++ (0,1);
	\draw[line width = 1] (2,0) -- ++ (0,1);
	\draw[line width = 1] (3,0) -- ++ (0,1);
	\draw[line width = 1] (4,0) -- ++ (0,1);
	\draw[line width = 1] (5,0) -- ++ (0,1);
	\draw[line width = 1] (6,0) -- ++ (0,1);
	\draw[line width = 1] (7,0) -- ++ (0,1);
	\draw[line width = 1] (8,0) -- ++ (0,1);
	\draw[line width = 1] (9,0) -- ++ (0,1);
	\draw[line width = 1] (10,0) -- ++ (0,1);
	\draw[line width = 1] (11,0) -- ++ (0,1);
	%dessins
	\draw [decorate,decoration={brace,amplitude=5pt,raise=0.1ex}]
	(3,1) -- (6,1) node[midway,yshift=1.1em]{$I$};
	\begin{scope}[shift={(11,0)}]
		% horizontal
		\draw[line width = 1] (0,0) -- ++ (11,0);
		\draw[line width = 1] (0,1) -- ++ (11,0);
		% vertical
		\draw[line width = 1] (0,0) -- ++ (0,1);
		\draw[line width = 1] (1,0) -- ++ (0,1);
		\draw[line width = 1] (2,0) -- ++ (0,1);
		\draw[line width = 1] (3,0) -- ++ (0,1);
		\draw[line width = 1] (4,0) -- ++ (0,1);
		\draw[line width = 1] (5,0) -- ++ (0,1);
		\draw[line width = 1] (6,0) -- ++ (0,1);
		\draw[line width = 1] (7,0) -- ++ (0,1);
		\draw[line width = 1] (8,0) -- ++ (0,1);
		\draw[line width = 1] (9,0) -- ++ (0,1);
		\draw[line width = 1] (10,0) -- ++ (0,1);
		\draw[line width = 1] (11,0) -- ++ (0,1);
	\end{scope}
\end{tikzpicture}
			\end{center}
			\item[Case 2] if $\sigma^{-1}(\{h\})=\{j\}, j\in\IEM{k+1}{k+l}$, denote $h_j$ the number of lines to write $F_j$'s code. From the leftmost square of the line, use the first $n$ columns of the previous line  to fill a rectangle with $n$ columns and $h_j-1$ lines. For the other $m$ columns, write the $h_j-1$ last lines of $S$'s code describing forest $F_j$. Once it is done, copy the previous line and overwrite the content of the columns under $F_j$'s code and its left adjacent column by white squares. Hence, one just concatenates this grid (except for the first time) to the current one :
			
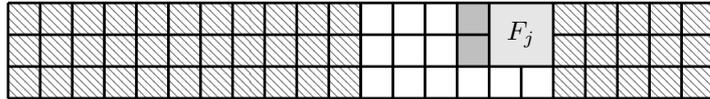
\begin{center}
			\begin{tikzpicture}[x=1.2em,y=1.2em,baseline=2.5em]
	%dessins
	\draw[pattern=north west lines, pattern color=gray] (0,0) rectangle ++ (11,3);
	% horizontal
	\draw[line width = 1] (0,0) -- ++ (11,0);
	\draw[line width = 1] (0,1) -- ++ (11,0);
	\draw[line width = 1] (0,2) -- ++ (11,0);
	\draw[line width = 1] (0,3) -- ++ (11,0);
	% vertical
	\draw[line width = 1] (0,0) -- ++ (0,3);
	\draw[line width = 1] (1,0) -- ++ (0,3);
	\draw[line width = 1] (2,0) -- ++ (0,3);
	\draw[line width = 1] (3,0) -- ++ (0,3);
	\draw[line width = 1] (4,0) -- ++ (0,3);
	\draw[line width = 1] (5,0) -- ++ (0,3);
	\draw[line width = 1] (6,0) -- ++ (0,3);
	\draw[line width = 1] (7,0) -- ++ (0,3);
	\draw[line width = 1] (8,0) -- ++ (0,3);
	\draw[line width = 1] (9,0) -- ++ (0,3);
	\draw[line width = 1] (10,0) -- ++ (0,3);
	\draw[line width = 1] (11,0) -- ++ (0,3);
	%dessins
	\begin{scope}[shift={(11,0)}]
	    \draw[pattern=north west lines, pattern color=gray] (6,0) rectangle ++ (5,3);
		\filldraw [fill=lightgray!40] (4,1) rectangle ++ (2,2);
		\filldraw [fill=lightgray] (3,1) rectangle ++ (1,2);
		%horizontal
		\draw [line width = 1] (0,0) -- ++ (11, 0);
		\draw [line width = 1] (0,1) -- ++ (11, 0);
		\draw [line width = 1] (0,2) -- ++ (4, 0);
		\draw [line width = 1] (6,2) -- ++ (5, 0);
		\draw [line width = 1] (0,3) -- ++ (11, 0);
		%vertical
		\draw [line width = 1] (0,0) -- ++ (0, 3);
		\draw [line width = 1] (1,0) -- ++ (0, 3);
		\draw [line width = 1] (2,0) -- ++ (0, 3);
		\draw [line width = 1] (3,0) -- ++ (0, 3);
		\draw [line width = 1] (4,0) -- ++ (0, 3);
		\draw [line width = 1] (5,0) -- ++ (0, 1);
		\draw [line width = 1] (6,0) -- ++ (0, 3);
		\draw [line width = 1] (7,0) -- ++ (0, 3);
		\draw [line width = 1] (8,0) -- ++ (0, 3);
		\draw [line width = 1] (9,0) -- ++ (0, 3);
		\draw [line width = 1] (10,0) -- ++ (0, 3);
		\draw [line width = 1] (11,0) -- ++ (0, 3);
		%dessins
		\draw (5,2) node{$F_j$};
	\end{scope}
\end{tikzpicture}
			\end{center}
			\item[Case 3] if $\sigma^{-1}(\{h\})=\{i,j\}, (i,j)\in\IEM{1}{k}\times\IEM{k+1}{k+l}$, denote $h_j$ the number of lines to write $F_j$'s code. From the leftmost square of the line, use the first $n$ columns of the previous line  to fill a rectangle with $n$ columns and $h_j-1$ lines. For the other $m$ columns, write the $h_j-1$ lines of $S$'s code describing forest $F_j$. Once it is done, copy the previous line and overwrite the content of the columns under $F_j$'s code and its left adjacent column \emph{and} the columns under $F_i$'s code and its right adjacent column with white squares. Hence, one just adds this line to the current grid where $I$ is the column positions for the code of $F_i$ and its right adjacacent column:
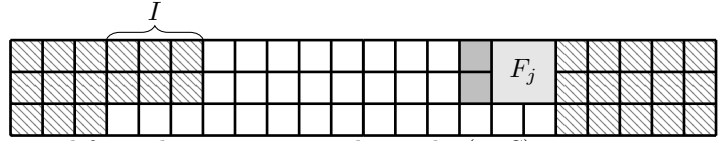
			\begin{center}
			\begin{tikzpicture}[x=1.2em,y=1.2em,baseline=2.5em]
	%dessins
	\draw[pattern=north west lines, pattern color=gray] (0,0) rectangle ++ (3,3);
	\draw[pattern=north west lines, pattern color=gray] (3,1) rectangle ++ (3,2);
	% horizontal
	\draw[line width = 1] (0,0) -- ++ (11,0);
	\draw[line width = 1] (0,1) -- ++ (11,0);
	\draw[line width = 1] (0,2) -- ++ (11,0);
	\draw[line width = 1] (0,3) -- ++ (11,0);
	% vertical
	\draw[line width = 1] (0,0) -- ++ (0,3);
	\draw[line width = 1] (1,0) -- ++ (0,3);
	\draw[line width = 1] (2,0) -- ++ (0,3);
	\draw[line width = 1] (3,0) -- ++ (0,3);
	\draw[line width = 1] (4,0) -- ++ (0,3);
	\draw[line width = 1] (5,0) -- ++ (0,3);
	\draw[line width = 1] (6,0) -- ++ (0,3);
	\draw[line width = 1] (7,0) -- ++ (0,3);
	\draw[line width = 1] (8,0) -- ++ (0,3);
	\draw[line width = 1] (9,0) -- ++ (0,3);
	\draw[line width = 1] (10,0) -- ++ (0,3);
	\draw[line width = 1] (11,0) -- ++ (0,3);
	%dessins
		\draw [decorate,decoration={brace,amplitude=5pt,raise=0.1ex}] (3,3) -- (6,3) node[midway,yshift=1.1em]{$I$};
		\begin{scope}[shift={(11,0)}]
	    \draw[pattern=north west lines, pattern color=gray] (6,0) rectangle ++ (5,3);
		\filldraw [fill=lightgray!40] (4,1) rectangle ++ (2,2);
		\filldraw [fill=lightgray] (3,1) rectangle ++ (1,2);
		%horizontal
		\draw [line width = 1] (0,0) -- ++ (11, 0);
		\draw [line width = 1] (0,1) -- ++ (11, 0);
		\draw [line width = 1] (0,2) -- ++ (4, 0);
		\draw [line width = 1] (6,2) -- ++ (5, 0);
		\draw [line width = 1] (0,3) -- ++ (11, 0);
		%vertical
		\draw [line width = 1] (0,0) -- ++ (0, 3);
		\draw [line width = 1] (1,0) -- ++ (0, 3);
		\draw [line width = 1] (2,0) -- ++ (0, 3);
		\draw [line width = 1] (3,0) -- ++ (0, 3);
		\draw [line width = 1] (4,0) -- ++ (0, 3);
		\draw [line width = 1] (5,0) -- ++ (0, 1);
		\draw [line width = 1] (6,0) -- ++ (0, 3);
		\draw [line width = 1] (7,0) -- ++ (0, 3);
		\draw [line width = 1] (8,0) -- ++ (0, 3);
		\draw [line width = 1] (9,0) -- ++ (0, 3);
		\draw [line width = 1] (10,0) -- ++ (0, 3);
		\draw [line width = 1] (11,0) -- ++ (0, 3);
		%dessins
		\draw (5,2) node{$F_j$};
	\end{scope}

\end{tikzpicture}
			\end{center}
		\end{description}
		The element obtained from this operation is denoted $\sigma(T,S)$.
	\end{enumerate} 
\end{defi}

\begin{exam}
		Consider the $(2,2)$-quasi-shuffle $\sigma=(1,3,2,3)$. Let us take 
	$t=$	\raisebox{-0.3\height}{\begin{tikzpicture}[line cap=round,line join=round,>=triangle 45,x=0.2cm,y=0.2cm]
		% Etage 1
		\draw (0,0)--(0,1);
		%Etage 2
		\draw (0,1)--(-1,2) node[left]{$f_1$};
		% Etage 3
		\draw (0,1)--(2,3);
		\draw (1,2)--(0,3) node[left,above]{$f_2$};
		%marquages des sommets
		\filldraw (0,1) circle (0.15);
		\filldraw (1,2) circle (0.15);
		\draw[color=black, fill=white] (2,3) circle (0.15);
		\end{tikzpicture}} and $s=\raisebox{-0.3\height}{\begin{tikzpicture}[line cap=round,line join=round,>=triangle 45,x=0.2cm,y=0.2cm]
		% Etage 1
		\draw (0,0)--(0,1);
		%Etage 2
		\draw (0,1)--(1,2) node[right]{$f_{3}$};
		% Etage 3
		\draw (0,1)--(-2,3);
		\draw (-1,2)--(0,3) node[right,above]{$f_{4}$};
		%marquages des sommets
		\filldraw (0,1) circle (0.15);
		\filldraw (-1,2) circle (0.15);
		\draw[color=black, fill=white] (-2,3) circle (0.15);
		\end{tikzpicture}}$.
	Then $\sigma(t,s)=$\raisebox{-0.5\height}{\begin{tikzpicture}[line cap=round,line join=round,>=triangle 45,x=0.3cm,y=0.3cm]
		\draw (0,0)--(0,4);
		%Noeud 1
		\draw (0,1)--(-1,2) node[left]{$f_1$};
		%Noeud 2
		\draw (0,2)--(1,3) node[right]{$f_3$};
		%Noeud 3
		\draw (0,3)--(-1,4) node[left]{$f_2$};
		\draw (0,3)--(1,4) node[right]{$f_4$};
		%marquage des noeuds
		\filldraw (0,1) circle (0.15);
		\filldraw (0,2) circle (0.15);
		\filldraw (0,3) circle (0.15);
		\draw[color=black, fill=white] (0,4) circle (0.15);
		\end{tikzpicture}}.
	Then, let us denote for all $i\in\IEM{1}{4}, \Code(f_i)=F_i$.
	Hence, the code associated to $t$ and $s$ respectively denoted $T$ and $S$ are of the shape:
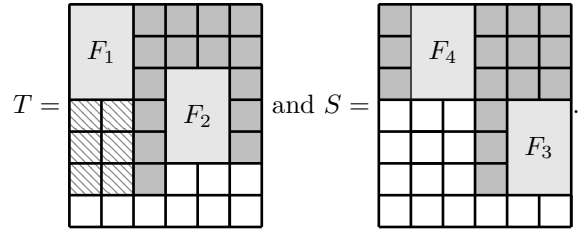
\[
T=\raisebox{-0.3\height}{\begin{tikzpicture}[x=1.2em,y=1.2em,baseline=2.5em]
		\filldraw [fill=lightgray!40] (0,4) rectangle ++ (2,3);
		\filldraw [fill=lightgray!40] (3,2) rectangle ++ (2,3);
		\filldraw [fill=lightgray] (2,5) rectangle ++ (3,2);
		\filldraw [fill=lightgray] (2,1) rectangle ++ (1,4);
		\filldraw [fill=lightgray] (5,2) rectangle ++ (1,5);
		\draw[pattern=north west lines, pattern color=gray] (0,1) rectangle ++ (2,3);
		%horizontal
		\draw [line width = 1] (0,7) -- ++ (6, 0);
		\draw [line width = 1] (2,6) -- ++ (4, 0);
		\draw [line width = 1] (2,4) -- ++ (1, 0);
		\draw [line width = 1] (5,4) -- ++ (1, 0);
		\draw [line width = 1] (0,4) -- ++ (2, 0);
		\draw [line width = 1] (2,5) -- ++ (4, 0);
		\draw [line width = 1] (0,3) -- ++ (3, 0);
		\draw [line width = 1] (0,2) -- ++ (6, 0);
		\draw [line width = 1] (0,1) -- ++ (6, 0);
		\draw [line width = 1] (0,0) -- ++ (6, 0);
		\draw [line width = 1] (5,3) -- ++ (1, 0);
		%vertical
		\draw [line width = 1] (0,0) -- ++ (0, 7);
		\draw [line width = 1] (2,0) -- ++ (0, 7);
		\draw [line width = 1] (3,0) -- ++ (0, 7);	
		\draw [line width = 1] (5, 0) -- ++ (0, 7);
		\draw [line width = 1] (6, 0) -- ++ (0, 7);
		\draw [line width = 1] (4, 0) -- ++ (0, 2);
		\draw [line width = 1] (5, 0) -- ++ (0, 1);
		\draw [line width = 1] (4, 5) -- ++ (0, 2);
		\draw [line width = 1] (1,0) -- ++ (0, 4);
		\draw(1,5.5) node{${F_1}$};
		\draw(4,3.5) node{${F_2}$};
\end{tikzpicture}}
\text{ and }
S=\raisebox{-0.3\height}{\begin{tikzpicture}[x=1.2em,y=1.2em,baseline=2.5em]
		\filldraw [fill=lightgray!40] (0,4) rectangle ++ (2,3);
		\filldraw [fill=lightgray!40] (3,1) rectangle ++ (2,3);
		\filldraw [fill=lightgray] (2,4) rectangle ++ (3,3);
		\filldraw [fill=lightgray] (2,1) rectangle ++ (1,3);
		\filldraw [fill=lightgray] (-1,4) rectangle ++ (1,3);
		%horizontal
		\draw [line width = 1] (2,6) -- ++ (3, 0);
		\draw [line width = 1] (-1,4) -- ++ (6, 0);
		\draw [line width = 1] (2,5) -- ++ (3, 0);
		\draw [line width = 1] (0,3) -- ++ (1, 0);
		\draw [line width = 1] (0,2) -- ++ (1, 0);
		\draw [line width = 1] (-1,1) -- ++ (6, 0);
		\draw [line width = 1] (-1,0) -- ++ (6, 0);
		\draw [line width = 1] (-1,7) -- ++ (6, 0);
		\draw [line width = 1] (-1,2) -- ++ (4, 0);
		\draw [line width = 1] (-1,3) -- ++ (4, 0);
		\draw [line width = 1] (-1,5) -- ++ (1, 0);
		\draw [line width = 1] (-1,6) -- ++ (1, 0);
		%vertical
		\draw [line width = 1] (2,0) -- ++ (0, 7);
		\draw [line width = 1] (3,0) -- ++ (0, 7);	
		\draw [line width = 1] (5, 0) -- ++ (0, 7);
		\draw [line width = 1] (-1, 0) -- ++ (0, 7);
		\draw [line width = 1] (4, 0) -- ++ (0, 1);
		\draw [line width = 1] (5, 0) -- ++ (0, 1);
		\draw [line width = 1] (4, 4) -- ++ (0, 3);
		\draw [line width = 1] (1,0) -- ++ (0, 4);
		\draw [line width = 1] (0,0) -- ++ (0, 4);
		\draw(1,5.5) node{${F_4}$};
		\draw(4,2.5) node{${F_3}$};
	\end{tikzpicture}}.
	\]
	Then $\sigma(T,S)$ is the code of the tree $\sigma(t,s)$ given by: 
	
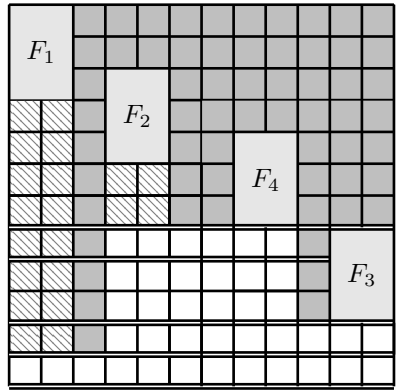
\begin{center}
	\begin{tikzpicture}[x=1.2em,y=1.2em,baseline=2.5em]
		\filldraw [fill=lightgray!40] (0,4) rectangle ++ (2,3);
		\filldraw [fill=lightgray!40] (3,2) rectangle ++ (2,3);
		\filldraw [fill=lightgray] (2,5) rectangle ++ (3,2);
		\filldraw [fill=lightgray] (2,-4) rectangle ++ (1,10);
		\filldraw [fill=lightgray] (5,3) rectangle ++ (7,4);
		\filldraw [fill=lightgray] (5,0) rectangle ++ (1,3);
		\draw[pattern=north west lines, pattern color=gray] (0,-4) rectangle ++ (2,8);
		\draw[pattern=north west lines, pattern color=gray] (3,0) rectangle ++ (2,2);
		%horizontal
		\draw [line width = 1] (0,7) -- ++ (12, 0);
		\draw [line width = 1] (2,6) -- ++ (10, 0);
		\draw [line width = 1] (2,4) -- ++ (1, 0);
		\draw [line width = 1] (5,4) -- ++ (7, 0);
		\draw [line width = 1] (2,5) -- ++ (10, 0);
		\draw [line width = 1] (0,3) -- ++ (3, 0);
		\draw [line width = 1] (0,2) -- ++ (6, 0);
		\draw [line width = 1] (0,1) -- ++ (6, 0);
		\draw [line width = 1] (0,0) -- ++ (12, 0);
		\draw [line width = 1] (5,3) -- ++ (7, 0);
		%vertical
		\draw [line width = 1] (0,-5) -- ++ (0, 12);
		\draw [line width = 1] (2,-5) -- ++ (0, 12);
		\draw [line width = 1] (3,-5) -- ++ (0, 12);
		\draw [line width = 1] (5, -5) -- ++ (0, 12);
		\draw [line width = 1] (6, -5) -- ++ (0, 12);
		\draw [line width = 1] (4, -5) -- ++ (0, 7);
		\draw [line width = 1] (5, -5) -- ++ (0, 6);
		\draw [line width = 1] (4, 5) -- ++ (0, 2);
		\draw [line width = 1] (1,-5) -- ++ (0, 9);
		\draw(1,5.5) node{${F_1}$};
		\draw(4,3.5) node{${F_2}$};
		
		\begin{scope}[shift={(7,-4)}]
			\filldraw [fill=lightgray!40] (0,4) rectangle ++ (2,3);
			\filldraw [fill=lightgray!40] (3,1) rectangle ++ (2,3);
			%\filldraw [fill=lightgray] (3,3) rectangle ++ (2,1);
			\filldraw [fill=lightgray] (2,4) rectangle ++ (3,4);
			\filldraw [fill=lightgray] (0,7) rectangle ++ (3,1);
			\filldraw [fill=lightgray] (2,1) rectangle ++ (1,3);
			\filldraw [fill=lightgray] (-1,4) rectangle ++ (1,4);
			%horizontal
			\draw [line width = 1] (-1,7) -- ++ (6, 0);
			\draw [line width = 1] (2,6) -- ++ (3, 0);
			\draw [line width = 1] (2,4) -- ++ (3, 0);
			\draw [line width = 1, double] (-7,4) -- ++ (12, 0);
			\draw [line width = 1] (2,5) -- ++ (3, 0);
			\draw [line width = 1] (0,3) -- ++ (3, 0);
			\draw [line width = 1] (0,2) -- ++ (3, 0);
			\draw [line width = 1, double] (-7,0) -- ++ (12, 0);
			\draw [line width = 1, double] (-7,-1) -- ++ (12, 0);
			\draw [line width = 1] (-7,1) -- ++ (12, 0);
			\draw [line width = 1, double] (-7,1) -- ++ (12, 0);
			\draw [line width = 1] (-7,2) -- ++ (9, 0);
			\draw [line width = 1, double] (-7,3) -- ++ (10, 0);
			\draw [line width = 1] (-7,5) -- ++ (7, 0);
			\draw [line width = 1] (-7,6) -- ++ (7, 0);
			%vertical
			\draw [line width = 1] (0,-1) -- ++ (0, 12);
			\draw [line width = 1] (2,-1) -- ++ (0, 12);
			\draw [line width = 1] (3,-1) -- ++ (0, 12);
			\draw [line width = 1] (5,-1) -- ++ (0, 12);
			\draw [line width = 1] (-1,-1) -- ++ (0, 7);
			\draw [line width = 1] (4,-1) -- ++ (0, 2);
			\draw [line width = 1] (5, 0) -- ++ (0, 1);
			\draw [line width = 1] (4, 4) -- ++ (0, 7);
			\draw [line width = 1] (1,-1) -- ++ (0, 5);
			\draw [line width = 1] (1,7) -- ++ (0, 4);
			\draw(1,5.5) node{${F_4}$};
			\draw(4,2.5) node{${F_3}$};
		\end{scope}
		\end{tikzpicture}
	\end{center}

For more details, the reader can refer to example~\ref{exam:algo_atomic_product} in the algorithmic section. In the particular case, where $T$ are $S$ the grids of respectively from examples~\ref{exam:right_comb_decomposition} and \ref{exam:left_comb_decomposition}, one gets the grid:
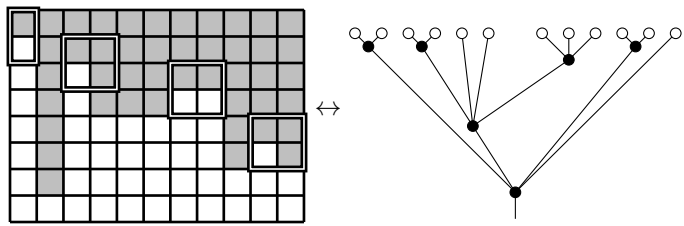
\[
		\raisebox{-0.5\height}{\begin{tikzpicture}[x=1em,y=1em]
				\filldraw [fill=lightgray] (0,7) rectangle ++ (1,1);
				\filldraw [fill=lightgray] (1,7) rectangle ++ (1,1);
				\filldraw [fill=lightgray] (2,7) rectangle ++ (1,1);
				\filldraw [fill=lightgray] (3,7) rectangle ++ (1,1);
				\filldraw [fill=lightgray] (4,7) rectangle ++ (1,1);
				\filldraw [fill=lightgray] (5,7) rectangle ++ (1,1);
				\filldraw [fill=lightgray] (6,7) rectangle ++ (1,1);
				\filldraw [fill=lightgray] (7,7) rectangle ++ (1,1);
				\filldraw [fill=lightgray] (8,7) rectangle ++ (1,1);
				\filldraw [fill=lightgray] (9,7) rectangle ++ (1,1);
				\filldraw [fill=lightgray] (10,7) rectangle ++ (1,1);
				\filldraw [fill=lightgray] (1,6) rectangle ++ (1,1);
				\filldraw [fill=lightgray] (2,6) rectangle ++ (1,1);
				\filldraw [fill=lightgray] (3,6) rectangle ++ (1,1);
				\filldraw [fill=lightgray] (4,6) rectangle ++ (1,1);
				\filldraw [fill=lightgray] (5,6) rectangle ++ (1,1);
				\filldraw [fill=lightgray] (6,6) rectangle ++ (1,1);
				\filldraw [fill=lightgray] (7,6) rectangle ++ (1,1);
				\filldraw [fill=lightgray] (8,6) rectangle ++ (1,1);
				\filldraw [fill=lightgray] (9,6) rectangle ++ (1,1);
				\filldraw [fill=lightgray] (10,6) rectangle ++ (1,1);
				\filldraw [fill=lightgray] (1,5) rectangle ++ (1,1);
				\filldraw [fill=lightgray] (3,5) rectangle ++ (1,1);
				\filldraw [fill=lightgray] (4,5) rectangle ++ (1,1);
				\filldraw [fill=lightgray] (5,5) rectangle ++ (1,1);
				\filldraw [fill=lightgray] (6,5) rectangle ++ (1,1);
				\filldraw [fill=lightgray] (7,5) rectangle ++ (1,1);
				\filldraw [fill=lightgray] (8,5) rectangle ++ (1,1);
				\filldraw [fill=lightgray] (9,5) rectangle ++ (1,1);
				\filldraw [fill=lightgray] (10,5) rectangle ++ (1,1);
				\filldraw [fill=lightgray] (1,4) rectangle ++ (1,1);
				\filldraw [fill=lightgray] (3,4) rectangle ++ (1,1);
				\filldraw [fill=lightgray] (4,4) rectangle ++ (1,1);
				\filldraw [fill=lightgray] (5,4) rectangle ++ (1,1);
				\filldraw [fill=lightgray] (8,4) rectangle ++ (1,1);
				\filldraw [fill=lightgray] (9,4) rectangle ++ (1,1);
				\filldraw [fill=lightgray] (10,4) rectangle ++ (1,1);
				\filldraw [fill=lightgray] (1,3) rectangle ++ (1,1);
				\filldraw [fill=lightgray] (8,3) rectangle ++ (1,1);
				\filldraw [fill=lightgray] (9,3) rectangle ++ (1,1);
				\filldraw [fill=lightgray] (10,3) rectangle ++ (1,1);
				\filldraw [fill=lightgray] (1,2) rectangle ++ (1,1);
				\filldraw [fill=lightgray] (8,2) rectangle ++ (1,1);
				\filldraw [fill=lightgray] (10,2) rectangle ++ (1,1);
				\filldraw [fill=lightgray] (1,1) rectangle ++ (1,1);
				\draw [line width = 1] (0,0) -- ++ (11, 0);
				\draw [line width = 1] (0,1) -- ++ (11, 0);
				\draw [line width = 1] (0,2) -- ++ (11, 0);
				\draw [line width = 1] (0,3) -- ++ (11, 0);
				\draw [line width = 1] (0,4) -- ++ (11, 0);
				\draw [line width = 1] (0,5) -- ++ (11, 0);
				\draw [line width = 1] (0,6) -- ++ (11, 0);
				\draw [line width = 1] (0,7) -- ++ (11, 0);
				\draw [line width = 1] (0,8) -- ++ (11, 0);
				\draw [line width = 1] (0,0) -- ++ (0, 8);
				\draw [line width = 1] (1,0) -- ++ (0, 8);
				\draw [line width = 1] (2,0) -- ++ (0, 8);
				\draw [line width = 1] (3,0) -- ++ (0, 8);
				\draw [line width = 1] (4,0) -- ++ (0, 8);
				\draw [line width = 1] (5,0) -- ++ (0, 8);
				\draw [line width = 1] (6,0) -- ++ (0, 8);
				\draw [line width = 1] (7,0) -- ++ (0, 8);
				\draw [line width = 1] (8,0) -- ++ (0, 8);
				\draw [line width = 1] (9,0) -- ++ (0, 8);
				\draw [line width = 1] (10,0) -- ++ (0, 8);
				\draw [line width = 1] (11,0) -- ++ (0, 8);
				%%%%% Mise en valeurs des forêts
				\draw [double,line width = 1] (0,8) rectangle ++ (1,-2);
				\draw [double,line width = 1] (2,7) rectangle ++ (2,-2);
				\draw [double,line width = 1] (6,6) rectangle ++ (2,-2);
				\draw [double,line width = 1] (9,4) rectangle ++ (2,-2);
		\end{tikzpicture}} \leftrightarrow \raisebox{-0.3\height}{\begin{tikzpicture}[x=1em,y=1em,baseline=2.5em]
	%bords de l'arbre
	\draw (0,0) -- (0,1);
	\draw (0,1) -- (-6,7);
	\draw (0,1) -- (6,7);
		% dessins de l'arbre
	\draw (-5,7) -- (-5.5,6.5);
	\draw (-4,7) -- (-3.5,6.5);
	\draw (-3,7) -- (-3.5,6.5);
	\draw (-3.5,6.5) -- (0,1);
	\draw (-2,7) -- (-1.59,3.5);
	\draw (-1,7) -- (-1.59,3.5);
	\draw (1,7) -- (2,6);
	\draw (2,7) -- (2,6);
	\draw (3,7) -- (2,6);
	\draw (2,6) -- (-1.59,3.5);
	\draw (4,7) -- (4.5,6.5);
	\draw (5,7) -- (4.5,6.5);
	\draw (4.5,6.5) -- (0,1);
	%noeuds de l'arbre
	\filldraw[fill=black] (0,1) circle (2pt);
	\filldraw[fill=white] (-6,7) circle (2pt);
	\filldraw[fill=white] (-5,7) circle (2pt);
	\filldraw[fill=white] (-4,7) circle (2pt);
	\filldraw[fill=white] (-3,7) circle (2pt);
	\filldraw[fill=white] (-2,7) circle (2pt);
	\filldraw[fill=white] (-1,7) circle (2pt);
	\filldraw[fill=white] (6,7) circle (2pt);
	\filldraw[fill=white] (5,7) circle (2pt);
	\filldraw[fill=white] (4,7) circle (2pt);
	\filldraw[fill=white] (3,7) circle (2pt);
	\filldraw[fill=white] (2,7) circle (2pt);
	\filldraw[fill=white] (1,7) circle (2pt);
	%sommets internes
	\filldraw [fill=black] (-1.59,3.5) circle (2pt);
	\filldraw [fill=black] (-5.5,6.5) circle (2pt);
	\filldraw [fill=black] (2,6) circle (2pt);
	\filldraw [fill=black] (-3.5,6.5) circle (2pt);
	\filldraw [fill=black] (4.5,6.5) circle (2pt);
\end{tikzpicture}
}
		\]
\end{exam}

With this definition, we are now able to define three operators in $\ima(\Code)$.
\begin{defi}\label{defi:operations_codes}
	Let $T, S\in \ima(\Code)$ such that ${k=|\RTL(T)|}$ and ${l=|\LTL(S)|}$. We define:
	\begin{gather*}
		T \prec S= \sum_{\substack{\sigma\in\batc(k,l) \\ \sigma^{-1}(\{1\})=\{1\}}} \sigma(T,S), \quad
		T \succ S= \sum_{\substack{\sigma\in\batc(k,l) \\ \sigma^{-1}(\{1\})=\{k+1\}}} \sigma(T,S), \\
		T \cdot S= \sum_{\substack{\sigma\in\batc(k,l) \\ \sigma^{-1}(\{1\})=\{1,k+1\}}} \sigma(T,S).
	\end{gather*}
\end{defi}

 Definition~\ref{defi:quasishuffle_code} gives us a mathematical background to implement programs in section~\ref{sec:algorithms}.
We now show that this structure is a tridendriform structure on $\ima(\Code)$ as it is isomorphic to the free tridendriform algebra of \Schw{} trees.

\begin{thm} 
	Let $t,s$ be two \Schw{} trees with shapes given in equation~\eqref{eq:tree_combs}. Let us put $k$ the right comb length of $t$ and $l$ the left comb length of $s$. Let $\sigma\in\batc(k,l)$. Then:
	\[
	\Code(\sigma(t,s))= \sigma(\Code(t),\Code(s)).
	\]
\end{thm}
\begin{proof}
In order to prove this theorem, we will proceed by induction over the sum of $k$ and $l$. If $l=0$ or $k=0$, this means that $s=\unitree$ or $t=\unitree$. We initialize the induction for $0,1$ and $2$.
\begin{description}
	\item[Initialization] if $\max(k,l)=1$ and $k+l=1$, then the identity shuffle is acting on this pair of tree and one of them is $\unitree$. By definition~\ref{defi:quasishuffle_code}, we get back the tree code of the non-unit element.
	We get to the case $k+l=2$. If $k=0$ or $l=0$, then the result is obvious. So, we can assume $k,l\neq 0$ hence $\sigma=(1,1)$, hence:
	\[
	t=\raisebox{-0.3\height}{\begin{tikzpicture}[x=0.25cm,y=0.25cm]
		\filldraw[color=black] (0,1) circle (0.15);
		\draw (0,0) -- (0,1);
		\draw (0,1) -- (1,2);
		\draw (0,1) -- (-1,2);
		\draw (-1,2) node[left] {$f_1$};
	\end{tikzpicture}}, 
	s=\raisebox{-0.3\height}{\begin{tikzpicture}[x=0.25cm,y=0.25cm]
		\draw (0,0) -- (0,1);
		\filldraw[color=black] (0,1) circle (0.15);
		\draw (0,1) -- (1,2);
		\draw (0,1) -- (-1,2);
		\draw (1,2) node[right] {$f_2$};
	\end{tikzpicture}} \text{ and } \sigma(t,s)=\raisebox{-0.3\height}{\begin{tikzpicture}[x=0.25cm,y=0.25cm]
	\draw (0,0) -- (0,1);
	\draw (0,1) -- (1,2);
	\draw (0,1) -- (-1,2);
	\filldraw[color=black] (0,1) circle (0.15);
	\draw (1,2) node[right] {$f_2$};
	\draw (-1,2) node[left] {$f_1$};
	\end{tikzpicture}}
	\]
	Denoting $\Code(f_1)=F_1$ and $\Code(f_2)=F_2$, their respective codes are:
	\[
	T=	\begin{tikzpicture}[x=1.2em,y=1.2em,baseline=2.5em]
    	\filldraw [fill=lightgray!40] (0,1) rectangle ++ (2,3);
    	\filldraw [fill=lightgray] (2,1) rectangle ++ (1,3);
    	%horizontal
    	\draw [line width = 1] (0,1) -- ++ (3, 0);
    	\draw [line width = 1] (2,2) -- ++ (1, 0);
    	\draw [line width = 1] (2,3) -- ++ (1, 0);
    	\draw [line width = 1] (0,4) -- ++ (3, 0);
    	\draw [line width = 1] (0,0) -- ++ (3, 0);
    	%vertical
    	\draw [line width = 1] (0,1) -- ++ (0, 3);
    	\draw [line width = 1] (2,1) -- ++ (0, 3);
    	\draw [line width = 1] (3,1) -- ++ (0, 3);
    	\draw [line width = 1] (0, 0) -- ++ (0, 1);
    	\draw [line width = 1] (1, 0) -- ++ (0, 1);
    	\draw [line width = 1] (2, 0) -- ++ (0, 1);
    	\draw [line width = 1] (3, 0) -- ++ (0, 1);
    	\draw(1,2.5) node{${F_1}$};
    	\end{tikzpicture}, S=	\begin{tikzpicture}[x=1.2em,y=1.2em,baseline=2.5em]
	\filldraw [fill=lightgray!40] (1,1) rectangle ++ (2,3);
	\filldraw [fill=lightgray] (0,1) rectangle ++ (1,3);
	%horizontal
	\draw [line width = 1] (0,1) -- ++ (3, 0);
	\draw [line width = 1] (0,2) -- ++ (1, 0);
	\draw [line width = 1] (0,3) -- ++ (1, 0);
	\draw [line width = 1] (0,0) -- ++ (3, 0);
	\draw [line width = 1] (0,1) -- ++ (3, 0);
	\draw [line width = 1] (0,4) -- ++ (3, 0);
	%vertical
	\draw [line width = 1] (0,1) -- ++ (0, 3);
	\draw [line width = 1] (1,1) -- ++ (0, 3);
	\draw [line width = 1] (3,1) -- ++ (0, 3);
	\draw [line width = 1] (0, 0) -- ++ (0, 1);
	\draw [line width = 1] (1, 0) -- ++ (0, 1);
	\draw [line width = 1] (2, 0) -- ++ (0, 1);
	\draw [line width = 1] (3, 0) -- ++ (0, 1);
	\draw(2,2.5) node{${F_2}$};
	\end{tikzpicture}
 \text{ and } \sigma(T,S)=\raisebox{-0.3\height}{	\begin{tikzpicture}[x=1.2em,y=1.2em,baseline=2.5em]
		\filldraw [fill=lightgray!40] (5,1) rectangle ++ (2,3);
		\filldraw [fill=lightgray!40] (1,3) rectangle ++ (2,3);
		\filldraw [fill=lightgray] (3,4) rectangle ++ (4,2);
		\filldraw [fill=lightgray] (3,1) rectangle ++ (2,3);
    		\draw[pattern=north west lines, pattern color=gray] (1,1) rectangle ++ (2,2);
		%vertical
		\draw [line width = 1] (1,0) -- ++ (0, 6);
		\draw [line width = 1] (3,0) -- ++ (0, 6);
		\draw [line width = 1] (4,0) -- ++ (0, 6);	
		\draw [line width = 1] (5,0) -- ++ (0, 6);
		\draw [line width = 1] (7,0) -- ++ (0, 6);
		\draw [line width = 1] (2,0) -- ++ (0, 3);
		\draw [line width = 1] (6,4) -- ++ (0, 2);
		\draw [line width = 1] (6,1) -- ++ (0, -1);
		%horizontal
		\draw [line width = 1] (1,0) -- ++ (6, 0);
		\draw [line width = 1] (1,1) -- ++ (6, 0);
		\draw [line width = 1] (1,6) -- ++ (6, 0);
		\draw [line width = 1] (1,2) -- ++ (4, 0);	
		\draw [line width = 1] (1,3) -- ++ (4, 0);
		\draw [line width = 1] (3,6) -- ++ (4, 0);
		\draw [line width = 1] (3,5) -- ++ (4, 0);
		\draw [line width = 1] (3,4) -- ++ (4, 0);
		%decorations
		\draw(2,4.5) node{$F_1$};
		\draw (6,2.5) node{$F_2$};
	\end{tikzpicture}}
	\]
	Note that the code of $\sigma(T,S)$ has for associated tree $\sigma(t,s)$. 
	\item[Heredity] Suppose that there exists $r\in\NN^*$ such that for all couple $(k,l)$ with $k+l\leq r$, we have for any $t,s$ whose tree lengths are respectively $k$ and $l$:
	\[
	\sigma(\Code(t),\Code(s))=\Code(\sigma(t,s)).
	\]
	Let us prove that this property is true for the next rank. Let $t,s$ be two trees with right comb length $k$ and left comb length $l$, let $\sigma\in\bat(k,l)$ such that $k+l=r+1$. Note that if $k=0$ or $l=0$, the result is obvious by definition~\ref{defi:operations_codes}. There are three possible cases left, in those cases, we use the notations of the proof of lemma~\ref{lem:Indution_batc}:
	\begin{description} 
		\item[First case] $\sigma^{-1}(\{1\})=\{1\}$:
		Then: 
		\begin{align*}
			\Code(\sigma(t,s))&=\Code\left(\raisebox{-0.3\height}{\begin{tikzpicture}[x=0.25cm,y=0.25cm]
				\draw (0,0) -- (0,1);
				\draw (0,1) -- (-1,2);
				\draw (0,1) -- (1,2);
				\filldraw[color=black] (0,1) circle (0.15);
				\draw (-1,2) node[left] {$f_1$};
				\draw (1,2) node[right] {$\sigma'(t',s)$};
			\end{tikzpicture}}\right)\text{ where } t=\raisebox{-0.3\height}{\begin{tikzpicture}[x=0.25cm,y=0.25cm]
			\draw (0,0) -- (0,1);
			\draw (0,1) -- (-1,2);
			\draw (0,1) -- (1,2);
			\filldraw[color=black] (0,1) circle (0.15);
			\draw (-1,2) node[left] {$f_1$};
			\draw (1,2) node[right] {$t'$};
		\end{tikzpicture}}, \\
		&=\raisebox{-0.5\height}{\begin{tikzpicture}[x=1.2em,y=1.2em,baseline=2.5em]
		\filldraw [fill=lightgray!40] (4,1) rectangle ++ (3,3);
		\filldraw [fill=lightgray!40] (1,3) rectangle ++ (2,3);
		\filldraw [fill=lightgray] (3,4) rectangle ++ (4,2);
		\filldraw [fill=lightgray] (3,0) rectangle ++ (1,4);
		\draw[pattern=north west lines, pattern color=gray] (1,1) rectangle ++ (2,2);
		%vertical
		\draw [line width = 1] (1,0) -- ++ (0, 6);
		\draw [line width = 1] (3,0) -- ++ (0, 6);
		\draw [line width = 1] (4,0) -- ++ (0, 6);
		\draw [line width = 1] (7,0) -- ++ (0, 6);
		\draw [line width = 1] (2,0) -- ++ (0, 3);
		\draw [line width = 1] (6,4) -- ++ (0, 2);
		\draw [line width = 1] (6,1) -- ++ (0, -1);
		\draw [line width = 1] (5,1) -- ++ (0, -1);
		\draw [line width = 1] (5,4) -- ++ (0, 2);
		%horizontal
		\draw [line width = 1] (1,0) -- ++ (6, 0);
		\draw [line width = 1] (1,6) -- ++ (6, 0);
		\draw [line width = 1] (1,0) -- ++ (4, 0);
		\draw [line width = 1] (1,1) -- ++ (3, 0);
		\draw [line width = 1] (1,2) -- ++ (3, 0);	
		\draw [line width = 1] (1,3) -- ++ (3, 0);
		\draw [line width = 1] (3,6) -- ++ (4, 0);
		\draw [line width = 1] (3,5) -- ++ (4, 0);
		\draw [line width = 1] (3,4) -- ++ (4, 0);
		\draw [line width = 1] (3,3) -- ++ (1, 0);
		%decorations
		\draw(2,4.5) node{${F_1}$};
		\draw (5.5,2.5) node{$T'$};
\end{tikzpicture}}
 \text{ where } T'=\Code(\sigma'(t',s)).
		\end{align*}
		By definition of the action of $\sigma$ over a pair of forests codes:
		\begin{align*}
			\sigma(\Code(t),\Code(s))&=\raisebox{-0.5\height}{\begin{tikzpicture}[x=1.2em,y=1.2em,baseline=2.5em]
		\filldraw [fill=lightgray!40] (4,1) rectangle ++ (3,3);
		\filldraw [fill=lightgray!40] (1,3) rectangle ++ (2,3);
		\filldraw [fill=lightgray] (3,4) rectangle ++ (4,2);
		\filldraw [fill=lightgray] (3,1) rectangle ++ (1,3);
		\draw[pattern=north west lines, pattern color=gray] (1,1) rectangle ++ (2,2);
		%vertical
		\draw [line width = 1] (1,0) -- ++ (0, 6);
		\draw [line width = 1] (3,0) -- ++ (0, 6);
		\draw [line width = 1] (4,0) -- ++ (0, 6);
		\draw [line width = 1] (7,0) -- ++ (0, 6);
		\draw [line width = 1] (2,0) -- ++ (0, 3);
		\draw [line width = 1] (6,4) -- ++ (0, 2);
		\draw [line width = 1] (6,1) -- ++ (0, -1);
		\draw [line width = 1] (5,1) -- ++ (0, -1);
		\draw [line width = 1] (5,4) -- ++ (0, 2);
		%horizontal
		\draw [line width = 1] (1,0) -- ++ (6, 0);
		\draw [line width = 1] (1,6) -- ++ (6, 0);
		\draw [line width = 1] (1,0) -- ++ (4, 0);
		\draw [line width = 1] (1,1) -- ++ (3, 0);
		\draw [line width = 1] (1,2) -- ++ (3, 0);	
		\draw [line width = 1] (1,3) -- ++ (3, 0);
		\draw [line width = 1] (3,6) -- ++ (4, 0);
		\draw [line width = 1] (3,5) -- ++ (4, 0);
		\draw [line width = 1] (3,4) -- ++ (4, 0);
		\draw [line width = 1] (3,3) -- ++ (1, 0);
		%decorations
		\draw(2,4.5) node{${F_1}$};
		\draw (5.5,2.5) node{${'T}$};
\end{tikzpicture}} \text{ where }{'T}=\sigma'(\Code(t'),\Code(s)) \\
			&= \sigma(\Code(t), \Code(s)),
		\end{align*}
		Applying the induction hypothesis because $\RTL(t')+\LTL(s)\leq r, T'='T$ so we get the equality needed.
	\item[Second case] $\sigma^{-1}(\{1\})=\{k+1\}$ %, we denote $\sigma'$ the unique element of $\batc(k,l-1)$ such that:
%		\[
%		\forall n\in\IEM{1}{k+l}, \sigma(n)=\begin{cases}
%			k+1 &\text{ if } n=k+1, \\
%			\sigma'(n)+1 & \text{ if }n\leq k, \\
%			\sigma'(n-1)+1 & \text{otherwise.}
%		\end{cases}
%		\]
			is similar to the first one. So it is left to the reader.
		\item[Third case] $\sigma^{-1}(\{1\})=\{1,k+1\}$.
		Then:
		\begin{align*}
			\Code(\sigma(t,s))&= \Code\left(\raisebox{-0.3\height}{\begin{tikzpicture}[x=0.25cm, y=0.25cm]
				\draw (0,0) -- (0,3);
				\draw (0,1) -- (-2,3);
				\draw (0,1) -- (2,3);
				\draw (0,3) node[above]{$\sigma'(t',s')$};
				\draw (-2,3) node[left]{$f_1$};
				\draw (2,3) node[right]{$f_{k+1}$};	
		        \filldraw[color=black] (0,1) circle (0.15);
			\end{tikzpicture}}\right) \text{ where } s=\raisebox{-0.3\height}{\begin{tikzpicture}[x=0.25cm,y=0.25cm]
			\draw (0,0) -- (0,1);
			\draw (0,1) -- (1,2);
			\draw (0,1) -- (-1,2);
		    \filldraw[color=black] (0,1) circle (0.15);
			\draw (-1,2) node[left]{$s'$};
			\draw (1,2) node[right]{$f_{k+1}$};
			\end{tikzpicture}} \\
			&=\raisebox{-0.5\height}{\begin{tikzpicture}[x=1.2em,y=1.2em,baseline=2.5em]
		\filldraw [fill=lightgray!40] (4,5) rectangle ++ (3,3);
		\filldraw [fill=lightgray!40] (1,7) rectangle ++ (2,3);
		\filldraw [fill=lightgray!40] (8,2) rectangle ++ (2,3);
		\filldraw [fill=lightgray] (8,5) rectangle ++ (2,5);
		\filldraw [fill=lightgray] (7,2) rectangle ++ (1,8);
		\filldraw [fill=lightgray] (4,8) rectangle ++ (3,2);
		\filldraw [fill=lightgray] (3,2) rectangle ++ (1,8);
		\draw[pattern=north west lines, pattern color=gray] (1,2) rectangle ++ (2,5);
		%vertical
		\draw [line width = 1] (1,1) -- ++ (0, 9);
		\draw [line width = 1] (3,1) -- ++ (0, 9);
		\draw [line width = 1] (4,1) -- ++ (0, 9);
		\draw [line width = 1] (7,1) -- ++ (0, 9);
		\draw [line width = 1] (8,1) -- ++ (0, 9);
		\draw [line width = 1] (10,1) -- ++ (0, 9);
		\draw [line width = 1] (2,1) -- ++ (0, 6);
		\draw [line width = 1] (6,8) -- ++ (0, 2);
		\draw [line width = 1] (6,1) -- ++ (0, 4);
		\draw [line width = 1] (9,1) -- ++ (0, 1);
		\draw [line width = 1] (5,1) -- ++ (0, 4);
		\draw [line width = 1] (5,8) -- ++ (0, 2);
		\draw [line width = 1] (9,5) -- ++ (0, 5);
		%horizontal
		\draw [line width = 1] (1,3) -- ++ (7, 0);
		\draw [line width = 1] (1,4) -- ++ (7, 0);
		\draw [line width = 1] (1,10) -- ++ (9, 0);
		\draw [line width = 1] (1,2) -- ++ (9, 0);
		\draw [line width = 1] (1,1) -- ++ (9, 0);
		\draw [line width = 1] (1,4) -- ++ (3, 0);
		\draw [line width = 1] (1,5) -- ++ (3, 0);
		\draw [line width = 1] (1,6) -- ++ (3, 0);	
		\draw [line width = 1] (1,7) -- ++ (3, 0);
		\draw [line width = 1] (3,10) -- ++ (7, 0);
		\draw [line width = 1] (3,9) -- ++ (7, 0);
		\draw [line width = 1] (3,8) -- ++ (7, 0);
		\draw [line width = 1] (7,7) -- ++ (3, 0);
		\draw [line width = 1] (7,5) -- ++ (3, 0);
		\draw [line width = 1] (7,6) -- ++ (3, 0);
		%decorations
		\draw(2,8.5) node{${F_1}$};
		\draw (5.5,6.5) node{$T'$};
		\draw (9,3.5) node{$F_{k+1}$};
		\end{tikzpicture}} \text{ where } T'=\Code(\sigma'(t',s'))
		\end{align*}

		Applying the induction hypothesis to $t',s'$ and $\sigma'$ as $\RTL(t')+\LTL(s')\leq r$, we deduce $\Code(\sigma'(t',s'))=\sigma'(\Code(t'),\Code(s')).$ Then applying the construction case number $3$ in definition~\ref{defi:quasishuffle_code}, one has:
		\[
		\Code(\sigma(t,s))=\sigma(\Code(t),\Code(s)).
		\]
	\end{description}
	Hence by the induction principle, we have proved the theorem. \qedhere
\end{description}	
\end{proof}
Finally, the map $\Code$ is also linear hence:
\begin{Cor}\label{Cor:iso_codes_trees}
	The $\Code$ map is an isomorphism of tridendriform algebras into its image.
\end{Cor}
\begin{proof}
	To prove this corollary, one just needs to show that $\Code$ is a tridendriform morphism from $(\K\Sch,\prec,\succ,\cdot)$ into $(\ima(\Code),\prec, \succ, \cdot)$. Let $t,s\in\Sch$ such that $t$ has right comb length $k$ and $s$ has left comb length $l$. We show $\Code$ is a morphism for the middle product:
	\begin{align*}
		\Code(t\cdot s)&=\sum_{\substack{\sigma\in\batc(k,l) \\ \sigma^{-1}{\{1\}}=\{1,k+1\}}} \Code(\sigma(t,s)) \\
		&=\sum_{\substack{\sigma\in\batc(k,l) \\ \sigma^{-1}{\{1\}}=\{1,k+1\}}} \sigma(\Code(t),\Code(s)) \\
		&=\Code(t) \cdot\Code(s).
	\end{align*}
	The other cases are similar.
\end{proof}
Now we are sure the two descriptions are isomorphic. Before getting into coding, we describe what is a cut on our codes and its relationship with some \emph{packed words}. This allows us to check our computations of primitive elements.

\subsection{Short description of the coproduct}\label{sec:coproduct_tree_codes}

We now describe without much details how to interpret the coproduct operation of theorem~\ref{thm:coproduit} for tree codes. Let us remind what is a pruning:
\begin{defi}
	Let $t$ be a \Schw{} tree. A \emph{pruning} of $t$ is a non-empty choice in $\Eint(t)$, the set of internal edges of $t$, such that any path from a leaf of the tree to its root meets at most one chosen edge. A pruning is called a \emph{single cut} when it chooses a unique edge.
	Hence, a pruning is be composed of many single cuts. 
	Given a pruning $c$, it splits the tree into many connected components removing cut edges. The component with the root of the starting tree is denoted $R^c(t)$ and the other components are denoted $P^c_1(t),\dots, P^c_k(t)$.
	We also consider following two elements as prunings:
	\begin{itemize}
		\item the \emph{empty cut} giving $R^c(t)=t$ and $P^c(t)=\unitree$;
		\item the \emph{total cut} giving $R^c(t)=\unitree$ and $P^c(t)=t$.
	\end{itemize}    
\end{defi}
\begin{exam}\label{exam:cut_tree}
	For instance, symbolizing a cut edge by an horizontal line, one has:
	\[
	\raisebox{-0.3\height}
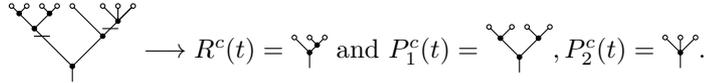
{\begin{tikzpicture}[line cap=round,line join=round,>=triangle 45,x=0.2cm,y=0.2cm]
			\draw (0,0)--(0,1);
			% bords de l'arbre
			\draw (0,1)--(-4,5);
			\draw (0,1)--(4,5);
			% bras droit
			\draw (2,3) -- (0,5);
			\draw (3,4) -- (2,5);
			\draw (3,4) -- (3,5);
			\draw (3,4) -- (4,5);
			%bras gauche
			\draw (-3.5,4.5) -- (-3,5);
			\draw (-2.5,3.5) -- (-1,5);
			\draw (-1.5,4.5) -- (-2,5);
			%coupe
			\draw (-2.5,3) -- (-1.5,3);
			\draw (2,3.5) -- (3,3.5);
			% marquages des sommets
			\draw[color=black, fill=white] (-4,5) circle (0.15);
			\draw[color=black, fill=white] (-2,5) circle (0.15);
			\draw[color=black, fill=white] (0,5) circle (0.15);
			\draw[color=black, fill=white] (2,5) circle (0.15);
			\draw[color=black, fill=white] (3,5) circle (0.15);
			\draw[color=black, fill=white] (4,5) circle (0.15);
			\draw[color=black, fill=white] (-3,5) circle (0.15);
			\draw[color=black, fill=white] (-1,5) circle (0.15);
			\filldraw (0,1) circle(0.15);
			\filldraw (-3.5,4.5) circle(0.15);
			\filldraw (-2.5,3.5) circle(0.15);
			\filldraw (-1.5,4.5) circle(0.15);
			\filldraw (3,4) circle(0.15);
			\filldraw (2,3) circle(0.15);
	\end{tikzpicture}} \longrightarrow R^c(t)= \balaisd{} \text{ and } P^c_1(t)=\YY{}, P^c_2(t)=\balais{}.
	\]
\end{exam}

In other words, a pruning consists of withdrawing an interval of angles such that both adjacent angles (if they exist) are still living when they all vanished. Moreover, the withdrawn intervals should not overlap to get an admissible cut.
Translating this in term of tree codes, doing a pruning of a tree code $T= \Code(t)$ where $t$ has $l$ angles, is choosing a subset of $\IEM{1}{l}$ where for each interval $I$ of this subset, when the tree code restricted to the columns of $I$ is for the first time a line of white squares, the two adjacent columns $\min(I)-1$ and $\max(I)+1$ (if they exist) have a black square on this line.

\begin{exam}
	Consider the tree and its cut in example~\ref{exam:cut_tree}, its tree code is given below and the intervals associated to the cut are in evidence. Then:
	\[
	\raisebox{-0.5\height}{\begin{tikzpicture}[x=1em,y=1em]
\filldraw [fill=lightgray] (0,6) rectangle ++ (1,1);
\filldraw [fill=lightgray] (1,6) rectangle ++ (1,1);
\filldraw [fill=lightgray] (2,6) rectangle ++ (1,1);
\filldraw [fill=lightgray] (3,6) rectangle ++ (1,1);
\filldraw [fill=lightgray] (4,6) rectangle ++ (1,1);
\filldraw [fill=lightgray] (5,6) rectangle ++ (1,1);
\filldraw [fill=lightgray] (6,6) rectangle ++ (1,1);
\filldraw [fill=lightgray] (1,5) rectangle ++ (1,1);
\filldraw [fill=lightgray] (2,5) rectangle ++ (1,1);
\filldraw [fill=lightgray] (3,5) rectangle ++ (1,1);
\filldraw [fill=lightgray] (4,5) rectangle ++ (1,1);
\filldraw [fill=lightgray] (5,5) rectangle ++ (1,1);
\filldraw [fill=lightgray] (6,5) rectangle ++ (1,1);
\filldraw [fill=lightgray] (1,4) rectangle ++ (1,1);
\filldraw [fill=lightgray] (3,4) rectangle ++ (1,1);
\filldraw [fill=lightgray] (4,4) rectangle ++ (1,1);
\filldraw [fill=lightgray] (5,4) rectangle ++ (1,1);
\filldraw [fill=lightgray] (6,4) rectangle ++ (1,1);
\filldraw [fill=lightgray] (3,3) rectangle ++ (1,1);
\filldraw [fill=lightgray] (4,3) rectangle ++ (1,1);
\filldraw [fill=lightgray] (5,3) rectangle ++ (1,1);
\filldraw [fill=lightgray] (6,3) rectangle ++ (1,1);
\filldraw [fill=lightgray] (3,2) rectangle ++ (1,1);
\filldraw [fill=lightgray] (4,2) rectangle ++ (1,1);
\filldraw [fill=lightgray] (3,1) rectangle ++ (1,1);
\draw [line width = 1] (0,0) -- ++ (7, 0);
\draw [line width = 1] (0,1) -- ++ (7, 0);
\draw [line width = 1] (0,2) -- ++ (7, 0);
\draw [line width = 1] (0,3) -- ++ (7, 0);
\draw [line width = 1] (0,4) -- ++ (7, 0);
\draw [line width = 1] (0,5) -- ++ (7, 0);
\draw [line width = 1] (0,6) -- ++ (7, 0);
\draw [line width = 1] (0,7) -- ++ (7, 0);
\draw [line width = 1] (0,0) -- ++ (0, 7);
\draw [line width = 1] (1,0) -- ++ (0, 7);
\draw [line width = 1] (2,0) -- ++ (0, 7);
\draw [line width = 1, double] (3,0) -- ++ (0, 7);
\draw [line width = 1] (4,0) -- ++ (0, 7);
\draw [line width = 1, double] (5,0) -- ++ (0, 7);
\draw [line width = 1] (6,0) -- ++ (0, 7);
\draw [line width = 1] (7,0) -- ++ (0, 7);
\end{tikzpicture}}\longleftrightarrow  
	\raisebox{-0.3\height}{\begin{tikzpicture}[line cap=round,line join=round,>=triangle 45,x=0.2cm,y=0.2cm]
		\draw (0,0)--(0,1);
		% bords de l'arbre
		\draw (0,1)--(-4,5);
		\draw (0,1)--(4,5);
		% bras droit
		\draw (2,3) -- (0,5);
		\draw (3,4) -- (2,5);
		\draw (3,4) -- (3,5);
		\draw (3,4) -- (4,5);
		%bras gauche
		\draw (-3.5,4.5) -- (-3,5);
		\draw (-2.5,3.5) -- (-1,5);
		\draw (-1.5,4.5) -- (-2,5);
		%coupe
		\draw (-2.5,3) -- (-1.5,3);
		\draw (2,3.5) -- (3,3.5);
		% marquages des sommets
		\draw[color=black, fill=white] (-4,5) circle (0.15);
		\draw[color=black, fill=white] (-2,5) circle (0.15);
		\draw[color=black, fill=white] (0,5) circle (0.15);
		\draw[color=black, fill=white] (2,5) circle (0.15);
		\draw[color=black, fill=white] (3,5) circle (0.15);
		\draw[color=black, fill=white] (4,5) circle (0.15);
		\draw[color=black, fill=white] (-3,5) circle (0.15);
		\draw[color=black, fill=white] (-1,5) circle (0.15);
		\filldraw (0,1) circle(0.15);
		\filldraw (-3.5,4.5) circle(0.15);
		\filldraw (-2.5,3.5) circle(0.15);
		\filldraw (-1.5,4.5) circle(0.15);
		\filldraw (3,4) circle(0.15);
		\filldraw (2,3) circle(0.15);
		\end{tikzpicture}}, R^c(t) \leftrightarrow \raisebox{-0.5\height}{\begin{tikzpicture}[x=1em,y=1em]
\filldraw [fill=lightgray] (0,2) rectangle ++ (1,1);
\filldraw [fill=lightgray] (1,2) rectangle ++ (1,1);
\filldraw [fill=lightgray] (0,1) rectangle ++ (1,1);
\draw [line width = 1] (0,0) -- ++ (2, 0);
\draw [line width = 1] (0,1) -- ++ (2, 0);
\draw [line width = 1] (0,2) -- ++ (2, 0);
\draw [line width = 1] (0,3) -- ++ (2, 0);
\draw [line width = 1] (0,0) -- ++ (0, 3);
\draw [line width = 1] (1,0) -- ++ (0, 3);
\draw [line width = 1] (2,0) -- ++ (0, 3);
\end{tikzpicture}}, P_1^c(t) \leftrightarrow \raisebox{-0.5\height}{\begin{tikzpicture}[x=1em,y=1em]
\filldraw [fill=lightgray] (0,3) rectangle ++ (1,1);
\filldraw [fill=lightgray] (1,3) rectangle ++ (1,1);
\filldraw [fill=lightgray] (2,3) rectangle ++ (1,1);
\filldraw [fill=lightgray] (1,2) rectangle ++ (1,1);
\filldraw [fill=lightgray] (2,2) rectangle ++ (1,1);
\filldraw [fill=lightgray] (1,1) rectangle ++ (1,1);
\draw [line width = 1] (0,0) -- ++ (3, 0);
\draw [line width = 1] (0,1) -- ++ (3, 0);
\draw [line width = 1] (0,2) -- ++ (3, 0);
\draw [line width = 1] (0,3) -- ++ (3, 0);
\draw [line width = 1] (0,4) -- ++ (3, 0);
\draw [line width = 1] (0,0) -- ++ (0, 4);
\draw [line width = 1] (1,0) -- ++ (0, 4);
\draw [line width = 1] (2,0) -- ++ (0, 4);
\draw [line width = 1] (3,0) -- ++ (0, 4);
\end{tikzpicture}}
, P^c_2(t)\leftrightarrow \raisebox{-0.5\height}{\begin{tikzpicture}[x=1em,y=1em]
\filldraw [fill=lightgray] (0,1) rectangle ++ (1,1);
\filldraw [fill=lightgray] (1,1) rectangle ++ (1,1);
\draw [line width = 1] (0,0) -- ++ (2, 0);
\draw [line width = 1] (0,1) -- ++ (2, 0);
\draw [line width = 1] (0,2) -- ++ (2, 0);
\draw [line width = 1] (0,0) -- ++ (0, 2);
\draw [line width = 1] (1,0) -- ++ (0, 2);
\draw [line width = 1] (2,0) -- ++ (0, 2);
\end{tikzpicture}}
	\]
\end{exam} 

\subsection{Links with packed words}

\subsubsection{Description of left priority packed words}

\begin{defi}
	We define $\LPPW$ the subset of the set of packed words $\PW$ called \emph{Left Priority Packed Word}. This set is generated by the empty packed word, denoted $\ew$, with a construction rule $c$ of any arity greater than $2$. It is defined for any $w_1,\dots, w_k\in\LPPW, k\geq 2$ and putting $N=\sum_{i=1}^{k} \max(w_i)+1$ by:
	\[
	c(w_1,\dots,w_k) = w_1~N~ \left(w_2 + \max(w_1)\right)~N \dots N~ \left(w_k+\sum_{i=1}^{k-1} \max(w_i)\right).
	\]
\end{defi}
	\begin{exam}
		Consider $w_1=132$ and $w_2=211$, hence $c(w_1,w_2)=1326544$.
	\end{exam}
	\begin{defi}
	    Let $t\in\Schl, n=h(t)$ and consider $(a_i)_{i\in\IEM{1}{k}}$ the total ordering of $A(t)$. We define a map $\tilde{W}:\Schl \rightarrow \PW$ mapping $t$ to a packed word $w=w_1\dots w_k$ where $w_i$ is $n-l(v_i)$ where $v_i$ is the vertex of the angle $a_i.$ We also define $W=\tilde{W}\circ s.$
	\end{defi}
\begin{Rq}
	The map $\tilde{W}$ is the \emph{snake walk} of a levelled \Schw{} tree.
\end{Rq}

\begin{exam}
	For instance:
	\[
	\tilde{W}\left(\raisebox{-0.5\height}{\begin{tikzpicture}[x=0.3cm, y=0.3cm]
			%structure de l'arbre
			\draw (0,0) -- (0,1);
			\draw (0,1) -- (3,4);
			\draw (0,1) -- (-3,4);
			\draw (-2.5,3.5) -- (-2,4);
			\draw (1,2) -- (-1,4);
			\draw (1,2) -- (1,4);
			\draw (-0.5,3.5) -- (0,4);	
			%les cercles
			\filldraw[color=black,fill=white] (3,4) circle (0.15);
			\filldraw[color=black,fill=white] (1,4) circle (0.15);
			\filldraw[color=black,fill=white] (0,4) circle (0.15);
			\filldraw[color=black,fill=white] (-1,4) circle (0.15);
			\filldraw[color=black,fill=white] (-2,4) circle (0.15);
			\filldraw[color=black,fill=white] (-3,4) circle (0.15);
			\filldraw (0,1) circle (0.15);
			\filldraw (-2.5,3.5) circle (0.15);
			\filldraw (1,2) circle (0.15);
			\filldraw (-0.5,3.5) circle (0.15);
			%niveaux
			\draw[dashed] (-3,1)--(3,1);
			\draw[dashed] (-3,2)--(3,2);
			\draw[dashed] (-3,3.5)--(3,3.5);
	\end{tikzpicture}}\right) = 13122 \text{ and }
W\left(\raisebox{-0.5\height}{\begin{tikzpicture}[x=0.3cm, y=0.3cm]
		%structure de l'arbre
		\draw (0,0) -- (0,1);
		\draw (0,1) -- (3,4);
		\draw (0,1) -- (-2.5,3.5);
		\draw (-2.5,3.5) -- (-2.5,4);
		\draw (-2.5,4) -- (-2,5);
		\draw (-2.5,4) -- (-3,5);
		\draw (1,2) -- (-1,4);
		\draw (1,2) -- (1,4);
		\draw (-0.5,3.5) -- (0,4);
		%last level
		\draw (0,4)--(0,5);
		\draw (-1,4)--(-1,5);
		\draw (1,4)--(1,5);
		\draw (3,4)--(3,5);
		%les cercles
		\filldraw[color=black,fill=white] (3,5) circle (0.15);
		\filldraw[color=black,fill=white] (1,5) circle (0.15);
		\filldraw[color=black,fill=white] (0,5) circle (0.15);
		\filldraw[color=black,fill=white] (-1,5) circle (0.15);
		\filldraw[color=black,fill=white] (-2,5) circle (0.15);
		\filldraw[color=black,fill=white] (-3,5) circle (0.15);
		\filldraw (0,1) circle (0.15);
		\filldraw (-2.5,4.) circle (0.15);
		\filldraw (1,2) circle (0.15);
		\filldraw (-0.5,3.5) circle (0.15);
\end{tikzpicture}}\right) = 14233.
	\]
\end{exam}

It turns out \Schw{} trees are in bijection with $\LPPW$:
\begin{prop}
	The languages $\left(\unitree, \vee\right)$ and $\left(\ew, c\right)$ are isomorphic. Hence,  $W$ is a bijection between $\LPPW$ and	 $\Sch$ .
\end{prop}
\begin{proof}
	One just needs to see that $\vee$ and $c$ are both injective construction rules.
\end{proof}

\subsubsection{Enumeration of prunings in $\LPPW$} \label{sec:cuts}

%\begin{defi}
%	We define \emph{prunings} of an element $w\in\LPPW$ denoted $P(w)$ using induction over $\LPPW$:
%	\begin{align*}
%		P(\emptyset)=\{\emptyset\}, && P(c(w_1,\dots,w_k))=P(w_1)\times P(w_2)\times \dots \times P(w_k) \cup \{ \tot \}. 
%	\end{align*}
%	where $\tot$ symbolizes the total cut.  
%\end{defi}
%This corresponds to the natural prunings for trees, however this definition is not really kind to use. To get a nicer definition, we enumerate single cuts of left priority packed word.
\begin{defi}
	A \emph{single cut} of $w\in\LPPW$ is a factor of $w$ (i.e a sequence of consecutive letters of $w$) where adjacent letters of this factor are strictly greater than all the letters of this factor.
\end{defi}

Noticing that a single cut deletes an interval of angles, we deduce a procedure to enumerate every single cuts of a tree $t$ with $k$ angles given $W(t)=w\in\LPPW$:
\begin{itemize}
	\item for $i\in\IEM1k$, get the smaller factor $f$ of $w$ containing $w_i$;
	\item from this factor, get the angles associated to the position of the original letters in $w$;
	\item this factor $f$ corresponds to the \emph{unique} single cut of $t$ withdrawing this list of angles.
\end{itemize}

\begin{Rq}
    We do not prove rigorously, the correspondence between the single cuts for left priority packed word and the single cuts of trees for reasons of brevity. 
\end{Rq}

\begin{exam}
	Consider the Schroeder tree of example~\ref{exam:cut_tree} whose image by $W$ is $1326544$.
%	\[
%	C_0 = \emptyset, \ C_1 = \{4\}, \ C_2 = \{4,5\}, \ C_3 = \{4,5,6,7\},\  C_4 = \{2,4,5,6,7\}, \ C_5 = \IEM{2}7, \ C_6=\IEM{1}7
%	\]
%	with complements 
%	\[
%	C'_0 = \IEM17, \ C'_1 = \IEM13\cup\IEM57, \ C'_2 = \IEM13\cup \IEM67,\ C'_3 = \IEM13, C'_4= \{1,3\}\,\  C'_5 = \{1\},\ C_6=\emptyset
%	\]
	Hence, applying the procedure detailed above, one gets $6$ simple cuts according to the next table
	\[
	\begin{tblr}{colspec={cccccccc}, vlines,
		vline{8}={white},
		cell{4}{7}={r=1,c=2}{c}}
		\text{Angle} & 1 & 2 & 3 & 4 & 5 & 6 & 7 \\
		\text{Factors} & 1 & 132 & 2 &1326544& 544& 44 &44 \\
		\text{Cut angles} & \{1\} & \{1,2,3\} & \{3\}& \IEM{1}{7} & \{5,6,7\}& \{6,7\} & \{6,7\} \\
		\text{Single cut} &
		\raisebox{-0.5\height}{\begin{tikzpicture}[line cap=round,line join=round,>=triangle 45,x=0.175cm,y=0.175cm]
				%cut of 1
			\draw (0,0)--(0,1);
			% bords de l'arbre
			\draw (0,1)--(-4,5);
			\draw (0,1)--(4,5);
			% bras droit
			\draw (2,3) -- (0,5);
			\draw (3,4) -- (2,5);
			\draw (3,4) -- (3,5);
			\draw (3,4) -- (4,5);
			%bras gauche
			\draw (-3.5,4.5) -- (-3,5);
			\draw (-2.5,3.5) -- (-1,5);
			\draw (-1.5,4.5) -- (-2,5);
			%coupe
			\draw (-3.5,4) -- (-2.5,4);
			% marquages des sommets
			\draw[color=black, fill=white] (-4,5) circle (0.15);
			\draw[color=black, fill=white] (-2,5) circle (0.15);
			\draw[color=black, fill=white] (0,5) circle (0.15);
			\draw[color=black, fill=white] (2,5) circle (0.15);
			\draw[color=black, fill=white] (3,5) circle (0.15);
			\draw[color=black, fill=white] (4,5) circle (0.15);
			\draw[color=black, fill=white] (-3,5) circle (0.15);
			\draw[color=black, fill=white] (-1,5) circle (0.15);
			\filldraw (0,1) circle(0.15);
			\filldraw (-3.5,4.5) circle(0.15);
			\filldraw (-2.5,3.5) circle(0.15);
			\filldraw (-1.5,4.5) circle(0.15);
			\filldraw (3,4) circle(0.15);
			\filldraw (2,3) circle(0.15);
		\end{tikzpicture}}
		&
		\raisebox{-0.5\height}{\begin{tikzpicture}[line cap=round,line join=round,>=triangle 45,x=0.175cm,y=0.175cm]
				%cut of 123
			\draw (0,0)--(0,1);
			% bords de l'arbre
			\draw (0,1)--(-4,5);
			\draw (0,1)--(4,5);
			% bras droit
			\draw (2,3) -- (0,5);
			\draw (3,4) -- (2,5);
			\draw (3,4) -- (3,5);
			\draw (3,4) -- (4,5);
			%bras gauche
			\draw (-3.5,4.5) -- (-3,5);
			\draw (-2.5,3.5) -- (-1,5);
			\draw (-1.5,4.5) -- (-2,5);
			%coupe
			\draw (-2.1,2.5) -- (-1.1,2.5);
			% marquages des sommets
			\draw[color=black, fill=white] (-4,5) circle (0.15);
			\draw[color=black, fill=white] (-2,5) circle (0.15);
			\draw[color=black, fill=white] (0,5) circle (0.15);
			\draw[color=black, fill=white] (2,5) circle (0.15);
			\draw[color=black, fill=white] (3,5) circle (0.15);
			\draw[color=black, fill=white] (4,5) circle (0.15);
			\draw[color=black, fill=white] (-3,5) circle (0.15);
			\draw[color=black, fill=white] (-1,5) circle (0.15);
			\filldraw (0,1) circle(0.15);
			\filldraw (-3.5,4.5) circle(0.15);
			\filldraw (-2.5,3.5) circle(0.15);
			\filldraw (-1.5,4.5) circle(0.15);
			\filldraw (3,4) circle(0.15);
			\filldraw (2,3) circle(0.15);
		\end{tikzpicture}}	&  
	\raisebox{-0.5\height}{\begin{tikzpicture}[line cap=round,line join=round,>=triangle 45,x=0.175cm,y=0.175cm]
			%cut of 3
	\draw (0,0)--(0,1);
	% bords de l'arbre
	\draw (0,1)--(-4,5);
	\draw (0,1)--(4,5);
	% bras droit
	\draw (2,3) -- (0,5);
	\draw (3,4) -- (2,5);
	\draw (3,4) -- (3,5);
	\draw (3,4) -- (4,5);
	%bras gauche
	\draw (-3.5,4.5) -- (-3,5);
	\draw (-2.5,3.5) -- (-1,5);
	\draw (-1.5,4.5) -- (-2,5);
	%coupe
	\draw (-2.5,4) -- (-1.5,4);
	% marquages des sommets
	\draw[color=black, fill=white] (-4,5) circle (0.15);
	\draw[color=black, fill=white] (-2,5) circle (0.15);
	\draw[color=black, fill=white] (0,5) circle (0.15);
	\draw[color=black, fill=white] (2,5) circle (0.15);
	\draw[color=black, fill=white] (3,5) circle (0.15);
	\draw[color=black, fill=white] (4,5) circle (0.15);
	\draw[color=black, fill=white] (-3,5) circle (0.15);
	\draw[color=black, fill=white] (-1,5) circle (0.15);
	\filldraw (0,1) circle(0.15);
	\filldraw (-3.5,4.5) circle(0.15);
	\filldraw (-2.5,3.5) circle(0.15);
	\filldraw (-1.5,4.5) circle(0.15);
	\filldraw (3,4) circle(0.15);
	\filldraw (2,3) circle(0.15);
\end{tikzpicture}}
			&
			\raisebox{-0.5\height}{\begin{tikzpicture}[line cap=round,line join=round,>=triangle 45,x=0.175cm,y=0.175cm]
					%cut of all
					\draw (0,0)--(0,1);
					% bords de l'arbre
					\draw (0,1)--(-4,5);
					\draw (0,1)--(4,5);
					% bras droit
					\draw (2,3) -- (0,5);
					\draw (3,4) -- (2,5);
					\draw (3,4) -- (3,5);
					\draw (3,4) -- (4,5);
					%bras gauche
					\draw (-3.5,4.5) -- (-3,5);
					\draw (-2.5,3.5) -- (-1,5);
					\draw (-1.5,4.5) -- (-2,5);
					%coupe
					\draw (-0.5,0.5) -- (0.5,0.5);
					% marquages des sommets
					\draw[color=black, fill=white] (-4,5) circle (0.15);
					\draw[color=black, fill=white] (-2,5) circle (0.15);
					\draw[color=black, fill=white] (0,5) circle (0.15);
					\draw[color=black, fill=white] (2,5) circle (0.15);
					\draw[color=black, fill=white] (3,5) circle (0.15);
					\draw[color=black, fill=white] (4,5) circle (0.15);
					\draw[color=black, fill=white] (-3,5) circle (0.15);
					\draw[color=black, fill=white] (-1,5) circle (0.15);
					\filldraw (0,1) circle(0.15);
					\filldraw (-3.5,4.5) circle(0.15);
					\filldraw (-2.5,3.5) circle(0.15);
					\filldraw (-1.5,4.5) circle(0.15);
					\filldraw (3,4) circle(0.15);
					\filldraw (2,3) circle(0.15);
			\end{tikzpicture}}	&
		\raisebox{-0.5\height}{\begin{tikzpicture}[line cap=round,line join=round,>=triangle 45,x=0.175cm,y=0.175cm]
				%cut of 567
				\draw (0,0)--(0,1);
				% bords de l'arbre
				\draw (0,1)--(-4,5);
				\draw (0,1)--(4,5);
				% bras droit
				\draw (2,3) -- (0,5);
				\draw (3,4) -- (2,5);
				\draw (3,4) -- (3,5);
				\draw (3,4) -- (4,5);
				%bras gauche
				\draw (-3.5,4.5) -- (-3,5);
				\draw (-2.5,3.5) -- (-1,5);
				\draw (-1.5,4.5) -- (-2,5);
				%coupe
				\draw (0.5,2) -- (1.5,2);
				% marquages des sommets
				\draw[color=black, fill=white] (-4,5) circle (0.15);
				\draw[color=black, fill=white] (-2,5) circle (0.15);
				\draw[color=black, fill=white] (0,5) circle (0.15);
				\draw[color=black, fill=white] (2,5) circle (0.15);
				\draw[color=black, fill=white] (3,5) circle (0.15);
				\draw[color=black, fill=white] (4,5) circle (0.15);
				\draw[color=black, fill=white] (-3,5) circle (0.15);
				\draw[color=black, fill=white] (-1,5) circle (0.15);
				\filldraw (0,1) circle(0.15);
				\filldraw (-3.5,4.5) circle(0.15);
				\filldraw (-2.5,3.5) circle(0.15);
				\filldraw (-1.5,4.5) circle(0.15);
				\filldraw (3,4) circle(0.15);
				\filldraw (2,3) circle(0.15);
		\end{tikzpicture}} 
			&
			\raisebox{-0.5\height}{\begin{tikzpicture}[line cap=round,line join=round,>=triangle 45,x=0.175cm,y=0.175cm]
					%cut of 67
					\draw (0,0)--(0,1);
					% bords de l'arbre
					\draw (0,1)--(-4,5);
					\draw (0,1)--(4,5);
					% bras droit
					\draw (2,3) -- (0,5);
					\draw (3,4) -- (2,5);
					\draw (3,4) -- (3,5);
					\draw (3,4) -- (4,5);
					%bras gauche
					\draw (-3.5,4.5) -- (-3,5);
					\draw (-2.5,3.5) -- (-1,5);
					\draw (-1.5,4.5) -- (-2,5);
					%coupe
					\draw (2,3.5) -- (3,3.5);
					% marquages des sommets
					\draw[color=black, fill=white] (-4,5) circle (0.15);
					\draw[color=black, fill=white] (-2,5) circle (0.15);
					\draw[color=black, fill=white] (0,5) circle (0.15);
					\draw[color=black, fill=white] (2,5) circle (0.15);
					\draw[color=black, fill=white] (3,5) circle (0.15);
					\draw[color=black, fill=white] (4,5) circle (0.15);
					\draw[color=black, fill=white] (-3,5) circle (0.15);
					\draw[color=black, fill=white] (-1,5) circle (0.15);
					\filldraw (0,1) circle(0.15);
					\filldraw (-3.5,4.5) circle(0.15);
					\filldraw (-2.5,3.5) circle(0.15);
					\filldraw (-1.5,4.5) circle(0.15);
					\filldraw (3,4) circle(0.15);
					\filldraw (2,3) circle(0.15);
			\end{tikzpicture}}
	\end{tblr}
	\]
\end{exam}

It defines properly how to do a pruning of an element of $\LPPW$ as a pruning is a set of non-intersecting single cuts. 

\section{Algorithms}\label{sec:algorithms}

In this section, we give a description of the algorithms we used to implement the problem in {\ttfamily C++} to get a faster execution code than one done with {\ttfamily Python} or {\ttfamily \Sage{}}. This code is available online on GitHUb~\cite{github}.  

\subsection{Quasi-shuffle enumeration}

In the code available on GitHub~\cite{github}, one can find a procedure \lstinline[language=C++]|Quasishuffle| that returns the list of all quasi-shuffle of a given type $(n,m)$. For this, we use the result from lemma~\ref{lem:Indution_batc} to get an inductive algorithm computing all quasi-shuffles of a given type.

\subsection{The class \Schw{} tree} \label{ss:schroeder_tree}

In an algorithmic context, a \Schw{} tree $t$ will be coded as an object \texttt{t} with data
\begin{itemize}
    \item an integer $\texttt{t.}n$ for the number of angles of $t$;
    \item an integer $\texttt{t.}h$ for the height of $t$;
    \item an array $\texttt{t.c}$, where $\texttt{t.c}[i]$ is for the set $C_i$ of the init chain $\varphi_n(t)$.
\end{itemize}

This defines the class of \Schw{} trees endowed with basics functions needed to perform any computation involving the product and coproduct of the free tridendriform algebra of \Schw{} trees.

 First, algorithm~\ref{A:GetForest} {\ttfamily GetForest} enables us to get back the code of a forest given its angle's list comparing a line with the one above to see if the forest code is finished.

\begin{algorithm}[h]
	\small
	\caption{\small Return the forest of \texttt{t} rooted at level $l$ and obtained with angles in $\IEM{a_l}{a_r}$.}
	\label{A:GetForest}
	\begin{algorithmic}[1]
		\Function{GetForest}{\texttt{t}, $l$, $a_l
			$, $a_r$}
		\State $\tt{c} \gets$ array of sets with capacity $\texttt{t.}h + 1 - l$
		\State $h \gets -1$
		\State $A \gets \IEM{a_l}{a_r}$
		\For {$j$ from $l$ to $\texttt{t.}h$}
		\State $X\gets -(a_l-1) + (\texttt{t.c}[j]\cap A)$
		\If{$h = -1$ or $X \not= \texttt{c}[h]$}
		\State $h \gets h + 1$
		\State $\texttt{c}[h] \gets X$
		\EndIf
		\EndFor
		\State \textbf{return} Forest \texttt{f} with $\texttt{f.}n = \text{card}(A)$, $\texttt{f.}h= h$ and $\texttt{f.c} = \texttt{c}$.
		\EndFunction
	\end{algorithmic}
\end{algorithm}
Second, algorithms~\ref{A:LeftComb} and~\ref{A:RightComb} give the respective left/right comb structures of a tree from its code. They are reading the comb forests from the root of the tree to its leftmost/rightmost leaf.  Hence, those algorithms give the same output as in examples~\ref{exam:left_comb_decomposition} and~\ref{exam:right_comb_decomposition} without the extra $\unitree$.

\begin{algorithm}[h]
\small
\caption{\small Return the left comb decomposition of the tree \texttt{t}}\label{A:LeftComb}
\begin{algorithmic}[1]
\Function{LeftCombDecomposition}{\texttt{t}}
\State $L \gets $ empty list
\State $A \gets \IEM{1}{\texttt{t.}n}$
\For {$l$ from $1$ to $\texttt{t.}h$}
\State $X \gets \texttt{t.c}[l]\cap A$
\If {$X \not= \emptyset$}
\State $F \gets \textsc{GetForest}(\texttt{t}, l, \min(X) + 1, \max(A))$
\State Add $F$ at begin of $L$.
\State $A \gets \IEM{1}{\min(X) - 1}$
\EndIf
\EndFor
  \State \textbf{return} $L$
\EndFunction
\end{algorithmic}
\end{algorithm}

\begin{algorithm}[h]
\small
\caption{\small Return the right comb decomposition of the tree \texttt{t}}\label{A:RightComb}
\begin{algorithmic}[1]
\Function{RightCombDecomposition}{\texttt{t}}
\State $R \gets $ empty list
\State $A \gets \IEM{1}{\texttt{t.}n}$
\For {$l$ from $1$ to $\texttt{t.}h$}
\State $X \gets \texttt{t.c}[i]\cap A$
\If {$X \not= \emptyset$}
\State $F \gets \textsc{GetForest}(\texttt{t}, l, \min(A), \max(X) - 1)$
\State Add $F$ at end of $R$.
\State $A \gets \IEM{\max(X) + 1}{\texttt{t}.n}$
\EndIf
\EndFor
  \State \textbf{return} $R$
\EndFunction
\end{algorithmic}
\end{algorithm}

Fiannly, one of the most important algorithm of this work is {\ttfamily AtomicProduct}. It describes exactly the procedure described in definition~\ref{defi:operations_codes}. For this purpose, we will read grids from the bottom to the top. Variables $h_t$ and $h_s$ are the current line of the grid of respectively {\ttfamily t} or {\ttfamily s} that we should read next, meanwhile $h$ is the current line in which we should write in the result grid {\ttfamily c}. Lines 24-27 is the copy of the code of {\ttfamily t} at the left top corner of {\ttfamily c}. Lines 9-21 is managing the different cases detailed in definition~\ref{defi:quasishuffle_code}. We apply this algorithm on the code in example~\ref{exam:algo_atomic_product} and we add the current height for copying the code of respectively $t$ and $s$, $h_t$ and $h_s$ with its different regimes.

\begin{algorithm}[h]
\small
\caption{\small Return atomic product of two Schroeder tree for a given quasi shuffle.}
\label{A:AtomicProduct}
\begin{algorithmic}[1]
\Function{AtomicProduct}{$\texttt{t}$, $\texttt{s}$, $\sigma$}
\State $L \gets \textsc{RightCombDecomposition}(\texttt{t})$ \Comment{Indexed from $1$}
\State $R \gets
\textsc{LeftCombDecomposition}(\texttt{s})$ \Comment{Indexed from $1$}
\State $k \gets$ size of $L$;  $l \gets$ size of $R$; $r\gets \max(\sigma)$
\State $h_t \gets 0$; $h_s \gets 0$
\State $f_L \gets 1$; $f_R \gets 1$
\State $\texttt{c}\gets$ array of size $(\texttt{t.}h+\texttt{s.}h + r - k - l)$ \Comment{Indexed from $0$}
\State $h\gets 0$; $\texttt{c}[0] = \emptyset$
\For {$i_r$ from $1$ to $r$}
\BeginBox
\Statex \hspace{1cm}(L) Behaviour of a left grafting
\If {$\sigma(k + f_R) \not= i_r$ }  \Comment{Case $\sigma^{-1}(i_r)= \{f_L\}$}
\State $h_t\gets h_t+1$; $f_L \gets \min(f_L + 1, k)$
\State $h\gets h+1$; $\texttt{c}[h] = \texttt{t.c}[h_t] \sqcup (\texttt{t.}n + \texttt{s.}c[h_s])$
\EndBox
\Else\Comment{We have $\sigma(k + f_R) = i_r$}
\BeginBox 
\Statex \hspace{0.5cm} (L+R) Behaviour for a simultaneous left and right graftings
\If{$\sigma(f_L) = i_r$ } \Comment{Case $\sigma^{-1}(i_r)= \{f_L, k + f_R\}$}
\State $h_t\gets h_t+1$; $f_L \gets \min(f_L + 1, k)$
\EndBox
\EndIf
\BeginBox
\Statex \hspace{1.5cm} (R) Behaviour for a right grafting
\For{$j$ from $0$ to $L[f_L].h$} \Comment{Insert code of $L[f_L]$}
\State $h_s \gets h_s + 1$
\State $h\gets h+1$; $\texttt{c}[h] = \texttt{t.c}[h_t] \sqcup (\texttt{t.}n + \texttt{s.}c[h_s])$
\EndFor
\State $f_R \gets \min(f_R +1, l)$
\EndBox
\EndIf
\EndFor

\BeginBox
\Statex  \hspace{1.5cm} (I) Initial behaviour of the algorithm
\While {$h_t < \texttt{t}.h$}\Comment{Copy forests from \texttt{t}}
\State $h_t \gets h_t + 1$
\State $h\gets h+1$; $\texttt{c}[h] = \texttt{t.c}[h_t] \sqcup (\texttt{t.}n + \texttt{s.}c[h_s])$\Comment{Here $h_s = \texttt{s.}h$}
\EndBox
\EndWhile
\State \textbf{return} Tree \texttt{p} with $\texttt{p.}n = \texttt{t.}{n} + \texttt{s.}{n}$, $\texttt{p.}h= h$ and $\texttt{f.c} = \texttt{c}$.
\EndFunction
\end{algorithmic}
\end{algorithm}

\begin{exam}\label{exam:algo_atomic_product}
We detail the application of algorithm~\ref{A:AtomicProduct} onto the product of $T$ and $S$ given below:
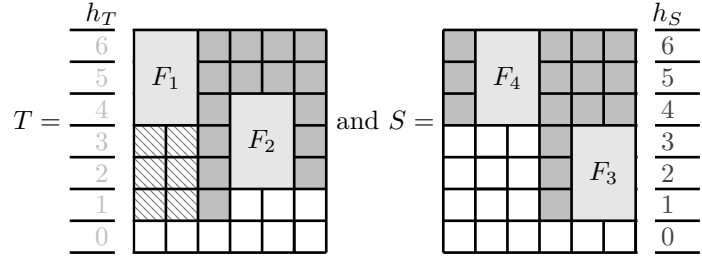
\[
T=\raisebox{-0.3\height}{\begin{tikzpicture}[x=1.2em,y=1.2em,baseline=2.5em]
		\filldraw [fill=lightgray!40] (0,4) rectangle ++ (2,3);
		\filldraw [fill=lightgray!40] (3,2) rectangle ++ (2,3);
		\filldraw [fill=lightgray] (2,5) rectangle ++ (3,2);
		\filldraw [fill=lightgray] (2,1) rectangle ++ (1,4);
		\filldraw [fill=lightgray] (5,2) rectangle ++ (1,5);
		\draw[pattern=north west lines, pattern color=gray] (0,1) rectangle ++ (2,3);
		%horizontal
		\draw [line width = 1] (-2,7) -- ++ (8, 0);
		\draw [line width = 1] (2,6) -- ++ (4, 0);
		\draw [line width = 1] (2,4) -- ++ (1, 0);
		\draw [line width = 1] (5,4) -- ++ (1, 0);
		\draw [line width = 1] (-2,4) -- ++ (4, 0);
		\draw [line width = 1] (2,5) -- ++ (4, 0);
		\draw [line width = 1] (-2,3) -- ++ (5, 0);
		\draw [line width = 1] (-2,2) -- ++ (8, 0);
		\draw [line width = 1] (-2,1) -- ++ (8, 0);
		\draw [line width = 1] (-2,0) -- ++ (8, 0);
		\draw [line width = 1] (5,3) -- ++ (1, 0);
		\draw [line width = 1] (-2,5) -- ++ (2, 0);
		\draw [line width = 1] (-2,6) -- ++ (2, 0);
		%vertical
		\draw [line width = 1] (0,0) -- ++ (0, 7);
		\draw [line width = 1] (2,0) -- ++ (0, 7);
		\draw [line width = 1] (3,0) -- ++ (0, 7);	
		\draw [line width = 1] (5, 0) -- ++ (0, 7);
		\draw [line width = 1] (6, 0) -- ++ (0, 7);
		\draw [line width = 1] (4, 0) -- ++ (0, 2);
		\draw [line width = 1] (5, 0) -- ++ (0, 1);
		\draw [line width = 1] (4, 5) -- ++ (0, 2);
		\draw [line width = 1] (1,0) -- ++ (0, 4);
		\draw(1,5.5) node{${F_1}$};
		\draw(4,3.5) node{${F_2}$};
		% couvrement blanc
		\filldraw [fill=white, color=white] (-0.575,-0.2) rectangle ++(0.52,7.4);
		% Numéros de lignes
		\draw (-1,7.5) node{\textcolor{black}{$h_T$}};
		\draw (-1,0.5) node{\textcolor{lightgray}{$0$}};
		\draw (-1,1.5) node{\textcolor{lightgray}{$1$}};
		\draw (-1,2.5) node{\textcolor{lightgray}{$2$}};
		\draw (-1,3.5) node{\textcolor{lightgray}{$3$}};
		\draw (-1,4.5) node{\textcolor{lightgray}{$4$}};
		\draw (-1,5.5) node{\textcolor{lightgray}{$5$}};
		\draw (-1,6.5) node{\textcolor{lightgray}{$6$}};
\end{tikzpicture}}
\text{ and }
S=\raisebox{-0.3\height}{\begin{tikzpicture}[x=1.2em,y=1.2em,baseline=2.5em]
		\filldraw [fill=lightgray!40] (0,4) rectangle ++ (2,3);
		\filldraw [fill=lightgray!40] (3,1) rectangle ++ (2,3);
		\filldraw [fill=lightgray] (2,4) rectangle ++ (3,3);
		\filldraw [fill=lightgray] (2,1) rectangle ++ (1,3);
		\filldraw [fill=lightgray] (-1,4) rectangle ++ (1,3);
		%horizontal
		\draw [line width = 1] (2,6) -- ++ (5, 0);
		\draw [line width = 1] (-1,4) -- ++ (8, 0);
		\draw [line width = 1] (2,5) -- ++ (5, 0);
		\draw [line width = 1] (0,3) -- ++ (1, 0);
		\draw [line width = 1] (0,2) -- ++ (1, 0);
		\draw [line width = 1] (-1,1) -- ++ (8, 0);
		\draw [line width = 1] (-1,0) -- ++ (8, 0);
		\draw [line width = 1] (-1,7) -- ++ (8, 0);
		\draw [line width = 1] (-1,2) -- ++ (4, 0);
		\draw [line width = 1] (-1,3) -- ++ (4, 0);
		\draw [line width = 1] (-1,5) -- ++ (1, 0);
		\draw [line width = 1] (-1,6) -- ++ (1, 0);
		\draw [line width = 1] (5,3) -- ++ (2, 0);
		\draw [line width = 1] (5,2) -- ++ (2, 0);
		%vertical
		\draw [line width = 1] (2,0) -- ++ (0, 7);
		\draw [line width = 1] (3,0) -- ++ (0, 7);	
		\draw [line width = 1] (5, 0) -- ++ (0, 7);
		\draw [line width = 1] (-1, 0) -- ++ (0, 7);
		\draw [line width = 1] (4, 0) -- ++ (0, 1);
		\draw [line width = 1] (5, 0) -- ++ (0, 1);
		\draw [line width = 1] (4, 4) -- ++ (0, 3);
		\draw [line width = 1] (1,0) -- ++ (0, 4);
		\draw(1,5.5) node{${F_4}$};
		\draw(4,2.5) node{${F_3}$};
		% Couvrement blanc
		\filldraw [fill=white, color=white] (5.075,-0.2) rectangle ++(0.5,7.4);
		\draw [line width = 1] (0,0) -- ++ (0, 7);
		% Numéros de lignes
		\draw (6,7.5) node{\textcolor{black}{$h_S$}};
		\draw (6,0.5) node{\textcolor{darkgray}{$0$}};
		\draw (6,1.5) node{\textcolor{darkgray}{$1$}};
		\draw (6,2.5) node{\textcolor{darkgray}{$2$}};
		\draw (6,3.5) node{\textcolor{darkgray}{$3$}};
		\draw (6,4.5) node{\textcolor{darkgray}{$4$}};
		\draw (6,5.5) node{\textcolor{darkgray}{$5$}};
		\draw (6,6.5) node{\textcolor{darkgray}{$6$}};
	\end{tikzpicture}}
	\]
	with the quasishuffle $\sigma=(1,3,2,3)$:
	
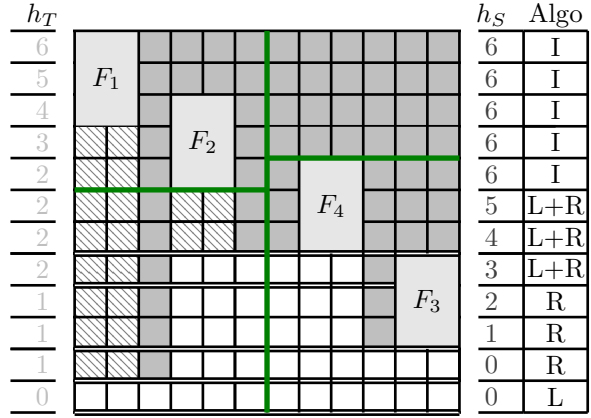
\begin{center}
	\begin{tikzpicture}[x=1.2em,y=1.2em,baseline=2.5em]
		\filldraw [fill=lightgray!40] (0,4) rectangle ++ (2,3);
		\filldraw [fill=lightgray!40] (3,2) rectangle ++ (2,3);
		\filldraw [fill=lightgray] (2,5) rectangle ++ (3,2);
		\filldraw [fill=lightgray] (2,-4) rectangle ++ (1,10);
		\filldraw [fill=lightgray] (5,3) rectangle ++ (7,4);
		\filldraw [fill=lightgray] (5,0) rectangle ++ (1,3);
		\draw[pattern=north west lines, pattern color=gray] (0,-4) rectangle ++ (2,8);
		\draw[pattern=north west lines, pattern color=gray] (3,0) rectangle ++ (2,2);
		%horizontal
		\draw [line width = 1] (-2,7) -- ++ (16, 0);
		\draw [line width = 1] (2,6) -- ++ (12, 0);
		\draw [line width = 1] (2,4) -- ++ (1, 0);
		\draw [line width = 1] (5,4) -- ++ (10, 0);
		\draw [line width = 1] (-2,6) -- ++ (2, 0);
		\draw [line width = 1] (-2,4) -- ++ (2, 0);
		\draw [line width = 1] (2,5) -- ++ (12, 0);
		\draw [line width = 1] (-2,5) -- ++ (2, 0);
		\draw [line width = 1] (-2,3) -- ++ (5, 0);
		\draw [line width = 1] (-2,2) -- ++ (8, 0);
		\draw [line width = 1] (-2,1) -- ++ (8, 0);
		\draw [line width = 1] (-2,0) -- ++ (16, 0);
		\draw [line width = 1] (5,3) -- ++ (9, 0);
		%vertical
		\draw [line width = 1] (0,-5) -- ++ (0, 12);
		\draw [line width = 1] (2,-5) -- ++ (0, 12);
		\draw [line width = 1] (3,-5) -- ++ (0, 12);
		\draw [line width = 1] (5, -5) -- ++ (0, 12);
		\draw [line width = 1] (6, -5) -- ++ (0, 12);
		\draw [line width = 1] (4, -5) -- ++ (0, 7);
		\draw [line width = 1] (5, -5) -- ++ (0, 6);
		\draw [line width = 1] (4, 5) -- ++ (0, 2);
		\draw [line width = 1] (1,-5) -- ++ (0, 9);
		\draw(1,5.5) node{${F_1}$};
		\draw(4,3.5) node{${F_2}$};
		
		\begin{scope}[shift={(7,-4)}]
			\filldraw [fill=lightgray!40] (0,4) rectangle ++ (2,3);
			\filldraw [fill=lightgray!40] (3,1) rectangle ++ (2,3);
			%\filldraw [fill=lightgray] (3,3) rectangle ++ (2,1);
			\filldraw [fill=lightgray] (2,4) rectangle ++ (3,4);
			\filldraw [fill=lightgray] (0,7) rectangle ++ (3,1);
			\filldraw [fill=lightgray] (2,1) rectangle ++ (1,3);
			\filldraw [fill=lightgray] (-1,4) rectangle ++ (1,4);
			%horizontal
			\draw [line width = 1] (-1,7) -- ++ (6, 0);
			\draw [line width = 1] (2,6) -- ++ (5, 0);
			\draw [line width = 1] (2,4) -- ++ (3, 0);
			\draw [line width = 1, double] (-7,4) -- ++ (12, 0);
			\draw [line width = 1] (-9,4) -- ++ (2, 0);
			\draw [line width = 1] (5,4) -- ++ (2, 0);
			\draw [line width = 1] (2,5) -- ++ (5, 0);
			\draw [line width = 1] (0,3) -- ++ (3, 0);
			\draw [line width = 1] (0,2) -- ++ (3, 0);
			\draw [line width = 1, double] (-7,0) -- ++ (12, 0);
			\draw [line width = 1, double] (-7,-1) -- ++ (12, 0);
			\draw [line width = 1] (-9,0) -- ++ (2, 0);
			\draw [line width = 1] (-9,-1) -- ++ (2, 0);
			\draw [line width = 1] (5,0) -- ++ (2, 0);
			\draw [line width = 1] (5,-1) -- ++ (2, 0);
			\draw [line width = 1] (-9,1) -- ++ (12, 0);
			\draw [line width = 1, double] (-7,1) -- ++ (12, 0);
			\draw [line width = 1] (-9,1) -- ++ (2, 0);
			\draw [line width = 1] (5,1) -- ++ (2, 0);
			\draw [line width = 1] (-9,2) -- ++ (9, 0);
			\draw [line width = 1, double] (-7,3) -- ++ (10, 0);
			\draw [line width = 1] (-9,3) -- ++ (2, 0);
			\draw [line width = 1] (-7,5) -- ++ (7, 0);
			\draw [line width = 1] (-7,6) -- ++ (7, 0);
			\draw [line width = 1] (5,2) -- ++ (2, 0);
			\draw [line width = 1] (5,3) -- ++ (2, 0);
			%vertical
			\draw [line width = 1] (0,-1) -- ++ (0, 12);
			\draw [line width = 1] (2,-1) -- ++ (0, 12);
			\draw [line width = 1] (3,-1) -- ++ (0, 12);
			\draw [line width = 1] (5,-1) -- ++ (0, 12);
			\draw [line width = 1] (-1,-1) -- ++ (0, 7);
			\draw [line width = 1] (4,-1) -- ++ (0, 2);
			\draw [line width = 1] (5, 0) -- ++ (0, 1);
			\draw [line width = 1] (4, 4) -- ++ (0, 7);
			\draw [line width = 1] (1,-1) -- ++ (0, 5);
			\draw [line width = 1] (1,7) -- ++ (0, 4);
			\draw(1,5.5) node{${F_4}$};
			\draw(4,2.5) node{${F_3}$};
		\end{scope}
		% recouvrements blancs
		\filldraw [fill=white, color=white] (12.075,-6.2) rectangle ++ (0.5,14.4);
		\filldraw [fill=white, color=white] (-0.575,-6.2) rectangle ++ (0.5,14.4);
		% Numéros de lignes
		\draw (13,-4.5) node{\textcolor{darkgray}{$0$}};
		\draw (13,-3.5) node{\textcolor{darkgray}{$0$}};
		\draw (13,-2.5) node{\textcolor{darkgray}{$1$}};
		\draw (13,-1.5) node{\textcolor{darkgray}{$2$}};
		\draw (13,-0.5) node{\textcolor{darkgray}{$3$}};
		\draw (13,0.5) node{\textcolor{darkgray}{$4$}};
		\draw (13,1.5) node{\textcolor{darkgray}{$5$}};
		\draw (13,2.5) node{\textcolor{darkgray}{$6$}};
		\draw (13,3.5) node{\textcolor{darkgray}{$6$}};
		\draw (13,4.5) node{\textcolor{darkgray}{$6$}};
		\draw (13,5.5) node{\textcolor{darkgray}{$6$}};
		\draw (13,6.5) node{\textcolor{darkgray}{$6$}};
		\draw (13,7.5) node{\textcolor{black}{$h_S$}};
		\draw (-1,7.5) node{\textcolor{black}{$h_T$}};
		\draw (-1,-4.5) node{\textcolor{lightgray}{$0$}};
		\draw (-1,-3.5) node{\textcolor{lightgray}{$1$}};
		\draw (-1,-2.5) node{\textcolor{lightgray}{$1$}};
		\draw (-1,-1.5) node{\textcolor{lightgray}{$1$}};
		\draw (-1,-0.5) node{\textcolor{lightgray}{$2$}};
		\draw (-1,0.5) node{\textcolor{lightgray}{$2$}};
		\draw (-1,1.5) node{\textcolor{lightgray}{$2$}};
		\draw (-1,2.5) node{\textcolor{lightgray}{$2$}};
		\draw (-1,3.5) node{\textcolor{lightgray}{$3$}};
		\draw (-1,4.5) node{\textcolor{lightgray}{$4$}};
		\draw (-1,5.5) node{\textcolor{lightgray}{$5$}};
		\draw (-1,6.5) node{\textcolor{lightgray}{$6$}};
		% étapes de l'algorithme
		\draw [line width=1] (14,-5) --++ (0,12);
		\draw [line width=1] (16,-5) --++ (0,12);
		\draw [line width=1] (14,-5) --++ (2,0);
		\draw [line width=1] (14,-4) --++ (2,0);
		\draw [line width=1] (14,-3) --++ (2,0);
		\draw [line width=1] (14,-2) --++ (2,0);
		\draw [line width=1] (14,-1) --++ (2,0);
		\draw [line width=1] (14,0) --++ (2,0);
		\draw [line width=1] (14,1) --++ (2,0);
		\draw [line width=1] (14,2) --++ (2,0);
		\draw [line width=1] (14,3) --++ (2,0);
		\draw [line width=1] (14,4) --++ (2,0);
		\draw [line width=1] (14,5) --++ (2,0);
		\draw [line width=1] (14,6) --++ (2,0);
		\draw [line width=1] (14,7) --++ (2,0);
		\draw [line width=1] (15,-4.5) node{L};
		\draw [line width=1] (15,-3.5) node{R};
		\draw [line width=1] (15,-2.5) node{R};
		\draw [line width=1] (15,-1.5) node{R};
		\draw [line width=1] (15,-0.5) node{L+R};
		\draw [line width=1] (15,0.5) node{L+R};
		\draw [line width=1] (15,1.5) node{L+R};
		\draw [line width=1] (15,2.5) node{I};
		\draw [line width=1] (15,3.5) node{I};
		\draw [line width=1] (15,4.5) node{I};
		\draw [line width=1] (15,5.5) node{I};
		\draw [line width=1] (15,6.5) node{I};
		\draw [line width=1] (15,7.5) node{Algo};
		% découpage de la grille
		\draw [line width=2, color=green!50!black] (6,-5) --++ (0,12);
		\draw [line width=2, color=green!50!black] (0,2) --++ (6,0);
		\draw [line width=2, color=green!50!black] (6,3) --++ (6,0);
	\end{tikzpicture}
	\end{center}
\end{exam}

\subsection{Algorithms for cuts}

In order to enumerate the list of single cuts of a tree, we use the \emph{left priority packed word} associated to a tree and we use the description of single cuts introduced in section~\ref{sec:coproduct_tree_codes}. Algorithm~\ref{A:PackedWord} gives this word looking at the new black squares appearing between a line and the one above it (line 4--6).

\begin{algorithm}[h]
\small
\caption{\small Return the left priority packed word associated to a Schroeder tree.}
\label{A:PackedWord}
\begin{algorithmic}[1]
\Function{PackedWord}{$\texttt{t}$}
\State \texttt{w} $\gets$ array of size $\texttt{t.}n$ \Comment{Indexed from $1$}
\For {$l$ from $0$ to $\texttt{t}.h - 1$}
\State $D \gets \texttt{t.c}[l + 1] \setminus \texttt{t.c}[l]$
\For {$d\in D$}
\State $\texttt{w}[d] \gets \texttt{t}.h - l$
\EndFor
\EndFor
\State \textbf{return} $w$
\EndFunction
\end{algorithmic}
\end{algorithm}

 Algorithm~\ref{A:SingleCuts} is enumerating the list of single cuts of a tree. For this, we look first at the angles vanishing at a given level (line 6-8), then we look for the subfactor that contains all these angles such that its adjacent columns contain a larger integer than the maximum of this subfactor (line 9--13). 
 This subfactor contains the locations of the angles that a single cut of the tree will withdraw.
\begin{algorithm}[h]
\small
\caption{\small Return all the single cuts associated to a Schroeder tree.}
\label{A:SingleCuts}
\begin{algorithmic}[1]
\Function{SingleCuts}{$\texttt{t}$}
\State $w\gets$\textsc{PackedWord}(\texttt{t}) \Comment{Indexed from $1$}
\State $sc \gets \emptyset$
\State $m\gets \max(w)$
\For {$l$ from $1$ to $m$}
\State $A\gets w^{-1}[l]$
\State $a_l \gets \min(A)$
\State $a_r \gets \max(A)$
\While {$a_l > 1$ and $w[a_l - 1] < l$}
\State $a_l\gets a_l - 1$
\EndWhile
\While {$a_r < \texttt{t.}n$ and $w[a_r + 1] < l$}
\State $a_r\gets a_r + 1$
\EndWhile
\State $sc \gets sc  \sqcup \{\IEM{a_l}{a_r}\}$
\EndFor
\State \textbf{return} $sc$
\EndFunction
\end{algorithmic}
\end{algorithm}

We perform an exploration of the prunings of a Schroeder tree $t$ using a stack \texttt{stack} which is a container equipped with four operations:
\begin{itemize}
    \item \texttt{stack.empty} testing if the stack is empty;
    \item \texttt{stack.push} putting an element at the top of the stack;
    \item \texttt{stack.top} returning the element at the top of the stack;
    \item \texttt{stack.pop} removing the element at the top of the stack.
\end{itemize}

Algorithm~\ref{A:PruningEnumeration} gives an enumeration of all the prunings. The {\ttfamily Treat} function (line~6) symbolizes a task to perform. To enumerate properly prunings, we begin with the empty cut $\emptyset$, we treat this cut and we add onto {\ttfamily stack} all the single cuts whose minimum cutten angle is smaller than the current treated angle (line 4-9). At the end, we have seen all possible admissible prunings.

	\begin{algorithm}[h]
\small
\caption{\small Prunings enumeration of a Schroeder tree.}
\label{A:PruningEnumeration}
\begin{algorithmic}[1]
\Procedure{PruningEnumeration}{$\texttt{t}$}
\State $sc\gets \textsc{SingleCuts}(\texttt{t})$
\State \texttt{stack.push}($\emptyset$)\Comment{The empty pruning}
\While {not \texttt{stack.empty}()}
\State $p\gets \texttt{stack.pop}()$
\State \texttt{Treat}(\texttt{t}, $p$) \Comment{Perform some task for the pruning $p$ of \texttt{t}}
\For{$c$ in $sc$}
\If {$\max(p) < \min(c)$} \Comment{With convention $\max(\emptyset)=-\infty$}
\State \texttt{stack.push}($p\sqcup \{c\}$)
\EndIf
\EndFor
\EndWhile
\EndProcedure
\end{algorithmic}
\end{algorithm}

 To implement the coproduct, one can just define the {\ttfamily Treat} procedure such that it gets trees generated by the pruning, to get the coproduct $\Delta$ (theorem~\ref{thm:coproduit}).

\section{Library \texttt{FTridend}} \label{sec:coding}

We have designed a \texttt{C++} library, called \texttt{FTridend}, to work with the tridendriform tree algebra.
The source code of this library is available on GitHub \cite{github}.
In this section we give some usefull details of the implementation of \texttt{FTridend}.

\subsection{Schroeder tree}

As introduced in section \ref{ss:schroeder_tree} a Schroeder tree $t$ is represented in a machine by its number of angles $n$, its height $h$ and its initial chain $\varphi_n(t) = (C_0, \ldots C_h)$ where each~$C_i$ is a subset of $\IEM{1}{n}$.

The initial chain $C$ is naturally encoded as the array $[C_0,\ldots,C_h]$.
Now the question, is how to encode the subset $C_i=\{x_1,\ldots,x_\ell\}$ of $\IEM{1}{n}$?
A naive solution consists in use an array $[x_1,\ldots,x_\ell]$ of integers.
As we perform lots of union and intersection of subsets $\IEM{1}{n}$, this approach needs lots of memory accesses, pouring the efficiency of our algorithms.

The solution we have chosen is to use one-hot encoding for implementation of the parts of initial chains. 
For this we have to fix a positive integer $N$ and assume that all Schroeder trees we will consider have at most $N$ angles. 
With this assumptions each part of an initial chain is a subset of $\IEM{1}{N}$ and is characterized by a unique integer $\IEM{0}{2^N-1}$ thanks to the map
\[
\eta_N:\left\lbrace\begin{array}{rcl}
\mathcal{P}(\IEM{1}{N})&\to&\IEM{0}{2^N-1}\\
A & \mapsto & \displaystyle\sum_{i \in A} 2^{i-1}
\end{array}\right.
\]
To fit hardware integers format available on personal computers, typical values of $N$ are $16$, $32$ or~$64$.
Basic logical instructions on integers allow very efficient operations on subsets of $\IEM{1}{N}$ :   
\[
\eta_N(A\cap B) = \eta_N(A)\  \texttt{and}\ \eta_N(B), \quad \eta(A\cup B) = \eta_N(A)\ \texttt{xor}\ \eta_N(B).
\]
More specific machine instructions can be used in this context. For instance, $\texttt{popcount}(\eta_N(A))$ returns the size of the set $A$. 
 
\subsection{Abstract module}

Two abstract vector spaces arise from the tridendriform algebra: $\K\Sch$ and $\K\Sch \otimes \K\Sch$.
Let us focus on the case $\K\Sch$. 
As the typical elements of $\K\Sch$  that we will use contain only few non zero terms, we code elements of $\K\Sch$ with sparse vector.
Hence elements of $\K\Sch$ will be naturally coded using an \texttt{unordered\_map} container. 

Such a container can be seen as a dictionary mapping keys to values.
For our applications, values will be integers.
In order to use an \texttt{unordered\_map} to encode $\K\Sch$ we have to design an hash function on Schroeder tree with values on $64$ bits integers.
This will be done using the $\eta_N$ map defined on the previous section.

Let $t$ be a Schroeder tree with initial chain $C=(C_0,\ldots,C_h)$.
The hash value of $t$ is given
\[
\text{hash}(t)=\sum_{i=0}^h \eta_N(C_i).
\]  
We implement elements of $\K\Sch\otimes \K\Sch$ the same way except that keys are couples of Schroeder tree.
The hash value of $(t_1,t_2)$ is then defined to be $\text{hash}(t_1) + 2^{32}\times \text{hash}(t_2)$.
We multiply the hash of value of $t_2$ by $2^{32}$ to break symmetries.

\subsection{Quasi-shuffle enumeration}

Consider two Schroeder trees $t$ and $s$ and denote by $k$ (respectively $l$) the right tree length of $t$ (respectively the left tree length of $s$), see definition \ref{defi:tree_length}.
Algorithm~\ref{A:AtomicProduct} returns the atomic product of $t$ and $s$ relatively to a given element $\sigma\in\batc(k,l)$.
To compute $t \odot s$ for $\odot\in\{\pleft,\pmid,\pright,\preceq,\succeq,\ast\}$, we have to enumerate specific subsets of $\batc(k,l)$. We define
\begin{align*}
 \batc_{L}(k,l) &\coloneqq \{\sigma \in \batc(k,l),\, \sigma^{-1}(\{1\})=\{1\}\}\\
 \batc_{M}(k,l) &\coloneqq \{\sigma \in \batc(k,l),\, \sigma^{-1}(\{1\})=\{1,k+1\}\}\\
 \batc_{R}(k,l) &\coloneqq \{\sigma \in \batc(k,l),\, \sigma^{-1}(\{1\})=\{k+1\}\}
 \end{align*}
The following table gives subset attached to each product:
\[
\begin{array}{|c|c|}
\hline
 \text{product} & \text{subset of $\batc(k,l)$} \\
 \hline \hline
 t \pleft s & \batc_L(k,l) \\
 \hline
 t \pmid s & \batc_M(k,l) \\
 \hline
t \pright s & \batc_R(k,l) \\
\hline
t \preceq s &  \batc_L(k,l) \sqcup \batc_M(k,l)\\
\hline
t \succeq s &  \batc_M(k,l) \sqcup \batc_R(k,l)\\
\hline
t \ast s &  \batc_L(k,l) \sqcup \batc_M(k,l) \sqcup \batc_R(k,l)\\
\hline
 \end{array}
\]
Thanks to Lemma~\ref{lem:Indution_batc}, we can enumerate sets $\batc_{L}(k,l)$,  $\batc_{M}(k,l)$ and $\batc_{R}(k,l)$ from sets $\batc(k-1,l)$, $\batc(k-1,l-1)$ and $\batc(k,l-1)$ respectively.
Our enumeration of element $\sigma\in\batc(k,l)$ to compute $t\odot l$ is based on the scheme of figure~\ref{fig:scheme}.
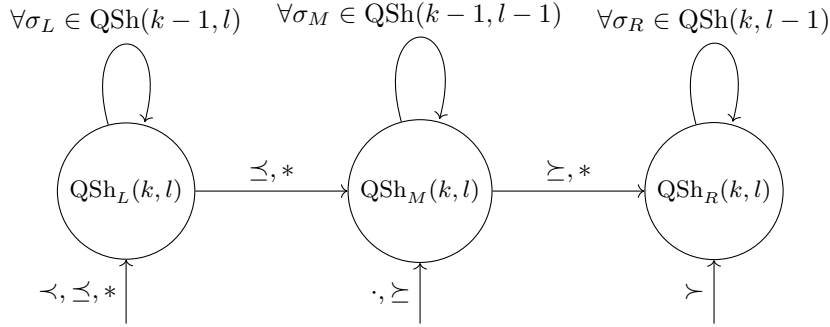
\begin{figure}[h]
    \centering
    \begin{tikzpicture}[x=5pt, y=5pt]
\node[state,minimum size=36pt] (L) at (0,0) {\small$\batc_{L}(k,l)$};
\node[state,minimum size=36pt] (M) at (22,0) {\small$\batc_{M}(k,l)$};
\node[state,minimum size=36pt] (R) at (44,0) {\small$\batc_{R}(k,l)$};
\path[->] (L) edge [loop above] node {$\forall \sigma_L\in\batc(k-1,l)$} ();
\path[->] (M) edge [loop above] node {$\forall \sigma_M\in\batc(k-1,l-1)$} ();
\path[->] (R) edge [loop above] node {$\forall \sigma_R\in\batc(k,l-1)$} ();
\path[->] (L) edge node [above] {$\preceq,\ast$} (M);
\path[->] (M) edge node [above] {$\succeq,\ast$} (R);
\path[->] (0,-10) edge node [left] {$\pleft,\preceq,\ast$} (L);
\path[->] (22,-10) edge node [left] {$\pmid,\succeq$} (M);
\path[->] (44,-10) edge node [left] {$\pright$} (R);
\end{tikzpicture}
    \caption{Scheme to enumerate quasi-shuffles}
    \label{fig:scheme}
\end{figure}
For instance, to compute $s \preceq t$ we follow the arrows labelled $\preceq$ starting with the vertical one.
Elements of $\batc_{L}(k,l)$ are build from elements of $\batc(k-1, l)$.
Once we have consumed all elements of $\batc(k-1, l)$, we go to the next right state and we construct all the elements of $\batc_{M}(k,l)$ from an enumeration of $\batc(k-1, l -1)$.
Once this is done the enumeration of elements of $\batc(k,l)$ used to compute $t \preceq s$ is done.
Summing all the atomic products obtained with Algorithm~\ref{A:AtomicProduct} and using the sparse vector structure defined in the previous section we obtain $t \preceq s \in \K\Sch$.

\subsection{Primitive elements generation}

Finally, we can reach our main objective and encode the $\Omega$ map detailed in theorem~\ref{thm:génération2} (a simplified version of the work~\cite{euleridem}) in order to be able to compute a large amount of primitives elements. Let us remind that the process to compute them is detailled in figure~\ref{fig:algo_idea}.

For this purpose, given any $n\in\NN$ we want to compute $\Prim_{\Coass}(n)$. Suppose we know $\Prim_{\Coass}(k)$ for all $n<k$ and $\Prim_{\Codend}(n)$. To compute a new primitive from what we know, we need to choose an element from the set $W_n$ of equation~\ref{eq:Wn}:
\[
W_n=\Prim_{\Codend}(\K\Sch)_n \oplus \left\langle w \, \middle| \, w\in T^k(\Prim_{\Coass}(\K\Sch))(n) \text{ with } k>1 \right\rangle.
\]
So, an overall description of the algorithm computing all the primitives of degree $n$ is:
\begin{itemize}
	\item first, we generate an enumeration of ordered partition of $n=a_1+\dots +a_l$ of length $l$ , denoted $\mathbf{a}$, for all $l\in\IEM{1}{n}$;
	\item if $\mathbf{a}=(n)$, we apply the identity from $\Prim_{\Codend}(\K\Sch(n))$ into $\Prim_{\Coass}(\K\Sch(n))$;
	\item else according to this partition $\mathbf{a}=(a_1,\dots, a_l)$ with $l\geq 2$, we initialize $l$ generators $e_i$ for each $\Prim_{\Coass}(\K\Sch(n))$ given an enumeration. Then, for each state of $(e_1,\dots,e_l)$ (until each generator comes to its end) it gives a primitive element $\Omega(a_1\otimes a_2 \otimes \ldots \otimes a_l)$;
	\item we put in memory the results in a matrix, fixing an enumeration of \Schw trees of degree $n$. Columns are indexed by this enumeration and the $i$th row is the $i$th result obtained by this algorithm;
	\item Once the enumeration of ordered partitions is finished, we compute a generating family of the image space (this is an ad hoc choice) of this matrix that we can display.
\end{itemize}
Some computations following this process are given in appendix~\ref{appen:A}.

\section{Conclusion}

Thanks to this work, we have made effective the computation of all primitives elements implementing the whole Hopf algebraic structure of the free tridendriform algebra for reasonable degrees. We also spotted a link between \Schw{} trees and some particular packed that might be relevant to study.

The code in {\ttfamily SageMaths} is available online~\cite{github} and will be customized and updated, adding new features like the graded dual of this algebra $\TSym$~\cite{Bergeron_23} and bug corrections. Hence, one can look for a non-inductive formula in order to produce all the primitive elements of $\K\Sch{}$ relying in this work or to perform any experimental test he wants about \Schw{} trees.

In future works, one can mimic such a work to implement the free Rota-Baxter algebra the same way to develop it for the study of arborified multiple zeta values in number theory~\cite{Catoire_2025}.

\appendix

\tikzset{x=0.50cm,y=0.50cm}
\section{Primitives elements}\label{appen:A}

\subsection{A basis of $\Prim_{\Coass}(\K\Sch(1))$}

\begin{enumerate}
	\item
	
\end{enumerate}

\subsection{A basis of $\Prim_{\Coass}(\K\Sch(2))$}

\begin{enumerate}
	\item -\,%
	
\end{enumerate}

\subsection{A basis of $\Prim_{\Coass}(\K\Sch(3))$}\label{appen:sec:prim3}

\begin{enumerate}
	\item -\,%
	
\end{enumerate}

\subsection{A basis of $\Prim_{\Coass}(\K\Sch(4))$} \label{appen:sec:prim4}

\begin{enumerate}
	
	\item %
	
\end{enumerate}

\bibliographystyle{plain}
\bibliography{article.bib}

\end{document}